\documentclass[12pt,letterpaper]{article}

\usepackage{amsmath,amssymb,amsfonts,amsthm,bbm}
\usepackage{mathtools}   

\usepackage{enumitem}   

\usepackage{authblk} 

\usepackage[titletoc,title]{appendix}  

\usepackage{kpfonts}
\usepackage[T1]{fontenc}

\usepackage[table,xcdraw,dvipsnames]{xcolor}   
\usepackage{hyperref}
\newcommand\myshade{85}
\colorlet{mylinkcolor}{YellowOrange}
\colorlet{mycitecolor}{Aquamarine}
\colorlet{myurlcolor}{violet}

\hypersetup{
  linkcolor  = mylinkcolor!\myshade!black,
  citecolor  = mycitecolor!\myshade!black,
  urlcolor   = myurlcolor!\myshade!black,
  colorlinks = true,
}

\usepackage{caption}
\usepackage{subcaption}

\usepackage{arydshln}
\usepackage{multirow}

\usepackage[linesnumbered,ruled,vlined]{algorithm2e}

\SetCommentSty{mycommfont}

\SetKwInput{KwInput}{Input}                
\SetKwInput{KwOutput}{Output}              

\usepackage{listings}             
\renewcommand{\hat}{\widehat}
\renewcommand{\tilde}{\widetilde}

\newcommand{\bfm}[1]{\ensuremath{\boldsymbol{#1}}} 

\def\bzero{\bfm 0} 
 
\def\bbone{\mathbbm{1}} 

\def\ba{\bfm a}   \def\bA{\bfm A}  
\def\bb{\bfm b}   \def\bB{\bfm B}  
   \def\bC{\bfm C}  
     
\def\be{\bfm e}   \def\bE{\bfm E}  
\def\bff{\bfm f}    
\def\bg{\bfm g}     
   \def\bH{\bfm H}  
   \def\bI{\bfm I}

   \def\bM{\bfm M}  
   \def\bN{\bfm N}  
   \def\bO{\bfm O}  
   \def\bP{\bfm P}  
\def\bq{\bfm q}   \def\bQ{\bfm Q}  
     \def\RR{\mathbb{R}}
\def\bs{\bfm s}     
     
\def\bu{\bfm u}   \def\bU{\bfm U}  
\def\bv{\bfm v}   \def\bV{\bfm V}  
     
\def\bx{\bfm x}   \def\bX{\bfm X}  
\def\by{\bfm y}   \def\bY{\bfm Y}  
\def\bz{\bfm z}   \def\bZ{\bfm Z}

\def\calA{{\cal  A}} 
 
\def\calC{{\cal  C}} 
\def\calD{{\cal  D}}

\def\calM{{\cal  M}} 
\def\calN{{\cal  N}}

\def\calR{{\cal  R}} 
\def\calS{{\cal  S}} 
 
\def\calU{{\cal  U}}

\newcommand{\bfsym}[1]{\ensuremath{\boldsymbol{#1}}}

 \def\bepsilon{\bfsym \varepsilon}
              \def\bSigma{\bfsym \Sigma}
         
           \def\bOmega {\bfsym {\Omega}}

 \def\bxi{\bfsym {\xi}}          \def\bXi{\bfsym {\Xi}}

          \def\bPhi{\bfsym {\Phi}}
          \def\bPsi{\bfsym {\Psi}}
 
\providecommand{\abs}[1]{\left\lvert#1\right\rvert}
\providecommand{\norm}[1]{\left\lVert#1\right\rVert}

\providecommand{\angles}[1]{\left\langle #1 \right\rangle}
\providecommand{\paran}[1]{\left( #1 \right)}
\providecommand{\brackets}[1]{\left[ #1 \right]}
\providecommand{\braces}[1]{\left\{ #1 \right\}}

\providecommand{\defeq}{\triangleq}

\usepackage{mathtools}
\DeclarePairedDelimiterX{\infdivx}[2]{(}{)}{%
  #1 \; \delimsize\| \; #2%
}

\DeclareMathOperator{\Cov}{Cov}

\newcommand{\E}[1]{{\mathbb{E}} \left[ #1 \right]} 

\DeclareMathOperator{\rank}{rank}

\DeclareMathOperator{\Var}{Var}

\newcommand{\vect}[1]{{\textsc{Vec}} \left( #1 \right)}
\newcommand{\mat}[1]{{\textsc{Mat}} \left( #1 \right)}

\newtheorem{definition}{Definition}[section]
\newtheorem{assumption}[definition]{Assumption}
\newtheorem{lemma}[definition]{Lemma}

\newtheorem{theorem}[definition]{Theorem}

\newcommand{\smlo}[1]{{\rm o} \left( #1 \right) }

\newcommand{\Op}[1]{{\mathcal{O}_p} \left( #1 \right) }

\definecolor{royalpurple}{rgb}{0.47, 0.32, 0.66}

\def\beq{\begin{equation}}
\def\eeq{\end{equation}}
 
\usepackage[framemethod=TikZ]{mdframed} 

\usepackage{fancybox}

\usepackage[top=1in, bottom=1in, left=1in, right=1in]{geometry}

\usepackage{setspace}
\usepackage[authoryear]{natbib}  
\newcommand{\mybibsty}{chicago}
\newcommand{\mybib}{0-main}

\usepackage{graphicx}
\graphicspath{{figs/}}

\begin{document}
%
%

\newcommand{\blind}{0}

\if0\blind
{
	\title{Modeling Multivariate Spatial-Temporal Data with \\ Latent Low-Dimensional Dynamics}
    
    \author[1]{Elynn Y. Chen 
    \thanks{Supported in part by NSF Grants DMS-1803241.}}
    \author[2]{Xin Yun \thanks{Equal contribution.}
    }
    \author[3]{Rong Chen \thanks{Supported in part by NSF Grants DMS-1503409, DMS-1737857, DMS-1803241 and IIS-1741390.}}
    \author[4]{Qiwei Yao \thanks{Corresponding author. Email: \url{Q.Yao@lse.ac.uk}}}
    
    \affil[2]{University of California, Berkeley}
    \affil[1]{School of Data Science, Fudan University}
    \affil[3]{Rutgers University}
    \affil[4]{London School of Economics and Political Science}
	\date{\today}
	\maketitle
} \fi

\if1\blind
{
	\title{Modeling Multivariate Spatial-Temporal Data with \\ Latent Low-Dimensional Dynamics}
	\author[1]{}
	\date{\vspace{-1ex}}
	\maketitle
	\vspace{-10ex}
} \fi

\begin{abstract}
High-dimensional multivariate spatial-temporal data arise frequently in a wide range of applications; however, there are relatively few statistical methods that can simultaneously deal with spatial, temporal and variable-wise dependencies in large data sets. 
In this paper, we propose a new approach to utilize the correlations in variable, space and time to achieve dimension reduction and to facilitate spatial/temporal predictions in the high-dimensional settings.
The multivariate spatial-temporal process is represented as a linear transformation of a lower-dimensional latent factor process. 
The spatial dependence structure of the factor process is further represented non-parametrically in terms of latent empirical orthogonal functions.
The low-dimensional structure is completely unknown in our setting and is learned entirely from data collected irregularly over space but regularly over time. 
We propose innovative estimation and prediction methods based on the latent low-rank structures. 
Asymptotic properties of the estimators and predictors are established.
Extensive experiments on synthetic and real data sets show that, while the dimensions are reduced significantly, the spatial, temporal and variable-wise covariance structures are largely preserved.
The efficacy of our method is further confirmed by the prediction performances on both synthetic and real data sets. 

\vspace{3em}

\end{abstract}

%
%

\section{Introduction}  \label{sec:intro}


The increasing availability of multivariate data collected over geographic regions and time in various applications has created unique opportunities and challenges for practitioners seeking to capitalize on its full utility. 
For example, United States Environmental Protection Agency publishes daily, from more than 20,000 monitoring stations, a collection of environmental and meteorological measurements such as temperature, pressure, wind speed/direction and levels of various pollutants. 
Such data naturally constitute a tensor (multi-dimensional array) with three modes (dimensions) representing $n$ spacial locations, $T$ time points and $p$ variables, respectively. 
Since physical processes rarely occur in isolation but rather influence and interact with one another, simultaneously modeling the dependencies among different variables, regions, and time points is of great potential to reduce dimensions, produce more accurate estimation/prediction and further provide a deeper understanding of real world phenomena. 
At the same time, methodological issues arise because these data exhibit complex multivariate spatial-temporal co-variances that may involve potential dependencies between spatial locations, time points and different processes. 

Traditionally, researchers mainly restrict their analysis to only two dimensions while fixing the third: 
time series analysis applied to a slice of such data at one location focuses on temporal modeling and prediction \citep{tsay2018nonlinear,box2015time,tsay2014multivariate,brockwell2013time,fanyao2005nonlinear}; 
spatial statistical models for a slice of such data at one time point address spatial dependence and prediction over unobserved locations \citep{cressie2015statisticsS}; 
and uni-variate spatial-temporal statistics concentrate on only one variable observed over space and time \citep{huang1996spatio,cressie2015statisticsST,lopes2008spatial,cressie2008fixed}. 

In this paper, we propose a new class of multivariate spatial-temporal models that characterize spatial, temporal and variable dependence simultaneously. 
This is made possible by an innovative combination of multivariate factor models \citep{fan2018robust, chang2015high, lam2012factor, lam2011estimation, bai2003inferential, bai2002determining} and the method of latent empirical orthogonal functions \citep{monahan2009empirical,hannachi2007empirical,von2001statistical,wilks1995statistical}.
Specifically, the $p$-dimensional spatial-temporal process is represented as a linear combination of a $r$-dimensional latent common factor process ($r\ll p$), which captures the correlations among $p$ variables.
The factor spatial-temporal processes are further represented in terms of latent empirical orthogonal functions (EOFs), which captures the spatial dependencies. 
As we shall see later, the EOFs in our setting have a close relationship with the loading matrix in factor analysis. 
We refer to the EOFs in our setting as the spatial loading functions. 
The coefficients of spatial loading functions are time-varying random variables and thus capture the temporal dependence. 
We provide a detailed analysis of the covariance structure of the proposed model across variables, space and time in Section \ref{sec:cov-spacetimevariable}. 
It shows that the proposed model is a generalization of several low-rank models in the literature \citep{higdon2002space,wikle1999dimension,kammann2003geoadditive,cressie2010fixed,banerjee2008gaussian,finley2009improving,tzeng2018resolution}.
In addition, the low-dimensional structure and the the spatial loading functions are completely unknown in our setting and is learned entirely from data collected irregularly over space but regularly over time. 

The estimation builds upon the idea in \cite{wang2019factor} and further incorporates non-parametric estimation for the spatial loading functions.
Particularly, we assembled the observations from $n$ discrete spatial locations as a time series of $n \times p$ matrices whose rows and columns correspond to $n$ sampling sites and $p$ variables, respectively.
As a result, the model on the discrete sampling locations can be reformulated in a similar form as the matrix factor model and it is estimated with a variant procedure based on the whiteness of spacial nugget effects. 
We also proposed prediction method for new locations and time points. 
Thanks to the innovative combination of reduced-rank models of two aspects, our method is able to efficiently handle multivariate spatial-temporal data sets with large $n$ (space), $p$ (variable) and $T$ (time points). 

\subsection{Related works} \label{sec:related}

To overcome the computational burden with large spatial or spatial-temporal data sets, researchers have developed reduced-rank approximations for univariate processes.
\cite{higdon2002space} uses kernel convolution, \cite{wikle1999dimension,kammann2003geoadditive,cressie2008fixed} successfully reduces the computational cost of kriging by using a flexible family of non-stationary covariance functions constructed from low rank basis functions.
\cite{banerjee2008gaussian} and \cite{finley2009improving} uses predictive process, and \cite{tzeng2018resolution} uses thin-plate splines.
See also reviews of low-rank representations for spatial processes in \cite{wikle2010low-rank,cressie2015statisticsS,cressie2015statisticsST}. 
Our method applies to multivariate processes and incorporates two aspects of dimension reductions. 
The first aspect is the variable-wise dimension reduction where the observed $p$-dimensional process is represented as a linear combination of $r$-dimensional latent factor process. 
Further, the latent factor process assumes a reduced-rank representation whose formulation is similar to the aforementioned reduced rank approximation methods.
However, the spatial loading functions is completely unknown.
Moreover, we don't impose any distributional assumptions on the underlying process, nor any parametric forms on its covariance function.


For multivariate spatial data, \cite{cook1994some} introduced the concept of a spatially shifted factor and a single-factor shifted-lag model and \cite{majure1997dynamic} discussed graphical methods for identifying shifts. 
Following the ideas of multiple-lag dynamic factor models that generalize static factor models in the time series setting, \cite{christensen2001generalized, christensen2002latent, christensen2003modeling} extended the shifted-lag model to a generalized shifted-factor model by adding multiple shifted-lags and developed a systematic statistical estimation, inference, and prediction procedure. 
However, they do not include the time dimension and their method is an analogy of the multiple-lag dynamic factor models applied in the spatial setting.
Thus, their definition of factors is very different from ours.
Moreover, the assumption that spatial processes are second-order stationary is required for the moment-based estimation procedure and the theoretical development.

Various multivariate spatial-temporal conditional auto-regressive models have also been proposed by \cite{carlin2003hierarchical, congdon2004multivariate,pettitt2002conditional,zhu2005generalized,daniels2006conditionally,tzala2008bayesian}, among others. 
Most of these papers, however, focus on empirical applications and do not offer any theoretical guarantees.
Also, their estimation methods necessitate assumptions on the distribution of the observations. 
\cite{bradley2015multivariate} introduced a multivariate spatial-temporal mixed effects model to analyze high-dimensional multivariate data sets that vary over different geographic regions and time points. 
They adopt a reduced rank spatial structure \citep{wikle2010low-rank} and model temporal behavior via vector auto-regressive components. 
Their method only applies to low-dimensional multivariate observations because they model each variable separately. 
The cross-dependence structures of multiple processes are modeled jointly by \cite{genton2015cross,bourotte2016flexible}.
These approaches impose separability and various independence assumptions, which are not appropriate for many settings, as these models fails to capture important interactions and dependencies between different variables, regions, and times \citep{stein2005space}. 
In addition, they assume the random effect term is common across all processes which is unrealistic especially in the case with a large number of variables.
Our method can effectively deal with data sets with large $n$, $p$, and $T$ by simultaneously modeling the variable-wise and spatial low-rankness. 
Besides, our modeling of the spatial dependence though latent factor processes is different from the aforementioned methods in that we impose no assumptions about the stationarity over space, nor the distribution of data, nor any restrictive form of spatial covariance functions. 

\subsection{Contribution}

We propose a new class of models for large-scale multivariate spatial-temporal processes. 
The model characterizes spatial, temporal and variable-wise dependencies simultaneously.
The spatial dimension $n$, the variable dimension $p$ and the time dimension $T$ can be very large at the same time. 
To our best knowledge, our model is the first to deal with spatial, temporal and variable-wise covariance simultaneously, while allowing large $n$, $p$ and $T$. 
It provides a flexible and rich cross-covariance structure for these dimensions simultaneously.

We develop efficient estimation and prediction procedures and establish theoretical properties of the estimators and predictors. 
The estimation procedure is based on a novel reformulation of the discrete observations of the $p$-dimensional spatial-temporal process. 
We believe this formulation is quite general and flexible to be extended to enable more sophisticated analysis along space, time or variable dimensions.  

\subsection{Notation and Organization}

When $\bA$ is a square matrix, we denote by $tr(\bA)$, $\lambda_{max}(\bA)$ and $\lambda_{max}(\bA)$ the trace, maximum and minimum eigenvalues of the matrix $\bA$, respectively. 
We use $\norm{\bA}_2 \defeq \sqrt{\lambda_{max}(\bA'\bA)}$ and $\norm{\bA}_F \defeq \sqrt{tr(\bA'\bA)}$ to denote the spectral and Frobenius norms of the matrix $\bA$, respectively. $\norm{\bA}_{min}$ denotes the positive square root of the minimal eigenvalue of $\bA'\bA$ or $\bA\bA'$, whichever is a smaller matrix. 
For two sequences $a_N$ and $b_N$, we write $a_N \asymp b_N$ if $a_N = O(b_N)$ and $b_N = O(a_N)$. 

The remainder of the article is outlined as follows. Section 2 introduces the model settings. Section 3 discusses estimation procedures for loading matrix and loading functions. Section 4 discuss the procedures for kriging and forecasting over space and time, respectively. Section 5 presents the asymptotic properties of the estimators. Section 6 illustrates the proposed model and estimation scheme on a synthetic data set; and finally Section 7 applies the proposed method to a real data set. Technique proofs are relegated to the Appendix.

\section{Model} \label{sec:model}

Consider a multivariate spatial-temporal process $\tilde\by_t(\bs) \in \mathbb{R}^p$:
\begin{equation} \label{eqn:stvp-1}
\tilde\by_t(\bs) = \bC^\top(\bs) \bz_t(\bs) + \bxi_t(\bs) + \bepsilon_t(\bs), \quad t = 0, \pm 1, \pm 2, \cdots, \; \bs \in \calS \subset \calR^2.
\end{equation}
The first mean process term with observable covariates $\bz_t(\bs)\in\RR^m$ and unknown coefficient matrix $\bC(\bs)\in\RR^{m \times p}$ captures the large-scale correlations.
The second term $\bxi_t(\bs)\in\RR^p$ is the zero-mean latent spatial-temporal vector process that captures the medium or small-scale correlation structure. It satisfies the conditions
\begin{equation} \label{eqn:xi-cond}
\E{\bxi_t(\bs)} = \bzero, \quad \Cov\brackets{\bxi_{t_{1}}(\bu), \bxi_{t_{2}}(\bv)} = \bSigma_{ \xi, \abs{t_{1} - t_{2}} }(\bu, \bv).
\end{equation}
The additive error vector $\bepsilon_t(\bs)$ is the unknown spatial nugget effects which are spatially uncorrelated but are allowed to be temporally correlated.
It is also uncorrelated with the signal process. 
That is, 
\begin{eqnarray}
& \E{\bepsilon_t(\bs)} = \bzero, \quad \Var\brackets{\bepsilon_t(\bs)} = \bSigma_\epsilon(\bs), \quad \Cov\brackets{\bepsilon_{t_1}(\bu), \bepsilon_{t_2}(\bv)} = \bzero \; \forall \; t_1, t_2, \;\bu \ne \bv, \label{eqn:eps-cond} \\
& \Cov\brackets{\bxi_{t_{1}}(\bu), \bepsilon_{t_2}(\bv)} = \bzero \; \forall \; t_1, t_2, \bu, \bv.
\end{eqnarray}
Given the observable covariates $\bz_t(\bs)\in\RR^m$, the coefficients $\bC(\bs)$ can be calculated by least square regression. 
To make the main idea clear, we focus on the zero-mean process $\by_t(\bs) = \tilde\by_t(\bs) - \bC^\top(\bs) \bz_t(\bs)$  with out loss of generality. 
That is,
\begin{equation} \label{eqn:stvp}
\by_t(\bs) = \bxi_t(\bs) + \bepsilon_t(\bs), \quad t = 0, \pm 1, \pm 2, \cdots, \; \bs \in \calS \subset \calR^2.
\end{equation}
Under the condition \eqref{eqn:xi-cond} and \eqref{eqn:eps-cond}, $\by_t(\bs)$ is second-order stationary in time $t$.
We have $\E{\by_t(\bs)}=\bzero$ and 
\begin{equation*}
\Cov\brackets{\by_{t_1}(\bu), \by_{t_2}(\bv)} =  \bSigma_{\xi, \abs{t_1-t_2}}(\bu, \bv) +\bSigma_{\epsilon, \abs{t_1 - t_2}}(\bu) \cdot \bbone\paran{\bu=\bv},
\end{equation*}
where the covariance $\bSigma_{\xi, t}(\bu, \bv)$ is assumed to be continuous in $\bu$ and $\bv$.

Model \eqref{eqn:stvp} does not impose any stationary conditions over space.
However, it requires that $\by_t(\bs)$ is second order stationary in time $t$ to enable the learning of the dependence across different locations and times.
In practice the data often show some trends and seasonal patterns in time. The existing de-trend and de-seasonality methods in time series analysis \citep{tsay2018nonlinear,tsay2014multivariate,fanyao2005nonlinear} can be applied to make each time series temporally stationary, including the inclusion of time trends in the mean term $\bC'(\bs)\bz_t(\bs)$.

\subsection{The covariance structures across variables, space and time}
\label{sec:cov-spacetimevariable}

To capture the correlation between the multiple processes, we assume that the latent spatial-temporal vector process are driven by a lower-dimension latent spatial-temporal factor process linearly in the form:
\begin{equation} \label{eqn:xi_fac}
\bxi_t(\bs) = \bB \bff_t(\bs),
\end{equation}
where $\bff_t(\bs) \in \mathbb{R}^r$ is the latent factor process ($r \ll p$) and $\bB$ is the $p \times r$ loading matrix that characterized the correlation between multiple processes. Equation (\ref{eqn:xi_fac}) is a generalization of the widely-used statistical factor models for high-dimensional data sets \citep{fan2018robust, chang2015high, lam2012factor, lam2011estimation, bai2003inferential, bai2002determining} to the spatial-temporal process.

To capture the spatial temporal correlations, we further assume a finite dimensional representation for $\bff_t(\bs)$, that is, the latent $r \times 1$ factor process $\bff_t(\bs)$ admits a finite functional structure,
\begin{equation} \label{eqn:fac_finstr}
\bff_t(\bs) = \sum_{j=1}^{d} a_j(\bs) \bx_{tj} ,
\end{equation}
where $a_j(\bs)$, $j \in [d]$ are deterministic and linearly independent functions (i.e. none of them can be written as a linear combination of the others) in the Hilbert space $L_2(\calS)$, and random vector $\bx_{tj} \in \mathbb{R}^r$. Equation (\ref{eqn:fac_finstr}) models the latent factor process as the linear combination of random vectors with weight $a_j(\bs)$.

Functions $a_1(\cdot), \cdots, a_d(\cdot)$ are not uniquely defined by (\ref{eqn:fac_finstr}) even with known $\bff_t$. 
Particularly, we can rewrite $\bff_t(\bs) = \sum_{j=1}^{d} a^*_j(\bs) \bx^*_{tj}$ where $a^*_j(\bs) = c a_j(\bs)$ and $\bx^*_{tj} = c^{-1} \bx_{tj}$ for any scalar $c \ne 0$.
There is no loss of generality in assuming that $a_1(\cdot), \cdots, a_d(\cdot)$ are orthonormal in the sense that
\[
\angles{a_i, a_j} = \bbone(i=j),
\]
as any set of linear independent functions in a Hilbert space can be standardized to this effect.
The above identification condition is defined on the whole space. 
We will elaborate more on the model identification in the next section. 
Combining \eqref{eqn:xi_fac} and \eqref{eqn:fac_finstr}, we have
\begin{equation} \label{eqn:xi_fac_finstr}
\bxi_t(\bs) = \bB \sum_{j=1}^{d} a_j(\bs) \bx_{tj} = \bB \bX'_t \ba(\bs),
\end{equation}
where $\bX_t=\left( \bx_{t1}, \cdots, \bx_{td} \right)'$ and $\ba(\bs) = \left( a_1(\bs), \cdots, a_d(\bs) \right)'$.
Therefore, the latent spatial-temporal covariance of vector process $\bxi_{t_1}(\bu)$ and $\bxi_{t_2}(\bv)$ can be written as
\begin{equation} \label{eqn:spatial_cov}
\bSigma_{ \xi, \abs{t_{1} - t_{2}} }(\bu, \bv)=\Cov\brackets{\bB \bX'_{t_1} \ba(\bu), \bB \bX'_{t_2} \ba(\bv)}=\bB \bSigma_{f,|t_1-t_2|}(\bu,\bv)\bB',
\end{equation}
where 
\begin{equation}  \label{eqn:fac_cov}
\bSigma_{f,|t_1-t_2|}(\bu,\bv)=\sum_{i=1}^d\sum_{j=1}^d a_i(\bu)a_j(\bv)\bSigma_{x,ij,|t_1-t_2|},
\end{equation}
and $\bSigma_{x,ij,\abs{t_1-t_2}}=\Cov\brackets{\bx_{t_1i}, \bx_{t_2j}} \in \RR^{r\times r}$.
Equation (\ref{eqn:spatial_cov}) captures the spatial-temporal dependence structure via the finite dimensional representation of latent factors in (\ref{eqn:fac_finstr}). Specifically, the covariance of factor $\bSigma_{f,\abs{t_1-t_2}}$ is the linear combination of $\bSigma_{x,ij,\abs{t_1-t_2}}$, which captures the time-dependence structure between $t_1$ and $t_2$. The weight $a_i(\bu)a_j(\bv)$ captures the spatial dependence between location $\bu$ and $\bv$. 

\paragraph{Relation to the univariate reduced-rank models.}
In the special case where $\bff_t(\bs)$ is a scalar, i.e. $r=1$, the covariance of latent factor assumes the following structure
\begin{equation}   \label{eqn:fac_cov-1}
\sigma_{f,\abs{t_1-t_2}}(\bu,\bv)=\ba(\bu)^\top\bSigma_{x,\abs{t_1-t_2}}\ba(\bv), 
\end{equation}
where $\bSigma_{x,\abs{t_1-t_2}}$ is a $d\times d$ matrix consisting of $\bSigma_{x,ij,\abs{t_1-t_2}}$ (which is a scalar when $r=1$) for all $i,j\in[d]$.
Spatial-temporal structure \eqref{eqn:fac_cov-1} corresponds to the low-rank empirical orthogonal function method in the literature of univariate geostatistics \citep{wikle1999dimension,kammann2003geoadditive,cressie2010fixed,banerjee2008gaussian,finley2009improving,tzeng2018resolution}.

\paragraph{Relation to the multivariate reduced-rank models.}
In the case of known low-dimensional factor process $\bff_t(\bs)$, the covariance of any pair of variables in $\bff$ assumes the structure in \eqref{eqn:fac_cov-1}. 
This corresponds to the low-rank approximation in the literature of multivariate geostatistics.  
In our setting, the latent factor process $\bff_t(\bs)$ is unknown and needs to be estimated from an observed high-dimensional process $\by_t(\bs)$.

\subsection{Discrete sample observations}
Since we only observe discrete observations, we assume that we have a $n \times p$ matrix $\bXi_t \defeq \brackets{ \bxi_t(\bs_1), \cdots, \bxi_t(\bs_n) }^\top$ where $\bxi_t(\bs_i)\in\RR^p$ consists of values of $\bxi_t(\bs)$ from the $i$-th sampling location. 
It follows from (\ref{eqn:xi_fac_finstr}) that
\begin{equation} \label{eqn:xi_fac_finstr_mat}
\bXi_t = \bA \bX_t \bB',
\end{equation}
where  $\bA = [a_j(\bs_i)]_{ij}$, $i \in [n]$ and $j \in [d]$. We are interested in estimating the loading matrix $\bB$, random matrix $\bX_t$, the spatial loading function matrix $\bA$, and the spatial loading functions $a_j(\bs)$ for $j\in[d]$.

Matrices $\bA$ and $\bB$ are not uniquely defined by \eqref{eqn:xi_fac}.
Specifically, we can rewrite $\bXi_t = \bA^* \bX^*_t \bB^{*'}$ where $\bA^* = \bA \bO_1$, $\bB^* = \bB \bO_2$, and $\bX^*_t = \bO_1^{-1} \bX_t \bO_2^{-1}$ for any invertible matrices $\bO_1$ and $\bO_2$.
To address this identification problem, we assume that columns of $\bA$ ($\bB$) are orthogonal.

Under the orthogonal assumption, the vector space spanned by the columns of $\bA(\bs)$ and $\bB$, denoted as $\calM(\bA(\bs))$ and $\calM(\bB)$, are uniquely defined.
In this article, we estimate matrix representations $\bQ_A$ and $\bQ_B$ of $\calM(\bA(\bs))$, $\calM(\bB)$ instead of $\bA(\bs)$ and $\bB$ under the assumption that
\begin{equation}  \label{eqn:Q-constraint}
\bQ_A' \bQ_A = \bI_d, \quad \text{and} \quad \bQ_B' \bQ_B = \bI_r,
\end{equation}
and the corresponding $\bZ_t$ such that (\ref{eqn:xi_fac_finstr_mat}) can be rewritten as
\begin{equation}\label{eqn:xi_fac_finstr_mat-1}
\bXi_t = \bA \bX_t \bB' = \bQ_A \bZ_t \bQ_B'.
\end{equation}
Given $\bQ_A$, the kernel reproducing Hilbert space (KRHS) spanned by $a_1(\cdot), \cdots, a_d(\cdot)$ is also uniquely defined and we estimate a set of representative functions $q_{a,1}(\cdot), \cdots, a_{q,d}(\cdot)$.
Therefore, the estimation of $\bA$, $\bB$, and $\bX_t$ in the multivariate spatial-temporal model can be converted to the estimation of $\bQ_A$, $\bQ_B$, and $\bZ_t$.
Further we use the estimators to estimate the latent spatial-temporal covariance and make spatial-temporal predictions for large scale multi-variate spatial temporal data set.
More details are discussed in the sequel.

\section{Estimation} \label{sec:est}

Let $\braces{ \tilde\by_t(\bs_i), \bz_t(\bs_i) }$, $i \in [n]$, $t \in [T]$ be the available observations over space and time, where $\tilde\by_t(\bs_i) \in \mathbb{R}^p$ and $\bz_t(\bs_i) \in \mathbb{R}^m$ is a vector of covariates  observed at location $\bs_i$ at time $t$. In this article, we restrict attention to the case where all variables have been measured at the same sample locations $\bs_i$, $i \in [n]$.

In general cases where $\bC(\bs) \ne \bzero$, we can estimate $\hat \bC(\bs)$ by least square regression from the observations $\braces{ \tilde\by_t(\bs_i), \bz_t(\bs_i) }$.
The following procedure can be applied to the residuals $\hat\by_t(\bs_i) \defeq \tilde\by_t(\bs_i) - \hat \bC^\top(\bs_i) \bz_t(\bs_i)$.
With out loss of generality, we consider a special case where $\bC(\bs) \equiv \bzero$ in \eqref{eqn:stvp}.
Now the observations are generated from the process
\begin{equation} \label{eqn:stvp_c0}
\by_t(\bs) = \bxi_t(\bs) + \bepsilon_t(\bs) = \bB \bX^\top_t \ba(\bs) + \bepsilon_t(\bs).
\end{equation}
From \eqref{eqn:xi_fac}, \eqref{eqn:fac_finstr}, and \eqref{eqn:xi_fac_finstr_mat}, we stack $\by_t(\bs_i)$, $i \in [n]$ together as rows and get
\begin{equation} \label{eqn:stvp_c0_mat}
\bY_t =  \bXi_t + \bE_t  = \bA \bX_t \bB^\top + \bE_t = \bQ_A \bZ_t \bQ_B^\top + \bE_t,
\end{equation}
where $\bY_t = \left( \by_t(\bs_1), \cdots, \by_t(\bs_n) \right)$ and $\bE_t = \left( \bepsilon_t(\bs_1), \cdots, \bepsilon_t(\bs_n) \right)^\top$.

Note that $\bA$ (or $\bB$) has the same column space as $\bQ_A$ (or $\bQ_B$).
They are different only up to a scalar factor or a rotation such that $\bA$ satisfies Condition \ref{cond:A_factor_strength} in Section \ref{sec:theory} while $\bQ_A$ satisfies $\bQ_A^\top \bQ_A = \bI_d$, and $\bB$ satisfies Condition \ref{cond:B_factor_strength} while $\bQ_B$ satisfies $\bQ_B^\top \bQ_B = \bI_r$.
In the following, we use the triplets $\paran{\bQ_A, \bZ_t, \bQ_B}$ and $\paran{\bA, \bX_t, \bB}$ interchangeably.

\subsection{Partitioned spatial loading spaces \texorpdfstring{$\calM (\bA_1)$}{MA1} and \texorpdfstring{$\calM (\bA_2)$}{MA2}} \label{subsec:est:A1A2}

Note that the nugget effect $\bepsilon_t(\bs)$ are uncorrelated over space.
We exploit this fact to exclude the covariance term incurred by the nugget effect.
Particularly, we divide $n$ locations $\bs_1, \ldots, \bs_n$ into two sets $\calS_1$ and $\calS_2$ with $n_1$ and $n_2$ elements respectively.
Preferably, we set $n_1 \asymp n_2 \asymp n/2$ according to Theorem \ref{thm:A_sep_err_bnd}.
Let $\bY_{lt}$ be a matrix consisting of $\by_t(\bs)$, $\bs \in \calS_l$, $l=1,2$ as rows. Then $\bY_{1t}$ and $\bY_{2t}$ are two matrices of dimension $n_1 \times p$ and $n_2 \times p$ respectively. It follows from \eqref{eqn:stvp_c0} that
\begin{equation} \label{eqn:stvp_mat_div}
\bY_{1t} =  \bXi_{1t} + \bE_{1t}  = \bA_1 \bX_t \bB^\top + \bE_{1t}, \qquad \bY_{2t} =  \bXi_{2t} + \bE_{2t}  = \bA_2 \bX_t \bB^\top + \bE_{2t},
\end{equation}
where $\bA_l$ is a $n_l \times d$ matrix, its rows are $\left( a_1(\bs), \ldots, a_d(\bs) \right)$ at different locations $\bs \in \calS_l$ and $\bE_{t,l}$ consists of $\bepsilon_t(\bs)$ as rows with $\bs \in \calS_l$, $l=1,2$.

For model identification, we assume that the columns of $\bA_l$, $l=1,2$ are orthogonal.
Under this assumption, $\calM(\bA_1)$ and $\calM(\bA_2)$, which are the column spaces of $\bA_1$ and $\bA_2$, are uniquely defined.
This however implies that $\bX_t$ in the second equation in (\ref{eqn:stvp_mat_div}) will be different from that in the first equation. Thus, we may rewrite (\ref{eqn:stvp_mat_div}) as
\begin{equation} \label{eqn:stvp_mat_div_1}
\bY_{1t} =  \bXi_{1t} + \bE_{1t}  = \bA_1 \bX_t \bB^\top + \bE_{1t}, \qquad \bY_{2t} =  \bXi_{2t} + \bE_{2t}  = \bA_2 \bX^*_t \bB^\top + \bE_{2t},
\end{equation}
where $\bX^*_t = \bO \bX_t$ and $\bO$ is an invertible $d \times d$ matrix.

Let $\by_{lt, \cdot j}$, $\be_{lt, \cdot j}$, and $\bb_{j \cdot}$ be the $j$-th column of $\bY_{lt}$,  $\bE_{lt}$, and $\bB$, $l=1,2$, $j \in [p]$, respectively.
Define spatial-cross-covariance matrix between the $i$-th and $j$-th variables as
\begin{equation} \label{eqn:Omega_ij}
\bOmega_{A,ij} = \Cov\brackets{\by_{1t, \cdot i}, \by_{2t, \cdot j}} = \bA_1 \Cov\brackets{ \bX_t \bb'_{i \cdot}, \bX^*_t \bb'_{j \cdot} } \bA_2.
\end{equation}
The covariance related to $\be_{1t,\cdot i}$ and $\be_{2t,\cdot j }$ are all zeros because they are spatial white noises and also uncorrelated with the signals.
When $d \ll n$, it is reasonable to assume that $\rank\brackets{\bOmega_{A,ij}} = d$.
Define
\begin{equation}
\bM_{A_1} = \sum_{i=1}^{p} \sum_{j=1}^{p} \bOmega_{A,ij} \bOmega^\top_{A,ij}, 
\quad\text{and}\quad
\bM_{A_2} = \sum_{i=1}^{p} \sum_{j=1}^{p} \bOmega^\top_{A,ij} \bOmega_{A,ij} \nonumber
\end{equation}

$\bM_{A_1}$ and $\bM_{A_2}$ share the same $d$ positive eigenvalues and $\bM_{A_l} \bq = \bzero$ for any vector $\bq$ perpendicular to $\calM(\bA_l)$, $l=1,2$. Therefore, the columns of a matrix representation of $\calM(\bA_l)$, $l=1,2$, can be estimated as the $d$ orthonormal eigenvectors of matrix $\bM_{A_l}$ corresponding to largest $d$ positive eigenvalues in the descending order.

Now we define the sample version of these quantities and introduce the estimation procedure. Suppose we have centered our observations $\bY_{1t}$ and $\bY_{2t}$, let $\hat{\bOmega}_{A,ij}$ be the sample cross-space covariance of $i$-th and $j$-th variables and $\hat{\bM}_{A_l}$ be the sample version of $\bM_{A_l}$, $l=1,2$, that is
\begin{equation} \label{eqn:M_hat}
\hat{\bOmega}_{A,ij} = \frac{1}{T} \sum_{t=1}^{T} \bY_{1t, \cdot i} \bY^\top_{2t, \cdot j}, \quad \hat{\bM}_{A_1} = \sum_{i=1}^{p} \sum_{j=1}^{p} \hat{\bOmega}_{A,ij} \hat{\bOmega}^\top_{A,ij}, \quad \hat{\bM}_{A_2} = \sum_{i=1}^{p} \sum_{j=1}^{p} \hat{\bOmega}^\top_{A,ij} \hat{\bOmega}_{A,ij}.
\end{equation}
A natural estimator for a matrix representation of $\calM(\bA_l)$, under the constraint that $\bQ_{A,l}^\top \bQ_{A,l} = \bI_d$, is defined as
\begin{equation} \label{eqn:Q_A_l}
\hat{\bQ}_{A,l} = \{ \hat{\bq}_{A, l1}, \cdots, \hat{\bq}_{A, ld} \}, \quad l = 1, 2,
\end{equation}
where $\hat{\bq}_{A, lj}$ is the eigenvector of $\hat{\bM}_{A_l}$ corresponding to its $j$-th largest eigenvalue.
Matrix $\hat{\bQ}_{A,l}$ estimates $\bA_{l}(\bs)$ up to a scalar factor while sharing the same column space.
However such an estimator ignores the fact that $\bxi_t(\bs)$ is continuous over the set $\calS$.
Section \ref{sec:re-estimate-A} estimates a refined spatial loading matrix $\hat{\bQ}_A$ and further estimates the loading function $\hat{\bQ}_A(\bs)$, which estimates $A(\bs)$ up to a scalar factor.

\subsection{Variable loading space \texorpdfstring{$\calM \paran{\bB}$}{MB}} \label{subsec:est:B}

To estimate the $p \times r$ variable loading matrix $\bB$, we again utilize the spatial whiteness properties of the nugget effect.
Recall that in Section \ref{subsec:est:A1A2}, the entire set of $n$ sampled locations are divided into two sets $\calS_1$ and $\calS_2$ of size $n_1$ and $n_2$, where $n_1 \asymp n_2 \asymp \frac{n}{2}$.
We keep only $m = \lfloor\frac{n}{2}\rfloor$ in each of $\calS_1$ and $\calS_2$ to calculate $\bB$.
When $n$ is even, we make use of all sampled locations, while when $n$ is odd, one of the sampled locations is dropped randomly.

We reuse the notation in equation \eqref{eqn:stvp_mat_div} for the observations in  $\calS_1$ and $\calS_2$ and rewrite it as ($\ref{eqn:stvp_mat_div_1}$) for model identification, except for now $\bY_{1t}$ and $\bY_{2t}$ are two matrices of same dimension $m\times p$.
Let $\by_{lt,i\cdot}$, $\be_{lt,i\cdot}$, and $\ba_{l,i\cdot}$ be the $i$-th row of $\bY_{lt}$, $\bE_{lt}$, and $\bA_l$, $l = 1,2$, respectively. Define the covariance matrix of $p$ variables sampled at the $i$-th location in $\calS_1$ and $j$-th location in $\calS_2$ as
\begin{equation*} \label{eqn:omega_B}
\bOmega_{B,ij} = \Cov\brackets{\by_{1t, i\cdot}, \by_{2t, j\cdot}} = \bB \Cov\brackets{ \bX_t^\top\ba_{1, i\cdot}, \bX_t^{*^\top}\ba_{2, j\cdot} } \bB^\top.
\end{equation*}
When $r \ll p$, it is reasonable to assume that $\rank\brackets{\bOmega_{B,ij}} = r$.
Let
\begin{equation}  \label{eqn:M_B}
\bM_B = \sum_{i=1}^{m} \sum_{j=1}^{m} \bOmega_{B,ij}\bOmega^\top_{B,ij}.
\end{equation}
Then, $\bM_B$ has $r$ positive eigenvalues and $\bM_B \bq = \bzero$ for any vector $\bq$ perpendicular to $\calM(\bB)$. Therefore, the columns of a matrix representation of $\calM(\bB)$ can be estimated as the $r$ orthonormal eigenvectors of matrix $\bM_B$ corresponding to the largest $r$ positive eigenvalues in the descending order.

Define the sample version of $\bOmega_{B,ij}$ and $\bM_B$ for centered observation $\bY_{t}$ as
\begin{equation} \label{eqn:M_B_hat}
\hat{\bOmega}_{B,ij} = \frac{1}{T} \sum_{t=1}^{T} \by_{1t, i\cdot} \by^\top_{2t, j\cdot}, \quad \hat{\bM}_B = \sum_{i=1}^m \sum_{j=1}^{m} \hat{\bOmega}_{B,ij}\hat{\bOmega}^\top_{B,ij}.
\end{equation}
A natural estimator for a matrix representation of $\calM(\bB)$ under constraint (\ref{eqn:Q-constraint}) is defined as
\begin{equation*}
\hat \bQ_B = \{ \hat{\bq}_{B,1}, \cdots, \hat{\bq}_{B,r} \},
\end{equation*}
where $\hat{\bq}_{B,i}$ is the eigenvector of $\hat{\bM}_B$ corresponding to its $i$-th largest eigenvalue.
Matrix $\hat{\bQ}_{B}$ estimates $\bB$ up to a scalar factor while sharing the same column space.

The above estimation procedure assumes that the latent dimensions $d \times r$ are known.
However, in practice we need to estimate $d$ and $r$ as well.
Two methods of estimating the latent dimension are (a) the eigenvalue ratio-based estimator, similar to those defined in \cite{lam2012factor, wang2019factor}; (b) the Scree plot which is standard in principal component analysis.
Let $\hat{\lambda}_1 \ge \hat{\lambda}_2 \ge \cdots \ge \hat{\lambda}_{r} \ge 0$ be the ordered eigenvalues of $\hat \bM_{B}$. The ratio-based estimator for $r$ is defined as
\begin{equation}  \label{eqn:eigen-ratio}
\hat{r} = \underset{{1 \le j \le r_{\max}}}{\arg \max} \frac{\hat{\lambda}_{j}}{\hat{\lambda}_{j+1}},
\end{equation}
where $r \le r_{\max} \le p$ is an integer.
In practice we may take $r_{\max} = \lceil p/2 \rceil$ or $r_{\max} = \lceil p/3 \rceil$.
Ratio estimators $\hat d_1$ and $\hat d_2$ is defined similarly with respect to $\hat \bM_{A_1}$ and $\hat \bM_{A_2}$, respectively.
We set $\hat d = \max \{ \hat d_1, \hat d_2 \}$.
\cite{chen2019statistical} shows that eigen-ratio estimators $\hat{d}$ and $\hat{r}$ are consistent under a similar setting.

\subsection{Signal matrix \texorpdfstring{$\bXi_t$}{Xit}} \label{subsec:est:latentfacsignal}

By (\ref{eqn:stvp_mat_div}), the estimators of two representations of the rotated latent matrix factor $\bZ_t$, $t \in [T]$, are defined as
\begin{equation} \label{eqn:fac_mat_est_diff}
\hat \bZ_{1t} = \hat \bQ_{A,1}^\top \bY_{1t} \hat \bQ_B, \qquad \hat \bZ_{2t} = \hat \bQ_{A,2}^\top \bY_{2t} \hat \bQ_B.
\end{equation}
The latent signal process are estimated by
\begin{equation} \label{eqn:signal_mat_est_sep}
\hat{\bXi}_t = \begin{bmatrix}\hat{\bXi}_{1 t} \\ \hat{\bXi}_{2 t} \end{bmatrix},
\end{equation}
where
\begin{equation*}
\hat{\bXi}_{1 t} = \hat \bQ_{A,1} \hat \bZ_{1t} \hat \bQ^\top_B = \hat \bQ_{A,1} \hat \bQ_{A,1}^\top \bY_{1t} \hat \bQ_B \hat \bQ^\top_B, \qquad  \hat{\bXi}_{2 t} = \hat \bQ_{A,2} \hat \bZ_{2t} \hat \bQ_B^\top = \hat \bQ_{A,2} \hat \bQ_{A,2}^\top \bY_{2t} \hat \bQ_B \hat \bQ_B^\top.
\end{equation*}
Equation (\ref{eqn:fac_mat_est_diff}) provides two estimates of $\bZ_t$ based on two partitioned sets of locations. Section \ref{sec:re-estimate-A} will re-estimate a unified version of latent factor matrix $\bZ_t$ from all sampling locations. Estimator of the latent signal process will also be re-estimated from all sampling locations. 

To mitigate the estimation error associated with the random partition of the location set, one could again carry out the estimation procedure with multiple random partitions and return the average estimates, similar to those done in \cite{huang2016krigings}.
To keep the core idea clear, we do not consider random partitions in this paper.
The results for the average estimates from random partitions can be derived similarly to \cite{huang2016krigings} based on the results of the present paper.

\subsection{Spatial loading space \texorpdfstring{$\calM\paran{\bA}$}{MA} and loading function \texorpdfstring{$\bQ_A(\bs)$}{QA}} \label{sec:re-estimate-A}

The procedure in Section \ref{subsec:est:A1A2} only estimates the spatial loading matrices $\hat \bQ_{A,1}$ and $\hat \bQ_{A,2}$ on two partitioned set of sampling locations.
Estimate loading functions from $\hat \bQ_{A,1}$ and $\hat \bQ_{A,2}$ separately will result in inefficient use of sampling locations.
In addition, equation \eqref{eqn:fac_mat_est_diff} gives estimators for two different representations of the latent matrix factor $\bZ_t$.
To get estimators of the $n \times d$ spatial loading matrix $\bQ_A$ for all sampling locations and $\bZ_t$, we use the estimated $\hat{\bXi}_t$ to re-estimate $\hat \bQ_A$ and $\hat{\bZ}_t$.

Recall that the population signals process is $\bxi_t(\bs) = \bB \bX^\top_t \bQ_A(\bs) = \bQ_B \bZ^\top_t \bq_a(\bs) $ and the $n \times p$ matrix $\bXi_t = \bA \bX_t \bB^\top = \bQ_A \bZ_t \bQ_B^\top$ is the signal matrix at discretized sampling locations at each time $t$.
To reduce dimension, we use $\bPsi_t = \bQ_A \bZ_t \in \mathbb{R}^{n \times r}$, rather than $\bXi_t \in \mathbb{R}^{n \times p}$.
Define
\[
\bM_A = \sum_{j=1}^{r} \Cov\brackets{ \bPsi_{t, \cdot j}, \bPsi_{t, \cdot j} } = \bQ_A \sum_{j=1}^{r}  \Cov\brackets{ \bZ_{t, \cdot j}, \bZ_{t, \cdot j} } \bQ_A^\top.
\]
However, true $\bXi_t$ or $\bPsi_t$ are not observable.
We estimate $\hat \bXi_t$ from \eqref{eqn:signal_mat_est_sep} and obtain
\[
\hat \bPsi_t = \hat \bXi_t \hat \bQ_B.
\]
From estimated values, we defined the estimated version of $\bM_A$ as
\[
\hat \bM_A = \frac{1}{T} \sum_{t=1}^T \hat \bPsi_t \hat \bPsi_t^\top,
\]
where $\hat{\bPsi}$ is chosen over $\hat{\bXi}$ because $\hat{\bPsi}$ has the same estimation error bound but is of lower dimension.

A natural estimator of a matrix representation of $\calM(\bA)$ under constraint \eqref{eqn:Q-constraint} is defined as
\[
\hat \bQ_A = \{ \hat{\bq}_{A,1}, \cdots, \hat{\bq}_{A,n} \},
\]
where $\hat{\bq}_{A,i}$ is the eigenvector of $\hat{\bM}_A$ corresponding to its $i$-th largest eigenvalue.
Matrix $\hat{\bQ}_{A}$ estimates $\bA$ up to a scalar factor while sharing the same column space.

The estimator of the rotated latent factor matrix $\bZ_t$ is obtained as
\begin{equation}
\hat \bZ_t = \hat \bQ_A^\top \hat \bPsi_t.
\end{equation}

Once $\hat \bQ_A$ is estimated, we estimate loading functions $q_{a,j}(\bs)$ from the estimated $n$ observations in column $\hat \bQ_{A, \cdot j}$ by the sieve approximation. Any set of bivariate basis functions can be chosen. In our procedure, we consider the tensor product linear sieve space $\Theta_n$, which is constructed as a tensor product space of some commonly used univariate linear approximating spaces, such as B-spline, orthogonal wavelets and polynomial series. Then for each $j \le d$,
\[
q_{a,j}(\bs) = \sum_{i=1}^{J_n} \beta_{i,j} u_i(\bs) + r_j(\bs).
\]
Here $\beta_{i,j}$ are the sieve coefficients of $i$ basis function $u_i(\bs)$ corresponding to the $j$-th factor loading function;  $r_j(\bs)$ is the sieve approximation error; $J_n$ represents the number of sieve terms which grows slowly as $n$ goes to infinity. We estimate $\hat{\beta}_{i,j}$ and the loading functions are approximated by $\hat q_{a,j}(\bs) = \sum_{i=1}^{J_n} \hat{\beta}_{i,j} u_i(\bs)$.

\section{Prediction} \label{sec:prediction}

\subsection{Spatial Prediction}

A major focus of spatial-temporal data analysis is the prediction of variable of interest over new locations.
For some new location $\bs_0 \in \calS$ and $\bs_0 \ne \bs_i$, $i \in [n]$, we aim to predict the unobserved value $\by_t(\bs_0)$ observations $\bY_t$, $t=[T]$.
By \eqref{eqn:stvp_c0}, we have $\by_t(\bs_0) = \bxi_t(\bs_0) + \bepsilon_t(\bs_0) = \bQ_B \bZ'_t \bq_a(\bs_0) + \bepsilon_t(\bs_0)$.
As recommended by \cite{cressie2015statisticsST}, we predict $\bxi_t(\bs_0)=\bQ_B \bZ'_t \bq_a(\bs_0)$ instead of $\by_t(\bs_0)$ directly.
Thus, a natural estimator is
\begin{equation} \label{eqn:pred_xi_s0}
	\hat{\bxi}_t(\bs_0) = \hat \bQ_B \hat \bZ'_t \hat \bq_a(\bs_0),
\end{equation}
where $\hat \bQ_B$, $\hat \bZ_t$ and $\hat \bq_a(\bs)$ are estimated following procedures in Section \ref{sec:est}.

For univariate spatial temporal process, \cite{huang2016krigings} propose the kriging with kernel smoothing for spatial prediction.
This method can be extended to our case by applying kriging with kernel smoothing for each one of the multivariate spatial temporal process.
We implement both our spatial prediction based on \eqref{eqn:pred_xi_s0} and kriging with kernel smoothing for each one of the multivariate spatial temporal process. Empirical results on synthetic as well as real data show that our method performance better than the kriging with kernel smoothing method.

\subsection{Temporal Prediction}

Temporal prediction focuses on predict the future values $\by_{t+h}(\bs_1), \ldots, \by_{t+h}(\bs_n)$ for some $h \ge 1$. By \eqref{eqn:stvp_c0}, we have $\by_{t+h}(\bs) = \bxi_{t+h}(\bs) + \bepsilon_{t+h}(\bs) = \bQ_B \bZ'_{t+h} \bq_a(\bs) + \bepsilon_{t+h}(\bs)$. 
Since $\bepsilon_{t+h}(\bs)$ is unpredictable white noise, the ideal predictor for $\by_{t+h}(\bs)$ is that for $\bxi_{t+h}(\bs)$. 
Thus, we focus on predict $\bxi_{t+h}(\bs) = \bQ_B \bZ'_{t+h} \bq_a(\bs)$. 
The temporal dynamics of the $\bxi_{t+h}(\bs)$ present in a lower dimensional matrix factor $\bZ_{t+h}$, thus a more effective approach is to predict $\bZ_{t+h}$ based on $\bZ_{t-l}, \ldots, \bZ_t$ where $l$ is a prescribed integer. 
Time series analysis \citep{tsay2014multivariate, tsay2018nonlinear} can be applied to $\bZ_t$ under general settings.
We use the auto-regression of order one (AR(1)) and take $l = 1$ to illustrate the idea.

Since the latent factor matrix time series $\bZ_t \in \mathbb{R}^{d \times r}$ is of low-dimension,
a straight forward method for predicting $\bZ_{t+h}$ is applying the multivariate time series analysis techniques to $\vect{\bZ_t}$.
Under vector auto-regressive model of order 1 -- VAR(1), we have
\begin{equation}
\vect{\bZ_t} = \bPhi \, \vect{\bZ_{t-1}} + \bu_t,  \nonumber
\end{equation}
where $\bPhi \in \mathbb{R}^{ d r \times d r }$ is the coefficient matrix of the VAR(1). Following the vector time series analysis \citep{tsay2014multivariate, tsay2018nonlinear}, we obtain estimators $\hat{\bPhi}$. A $h$-step forward prediction is given by
\begin{equation} \label{eqn:vec1}
\hat \bZ_{t+h}^{VAR} = \mat{ \hat \bPhi^h \vect{ \hat \bZ_t}  }.
\end{equation}

To preserve the matrix structure intrinsic to $\bZ_t$, we model $\{\bZ_t\}_{1:T}$ as the matrix auto-regressive model of order 1 -- MAR(1) \citep{yang2017autoregressive}.
 Mathematically,
\begin{equation}
\bZ_{t} = \bPhi_R \, \bZ_{t-1} \, \bPhi_C + \bU_t, \nonumber
\end{equation}
where $\bPhi_R \in \mathbb{R}^{d \times d}$ and $\bPhi_C \in \mathbb{R}^{r \times r}$ are row and column coefficient matrices, respectively.
The covariance structure of the matrix white noise $\bU_t$ is not restricted.
Thus, $\mathbf{vec}({\bU_t}) \sim \calN(\bzero, \bSigma_U)$ where $\bSigma_U$ is an arbitrary covariance matrix.
Matrix $\bPhi_R$ captures the auto-correlations between the spatial latent factors and $\bPhi_C$ captures the auto-correlations between the variable latent factors.
Following the generalized iterative method proposed in \cite{yang2017autoregressive}, we obtain estimators $\hat{\bPhi}_R$ and $\hat{\bPhi}_C$. A $h$-step forward prediction is given by
\begin{equation} \label{eqn:mar1}
\hat \bZ_{t+h}^{MAR} = \hat{\bPhi}^h_R \, \hat{\bZ}_{t} \, \hat{\bPhi}_C^h.
\end{equation}

Having an estimator $\hat \bZ_{t+h}$ from either vector AR(1) \eqref{eqn:vec1} or matrix AR(1) \eqref{eqn:mar1}, we obtain the prediction for $\by_{t+h}(\bs)$ by
\begin{equation}
\hat{\bxi}_{t+h}(\bs)= \hat \bQ_B \, \hat \bZ'_{t+h} \, \hat \bq_a(\bs),
\end{equation}
where $\hat \bQ_B$, $\hat \bZ_t$ and $\hat \bq_a(\bs)$ are estimated following procedures in Section \ref{sec:est}.

The advantage of MAR(1) over VAR(1) is that the number of unknowns in $\bPhi_R \in \mathbb{R}^{d \times d}$ and $\bPhi_C \in \mathbb{R}^{r \times r}$ is smaller than that in $\bPhi \in \mathbb{R}^{ d r \times d r }$.
This is especially important in high-dimensional setting.
Since the latent matrix factor $\bZ_t$ is of low-dimension in our case, they have similar performance as shown in the simulation.
\section{Asymptotic properties}  \label{sec:theory}

In this section, we investigate the rates of convergence for the estimators under the setting that $n$, $p$ and $T$ all go to infinity while $d$ and $r$ are fixed and the factor structure does not change over time.

\begin{assumption} \label{cond:vecXt_alpha_mixing}
\textbf{Alpha-mixing.} $\{ \Vec\brackets{\bX_t}, t=0,\pm 1, \pm 2, \cdots \}$ is $\alpha$-mixing. Specifically, for some $\gamma > 2$, the mixing coefficients satisfy the condition that
\[
\sum_{h=1}^{\infty} \alpha(h)^{1-2/\gamma} < \infty,
\]
where $ \alpha(h)=\underset{\tau}{\sup} \underset{A \in \mathcal{F}_{-\infty}^{\tau}, B \in \mathcal{F}_{\tau+h}^{\infty}}{\sup} \left| P(A \cap B) - P(A)P(B) \right| $ and $\mathcal{F}_{\tau}^s$ is the $\sigma$-field generated by $\{vec(\bX_t): {\tau} \le t \le s \}$.
\end{assumption}

\begin{assumption}  \label{cond:Xt_cov_fullrank_bounded}
Let $X_{t, ij}$ be the $ij$-th entry of $\bX_t$. Then,
$ E(\left| X_{t, ij} \right|^{2}) \le C$
for any $i = 1, \ldots, d$, $j = 1, \ldots, r$ and $t = 1, \ldots, T$, where $C$ is a positive constant and $\gamma$ is given in Condition \ref{cond:vecXt_alpha_mixing}.
\end{assumption}

Assumption \ref{cond:vecXt_alpha_mixing} requires the random vector $\Vec\brackets{\bX_t}$ be $\alpha$-mixing -- weaker than stationarity. 
Each entry of covariance matrix $\Var{\left[\Vec\brackets{\bX_t}\right] }$ is bounded according to Assumption \ref{cond:Xt_cov_fullrank_bounded}. There is no further requirement on the temporal dependence structure on $\bX_t$, i.e., $\Cov{[\Vec(\bX_{t_1}),\Vec(\bX_{t_2})']}$, $t_1\neq t_2$. This is weaker than that required in \cite{wang2019factor}. 
The following three assumptions control the signal noise ratio.
Matrix $\bXi_t$ can be seen as the signal of the observation $\bY_t$ and $\bE_t$ as the noise. 
Assumption \ref{cond:eigenval_cov_Et_bounded} control the noise strength by bounding each entry of spatial covariance matrix of noise $\bE_t$. 
The signal strength is measured jointly by the $L_2$-norm $\norm{\bA}_2^2$ and $\norm{\bB}_2^2$, which correspond to the spatial and variable strengthes, respectively. 

\begin{assumption}  \label{cond:eigenval_cov_Et_bounded}
\textbf{Noise strength.} Each entry of $\Var\brackets{\Vec\brackets{\bE_t}}$ remains bounded as $n$ and $p$ increase to infinity.
\end{assumption}

\begin{assumption}  \label{cond:B_factor_strength} 
\textbf{Variable factor strength.} There exists a constant $\gamma \in [0,1]$ such that $\norm{\bB}^2_{min} \asymp p^{1-\gamma} \asymp \norm{\bB}^2_2$ as $p$ goes to infinity and $r$ is fixed.
\end{assumption}

\begin{assumption} \label{cond:A_factor_strength}
\textbf{Spatial factor strength.}
For any partition $\{ \calS_1, \calS_2 \}$ of locations $\calS = \{ \bs_1, \ldots, \bs_n\}$, we have  $\norm{\bA_1}^2_{min} \asymp n_1 \asymp \norm{\bA_1}^2_2$ and $\norm{\bA_2}^2_{min} \asymp n_2 \asymp \norm{\bA_2}^2_2$ for any $\bs\in\mathcal{S}$, where $n_1$ and $n_2$ are the number of locations in sets $\calS_1$ and $\calS_2$ respectively. 
\end{assumption}
This assumption is satisfied automatically under Assumption \ref{cond:Holder} with randomly sampled $\calS_1$ and $\calS_2$. Assumption \ref{cond:Holder} further guarantee the accuracy of sieve approximation of loading function $\ba_j(\bs)$, $j=1,2,\ldots,d$.
\begin{assumption} \label{cond:Holder}
\textbf{Loading functions belongs to H{\"o}lder class.} For $j=1, \ldots, d$, the loading functions $\ba_j(\bs)$, $\bs \in \calS \in \mathbb{R}^2$ belongs to a H{\"o}lder class $\calA^{\kappa}_c(\calS)$ ($\kappa$-smooth) defined by
\[
\calA^{\kappa}_c(\calS) = \left \{a \in \calC^m(\calS): \underset{[\eta]\le m}{\sup} \; \underset{\bs \in \calS}{\sup} \left| D^{\eta}\, a(\bs) \right| \le c, \text{ and }
\underset{[\eta]=m}{\sup} \; \underset{\bu, \bv \in \calS}{\sup} \frac{\left| D^{\eta}\, a(\bu) - D^{\eta}\, a(\bv) \right|}{\norm{\bu - \bv}^{\alpha}_2}  \le c   \right \},
\]
for some positive number $c$. Here, $\calC^m(\calS)$ is the space of all $m$-times continuously differentiable real-value functions on $\calS$. The differential operator $D^{\eta}$ is defined as $D^{\eta} = \frac{\partial^{[\eta]}}{\partial s_1^{\eta_1} \partial s_2^{\eta_2}}$ and $[\eta] = \eta_1 + \eta_2$ for non-negative integers $\eta_1$ and $\eta_2$.
\end{assumption}

Theorem \ref{thm:A_sep_err_bnd} and \ref{thm:B_sep_err_bnd} present the error bounds for the estimated spatial loading spaces $\calM \paran{\bA_l}$, $l=1,2$, on partitioned sampling locations and for estimated variable loading space $\calM \paran{\hat \bB}$, respectively.
Asymptotically, the bounds are the similar to those derived under the time series settings in \cite{wang2019factor} and \cite{chen2019constrained}. 
Indeed, when we only consider the samples from discrete locations with spatial white noises, the estimation of model \eqref{eqn:stvp_c0_mat} is similar to that of the matrix-variate time series with temporal white noise. 

\begin{theorem} \label{thm:A_sep_err_bnd}
Under Assumption \ref{cond:vecXt_alpha_mixing}-\ref{cond:Holder} and $p^{\gamma} T^{-1/2} = \smlo{1}$, we have
\begin{equation*}
\calD \left( \calM(\hat{\bA}_i), \calM(\bA_i) \right) = \norm{\hat \bQ_{A,l} - \bQ_{A,l}} = \Op{ \sqrt{ n_1 n_2^{-1} p^{\gamma} + n_1^{-1} n_2 p^{\gamma} +  p^{2\gamma} } \, T^{-1/2} }.
\end{equation*}
If $n_1 \asymp n_2 \asymp n$, we have
\begin{equation*}
\calD \left( \calM(\hat{\bA}_i), \calM(\bA_i) \right) = \norm{\hat \bQ_{A,l} -  \bQ_{A,l}} = \Op{ p^{\gamma} T^{-1/2} }.
\end{equation*}
\end{theorem}

\begin{theorem} \label{thm:B_sep_err_bnd}
Under Assumption \ref{cond:vecXt_alpha_mixing}-\ref{cond:Holder} and $p^{\gamma} T^{-1/2} = o(1)$, we have
\begin{equation*}
\calD \left( \calM(\hat \bB), \calM(\bB) \right) = \norm{\hat \bQ_B - \bQ_B} = \Op{  p^{\gamma} T^{-1/2} }.
\end{equation*}
\end{theorem}

When $p$ is fixed, the convergence rate of $\bA_i$ and $\bB$ are $\sqrt{T}$, $i=1,2$. 
If dimension $p$ increases, the estimations of $\bA_i$ and $\bB$ become more difficult. 
The noise term is of order $np$. 
The signal contribute the accuracy of $\hat{\bA}_i$ and $\bB$ with $np^{1-\gamma}$, which is affected by the variable strength $\gamma$. 
If $\gamma$ is small (strong variable factor strength), the convergence speed of $\bA_i$ and $\bB$ is faster. 
Specifically, the convergence rate of $\bA_i$ and $\bB$ are not affected by $n$. 
The noise term and the signal contribution both have order $n$.

\begin{theorem} \label{thm:signal_err_bnd}
Under Assumption \ref{cond:vecXt_alpha_mixing}-\ref{cond:Holder} and $ p^{\gamma} T^{-1/2}=\smlo{1}$, if $n_1 \asymp n2 \asymp n$, then
\begin{equation*}
\frac{1}{\sqrt{np} }\norm{\hat{\bXi}_{it} - \bXi_{it}}_2  = \Op{ p^{\gamma/2} T^{-1/2} + n^{-1/2} p^{-1/2}},
\end{equation*}
for $i = 1, 2$, and
\begin{equation*}
\frac{1}{\sqrt{np}} \norm{\hat \bXi_{t} - \bXi_{t}}_2  = \Op{ p^{\gamma/2} T^{-1/2} + n^{-1/2} p^{-1/2}}
\end{equation*}
\end{theorem}
Theorem \ref{thm:signal_err_bnd} presents the error bound for estimated signal $\hat \bXi_{it}$ as in (\ref{eqn:signal_mat_est_sep}) for each partition and $\hat\bXi_t$ for all sampling locations. The error of estimated signal $\hat\bXi_{it}$ is contributed by the noise $\bE_t$, and the estimation error for $\bQ_{A}$ and $\bQ_{B}$. In the proof of this theorem, we show that $p^{\gamma/2}T^{-1/2}$ comes from the estimation error for $\bQ_{A}$ and $\bQ_{B}$ in Theorem \ref{thm:A_sep_err_bnd} and \ref{thm:B_sep_err_bnd}. Since we use the sample $\bY_t$ instead of $\bXi_t$, $n^{-1/2} p^{-1/2}$ comes from the noise $\bE_t$, which is a $p\times n$ matrix.
Theorem \ref{thm:reest_Q_A} presents the error bond for re-estimated spatial loading space $\calM \paran{\hat \bA}$ from estimated $\hat \bXi_{t}$ and $\hat \bQ_B$ of the first step.

\begin{theorem} \label{thm:reest_Q_A}
Under Assumption \ref{cond:vecXt_alpha_mixing}-\ref{cond:Holder} and $ p^{\gamma} T^{-1/2}=\smlo{1}$, if $n_1 \asymp n2 \asymp n$, then
\[
\norm{\hat\bQ_A - \bQ_A}_2 = \Op{ p^{\gamma} T^{-1/2} + n^{-1/2} p^{\gamma/2-1/2}}.
\]
\end{theorem}
Re-estimation introduces the noise error from $\bE_t$. Comparing to the result in Theorem \ref{thm:A_sep_err_bnd}, the re-estimated loading matrix $\hat\bQ_A$ has an extra error term $n^{-1/2} p^{\gamma/2-1/2}$, which results from the noise error $n^{-1} p^{-1}$ of $\hat\bXi_t$ that appears in Theorem \ref{thm:signal_err_bnd}. 
Simulations in Section \ref{sec:simu} show that the differences between the re-estimator and first estimator of $\bA$ are also negligible with finite $n$, $p$, $T$.

Let $\Delta_{npT} = p^{\gamma} T^{-1} + n^{-1} p^{-1}$ represent the estimation error from the first-step estimation. 
Note that under the identification constraint that $\bQ_A$ and $\bQ_B$ are orthonormal matrices, $\norm{\bZ_t}$ is of order $np^{1-\gamma}$.
Theorem \ref{thm:Zt-bound} shows the normalized error bound of $\bZ_t$. 

\begin{theorem} \label{thm:Zt-bound}
Under Assumption \ref{cond:vecXt_alpha_mixing}-\ref{cond:Holder}, the estimator of rotated latent factor matrix $\bZ_t$ satisfies
\begin{equation}
\frac{1}{np}\norm{\hat \bZ_t - \bZ_t}^2_2 = \Op{\Delta_{npT} +  p^{\gamma} \Delta_{npT}^2}.  \nonumber
\end{equation}
\end{theorem}

Theorem \ref{thm:space-krig-bound} presents the space kriging error bound based on sieve approximated function $\hat{\bA}(\bs)$. In the proof of Theorem \ref{thm:space-krig-bound}, we decompose the error of $\hat \bxi_t(\bs_0)$ and show that it is dominated by three parts. $J_n^{-2\kappa}  p^{-\gamma} $ is roughly the error of $\bq_a(\bs_0)$, which includes the sieve approximation error and estimation error. $ \Delta_{npT}  p^{\gamma} + \Delta_{npT}^2$ comes from the error of $\hat{\bZ_t}$, and $ p^{\gamma} T^{-1}$  is the error of $\hat \bQ_B$.

\begin{theorem} \label{thm:space-krig-bound}
Under Assumption \ref{cond:vecXt_alpha_mixing}-\ref{cond:Holder}, for a new site $\bs_0\in\mathcal{S}$
\begin{eqnarray}
\frac{1}{p} \norm{\hat \bxi_t(\bs_0) - \bxi_t(\bs_0)}_2^2 & = &  \Op{ J_n^{-2\kappa}  p^{-\gamma} + \Delta_{npT}   + p^{\gamma}\Delta_{npT}^2 +  p^{\gamma} T^{-1} }. \nonumber
\end{eqnarray}
\end{theorem}

\section{Simulation} \label{sec:simu}
In this section we study the numerical performance of the proposed method on synthetic data sets. We let $\bs_1, \cdots, \bs_n$ be drawn randomly from the uniform distribution on $[-1,1]^2$ and the observed data $\by_t(\bs)$ be generated according to Model \eqref{eqn:stvp_c0}:
\begin{equation}
\by_t(\bs) = \bxi_t(\bs) + \bepsilon_t(\bs) = \bB \bX'_t \ba(\bs) + \bepsilon_t(\bs). \nonumber
\end{equation}
The dimensions of $\bX_t$ are chosen to be $d=3$, $r=2$, and are fixed in all simulations. The latent factor $\bX_t$ is generated from the Gaussian matrix time series \eqref{eqn:mar1}:
\begin{equation}
\bX_{t} = \bPhi_R \, \bX_{t-1} \, \bPhi_C + \bU_t, \nonumber
\end{equation}
where $\bPhi_R = diag(0.7, \, 0.8, \, 0.9)$, $\bPhi_C = diag(0.8, \, 0.6)$ and the entries of $\bU_t$ are white noise Gaussian process with mean $\bzero$ and covariance structure such that $\bSigma_U = \Cov{vec(\bU_t)}$. Here we use $\bSigma_U=\bI_{dr}$. Alternatively, we could use Kronecker product covariance structure $\bSigma_U=\bSigma_C \otimes \bSigma_R$ or arbitrary covariance matrix $\bSigma_U$. As shown in \cite{yang2017autoregressive} and from our own experiments, this setting does not affect much on the results.

The entries of $\bB$ is independently sampled from the uniform distribution $\calU(-1,1) \cdot p^{\gamma/2}$. The nugget process $\bepsilon_t(\bs)$ are independent and normal with mean $\bzero$ and the covariance $(1+s_1^2+s_2^2)/2\sqrt{3}\cdot \bI_{p}$. The basis functions $a_j(\bs)$'s are designed to be
\begin{equation}\label{eqn:a_function}
a_1(\bs) = (s_1 - s_2)/2, \quad a_2(\bs)=\cos \left( \pi \sqrt{2(s_1^2+s_2^2)} \right), \quad a_3(\bs)=1.5s_1s_2.
\end{equation}
With the above generating model setting, the signal-noise-ratio of $p$-dimensional variable, which is defined as
\[
SNR \equiv \frac{ \int_{\bs \in [-1,1]^2} Trace \left[ Cov  \left(\bxi_{t}(\bs) \right) \right] d\bs  }{ \int_{\bs \in [-1,1]^2} Trace \left[ Cov \left(\bepsilon_{t}(\bs)\right) \right] d\bs } \approx 2.58.
\]
We run $200$ simulations for each combination of $n = 50, 100, 200, 400$, $p = 10, 20, 40$, and $T= 60, 120, 240$. With each simulation, we calculate $\hat{d}$, $\hat{r}$, $\hat{\bA}_1$, $\hat{\bA_2}$, $\hat{\bB}$ and $\hat{\bXi}_t$, re-estimate $\hat{\bA}$ and $\widetilde{\bXi}_t$, then use $\hat{\bA}$ to get approximated $\hat{a}_j(\bs)$ following the estimation procedure described in Section \ref{sec:est}.

Table \ref{table:freq_dr_est} presents the relative frequencies of estimated rank pairs over 200 simulations. The columns corresponding to the true rank pair $(3,2)$ is highlighted.

Specifically, we show the estimated performance of spatial loading matrix $\bA$, spatial-temporal covariance $\bSigma_{\xi,|t_1-t_2|}(\bu,\bv)$ and latent factor $\bff_t(\bs)$. Let $p=40$, $n=400$ and $T=240$. Figure \ref{fig:spdistA.0_3d} presents the true surface of loading function $\ba_1(\bs)$, $\ba_2(\bs)$, $\ba_3(\bs)$ in (\ref{eqn:a_function}) on the top, and the fitted surface of $\hat{\ba}_1(\bs)$, $\hat\ba_2(\bs)$, $\hat\ba_3(\bs)$ on the bottom, which are all quite close with the true surface in shape.
Figure \ref{fig:cov_real_xi.0} presents one example of the sample temporal covariance $\tilde{\bSigma}_{\xi,|t_1-t_2|}(\bs_1,\bs_2)$  (top three) and estimated temporal covariance $\hat{\bSigma}_{\xi,|t_1-t_2|}(\bs_1,\bs_2)$ (bottom three) of $\bSigma_{\xi,|t_1-t_2|}(\bs_1,\bs_1)$ with time lag $t_1-t_2=0,1,2$ and $\bs_1$ is randomly selected. Our proposed model and estimation method can duplicate the temporal dependence structure very well. Spatial covariance also shows the similar result.
Figure \ref{fig:spdistf.0} presents the true factor $\bff_t(\bs)$ and the estimated factor $\hat{\bff}_t(\bs)$ by proposed method. We can see that they are very close.

The performance of correctly estimating the loading spaces are measured by the space distance between the estimated and true loading matrices $\hat{\bA}$ and $\bA$, which is defined as
\[
\calD(\calM(\hat{\bA}), \calM(\bA)) = \left( 1 - \frac{1}{\max(d,\hat{d})} tr\left( \hat{\bA} (\hat{\bA}'\hat{\bA})^{-1} \hat{\bA}' \cdot \bA (\bA'\bA)^{-1} \bA' \right) \right)^{\frac{1}{2}}.
\]
It can be shown that $\calD(\calM(\hat{\bA}), \calM(\bA))$ takes its value in $[0,1]$, it equals to $0$ if and only if $\calM(\hat{\bA}) = \calM(\bA)$, and equals to $1$ if and only if $\calM(\hat{\bA}) \perp \calM(\bA)$.

\begin{table}[htpb!]
\centering
\caption{Relative frequency of estimated rank pair $(\hat{d}, \hat{r})$ over 200 simulations. The columns correspond to the true value pair $(3,2)$ are highlighted. Blank cell represents zero value.}
\label{table:freq_dr_est}
\resizebox{0.9\textwidth}{!}{%
\begin{tabular}{ccc|cccc|cccccc}
\hline
\multicolumn{3}{c|}{$(\hat{d}, \, \hat{r})$} & \multicolumn{4}{c|}{$\gamma = 0$} & \multicolumn{6}{c}{$\gamma = 0.5$} \\ \hline
T & p & n & \cellcolor[HTML]{EFEFEF}(3,2) & (3,1) & (2,2) & (1,2) & \cellcolor[HTML]{EFEFEF}(3,2) & (3,1) & (2,2) & (2,1) & (1,2) & (1,1) \\ \hline
60 & 10 & 50 & \cellcolor[HTML]{EFEFEF}0.77 & 0.01 & 0.04 & 0.19 & \cellcolor[HTML]{EFEFEF}0.11 & 0.02 & 0.12 & 0.02 & 0.61 & 0.14 \\
120 & 10 & 50 & \cellcolor[HTML]{EFEFEF}1.00 &   &   & 0.01 & \cellcolor[HTML]{EFEFEF}0.42 &   & 0.08 &   & 0.51 &   \\
240 & 10 & 50 & \cellcolor[HTML]{EFEFEF}1.00 &   &   &   & \cellcolor[HTML]{EFEFEF}0.91 &   & 0.01 &   & 0.09 &   \\ \hdashline
60 & 20 & 50 & \cellcolor[HTML]{EFEFEF}0.86 &   & 0.02 & 0.13 & \cellcolor[HTML]{EFEFEF}0.02 &   & 0.10 &   & 0.88 & 0.01 \\
120 & 20 & 50 & \cellcolor[HTML]{EFEFEF}1.00 &   &   &   & \cellcolor[HTML]{EFEFEF}0.08 &   & 0.04 &   & 0.88 &   \\
240 & 20 & 50 & \cellcolor[HTML]{EFEFEF}1.00 &   &   &   & \cellcolor[HTML]{EFEFEF}0.49 &   & 0.01 &   & 0.50 &   \\ \hdashline
60 & 40 & 50 & \cellcolor[HTML]{EFEFEF}0.96 &   & 0.01 & 0.04 & \cellcolor[HTML]{EFEFEF}0.03 &   & 0.09 &   & 0.88 & 0.01 \\
120 & 40 & 50 & \cellcolor[HTML]{EFEFEF}1.00 &   &   &   & \cellcolor[HTML]{EFEFEF}0.02 &   & 0.07 &   & 0.91 &   \\
240 & 40 & 50 & \cellcolor[HTML]{EFEFEF}1.00 &   &   &   & \cellcolor[HTML]{EFEFEF}0.32 &   & 0.01 &   & 0.68 &   \\ \hline
60 & 10 & 100 & \cellcolor[HTML]{EFEFEF}0.98 &   & 0.02 &   & \cellcolor[HTML]{EFEFEF}0.65 & 0.10 & 0.18 & 0.04 & 0.03 & 0.01 \\
120 & 10 & 100 & \cellcolor[HTML]{EFEFEF}1.00 &   &   &   & \cellcolor[HTML]{EFEFEF}0.99 & 0.01 & 0.01 &   &   &   \\
240 & 10 & 100 & \cellcolor[HTML]{EFEFEF}1.00 &   &   &   & \cellcolor[HTML]{EFEFEF}1.00 &   &   &   &   &   \\  \hdashline
60 & 20 & 100 & \cellcolor[HTML]{EFEFEF}1.00 &   &   &   & \cellcolor[HTML]{EFEFEF}0.73 &   & 0.22 &   & 0.06 &   \\
120 & 20 & 100 & \cellcolor[HTML]{EFEFEF}1.00 &   &   &   & \cellcolor[HTML]{EFEFEF}0.97 &   & 0.04 &   &   &   \\
240 & 20 & 100 & \cellcolor[HTML]{EFEFEF}1.00 &   &   &   & \cellcolor[HTML]{EFEFEF}1.00 &   &   &   &   &   \\  \hdashline
60 & 40 & 100 & \cellcolor[HTML]{EFEFEF}1.00 &   &   &   & \cellcolor[HTML]{EFEFEF}0.72 &   & 0.24 &   & 0.05 &   \\
120 & 40 & 100 & \cellcolor[HTML]{EFEFEF}1.00 &   &   &   & \cellcolor[HTML]{EFEFEF}0.96 &   & 0.04 &   &   &   \\
240 & 40 & 100 & \cellcolor[HTML]{EFEFEF}1.00 &   &   &   & \cellcolor[HTML]{EFEFEF}1.00 &   &   &   &   &   \\ \hline
60 & 10 & 200 & \cellcolor[HTML]{EFEFEF}1.00 &   &   &   & \cellcolor[HTML]{EFEFEF}0.80 & 0.15 & 0.02 & 0.01 & 0.03 &   \\
120 & 10 & 200 & \cellcolor[HTML]{EFEFEF}1.00 &   &   &   & \cellcolor[HTML]{EFEFEF}1.00 & 0.01 &   &   &   &   \\
240 & 10 & 200 & \cellcolor[HTML]{EFEFEF}1.00 &   &   &   & \cellcolor[HTML]{EFEFEF}1.00 &   &   &   &   &   \\ \hdashline
60 & 20 & 200 & \cellcolor[HTML]{EFEFEF}1.00 &   &   &   & \cellcolor[HTML]{EFEFEF}0.94 &   & 0.02 &   & 0.04 &   \\
120 & 20 & 200 & \cellcolor[HTML]{EFEFEF}1.00 &   &   &   & \cellcolor[HTML]{EFEFEF}1.00 &   &   &   &   &   \\
240 & 20 & 200 & \cellcolor[HTML]{EFEFEF}1.00 &   &   &   & \cellcolor[HTML]{EFEFEF}1.00 &   &   &   &   &   \\ \hdashline
60 & 40 & 200 & \cellcolor[HTML]{EFEFEF}1.00 &   &   &   & \cellcolor[HTML]{EFEFEF}0.97 &   & 0.01 &   & 0.03 &   \\
120 & 40 & 200 & \cellcolor[HTML]{EFEFEF}1.00 &   &   &   & \cellcolor[HTML]{EFEFEF}1.00 &   &   &   &   &   \\
240 & 40 & 200 & \cellcolor[HTML]{EFEFEF}1.00 &   &   &   & \cellcolor[HTML]{EFEFEF}1.00 &   &   &   &   &   \\ \hline
60 & 10 & 400 & \cellcolor[HTML]{EFEFEF}1.00 &   &   &   & \cellcolor[HTML]{EFEFEF}0.89 & 0.10 &   &   & 0.02 &   \\
120 & 10 & 400 & \cellcolor[HTML]{EFEFEF}1.00 &   &   &   & \cellcolor[HTML]{EFEFEF}1.00 & 0.01 &   &   &   &   \\
240 & 10 & 400 & \cellcolor[HTML]{EFEFEF}1.00 &   &   &   & \cellcolor[HTML]{EFEFEF}1.00 &   &   &   &   &   \\ \hdashline
60 & 20 & 400 & \cellcolor[HTML]{EFEFEF}1.00 &   &   &   & \cellcolor[HTML]{EFEFEF}1.00 &   &   &   & 0.01 &   \\
120 & 20 & 400 & \cellcolor[HTML]{EFEFEF}1.00 &   &   &   & \cellcolor[HTML]{EFEFEF}1.00 &   &   &   &   &   \\
240 & 20 & 400 & \cellcolor[HTML]{EFEFEF}1.00 &   &   &   & \cellcolor[HTML]{EFEFEF}1.00 &   &   &   &   &   \\ \hdashline
60 & 40 & 400 & \cellcolor[HTML]{EFEFEF}1.00 &   &   &   & \cellcolor[HTML]{EFEFEF}1.00 &   &   &   & 0.01 &   \\
120 & 40 & 400 & \cellcolor[HTML]{EFEFEF}1.00 &   &   &   & \cellcolor[HTML]{EFEFEF}1.00 &   &   &   &   &   \\
240 & 40 & 400 & \cellcolor[HTML]{EFEFEF}1.00 &   &   &   & \cellcolor[HTML]{EFEFEF}1.00 &   &   &   &   &   \\ \hline
\end{tabular}%
}
\end{table}

Figure \ref{fig:spdistA.0} presents the box plot of the average space distance
\[
\frac{1}{2}\left( \calD(\calM(\hat{\bA}_1), \calM(\bA_1)) + \calD(\calM(\hat{\bA}_2), \calM(\bA_2))  \right)
\]
and compare it with the box plot of space distance between re-estimated $\hat{\bA}$ and the truth $\bA$.
Figure \ref{fig:spdistB.0} presents the box plot of the space distance between $\hat{\bB}$ and the truth $\bB$.

\begin{figure}[ht!]
	\centering
	\includegraphics[width=\linewidth,height=\textheight,keepaspectratio=true]{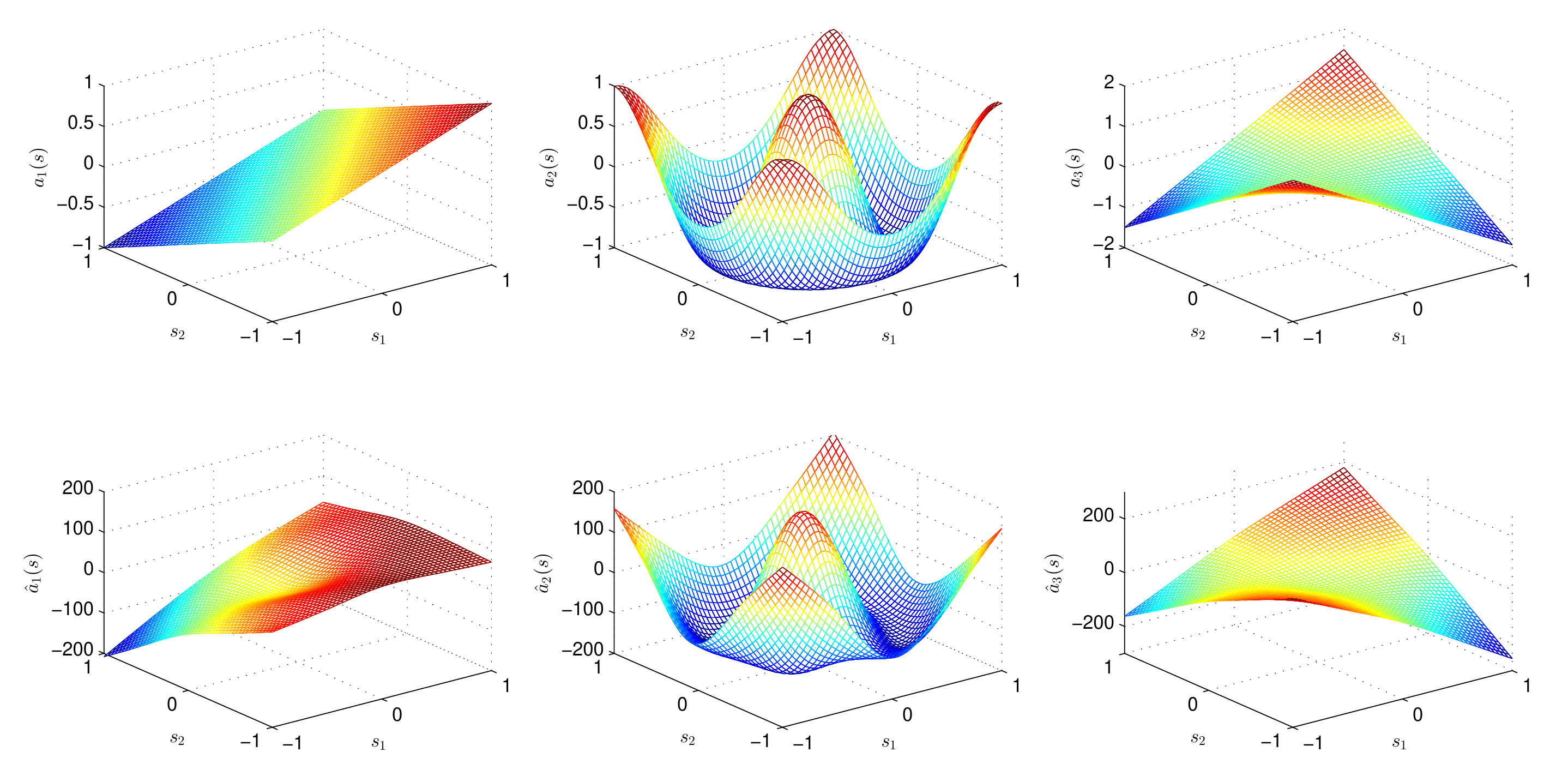}
	\caption{True surface of loading function $\ba_j(\bs)$ (top three) and fitted surface of $\hat{\ba}_j(\bs)$ (bottom three), $j=1,2,3$. Let $n=400$, $p=40$, and $T=240$.}
	\label{fig:spdistA.0_3d}
\end{figure}

\begin{figure}[ht!]
     \centering
     \begin{subfigure}[b]{0.3\textwidth}
         \centering
         \caption*{$\tilde{\bSigma}_{\xi,|t_1-t_2|}(\bs_1,\bs_2)$, $t_1-t_2=0$}
         \includegraphics[width=.95\textwidth]{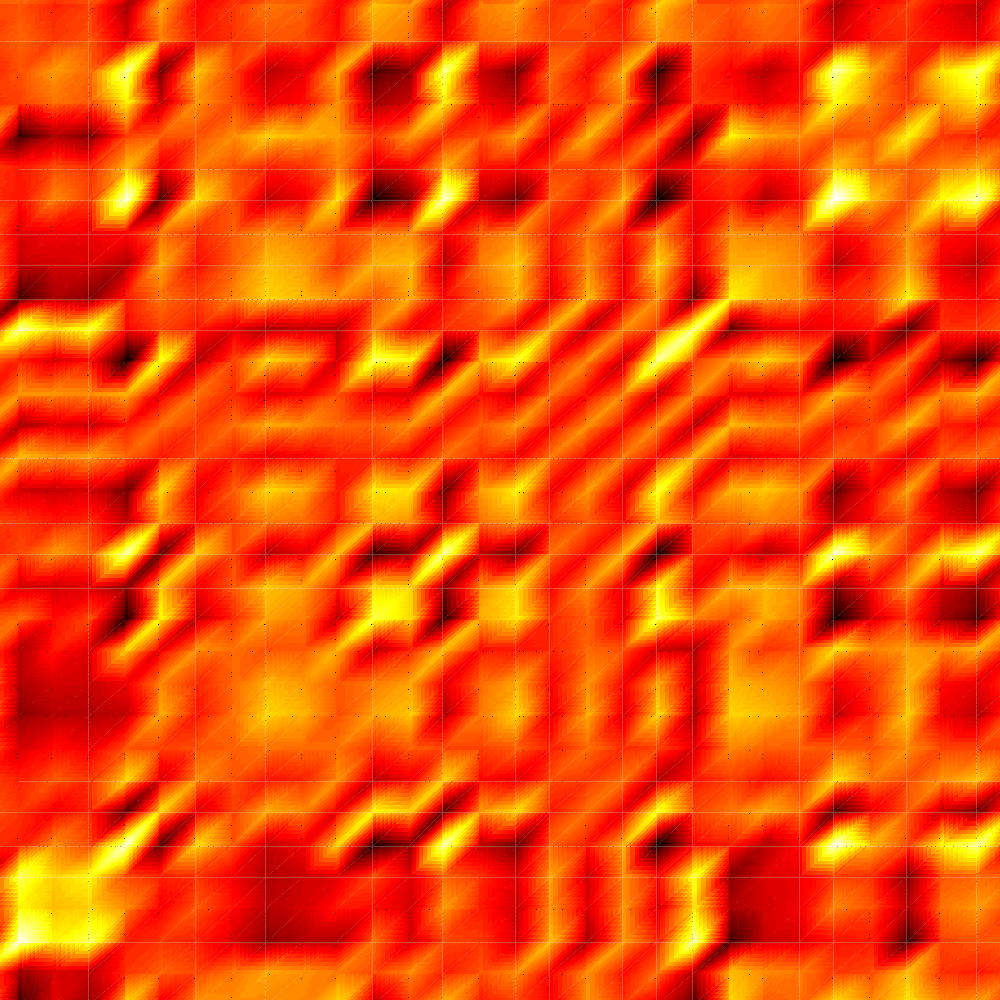}
     \end{subfigure}
     \hfill
     \begin{subfigure}[b]{0.3\textwidth}
         \centering
         \caption*{$\tilde{\bSigma}_{y,|t_1-t_2|}(\bs_1,\bs_2)$, $t_1-t_2=1$}
         \includegraphics[width=.95\textwidth]{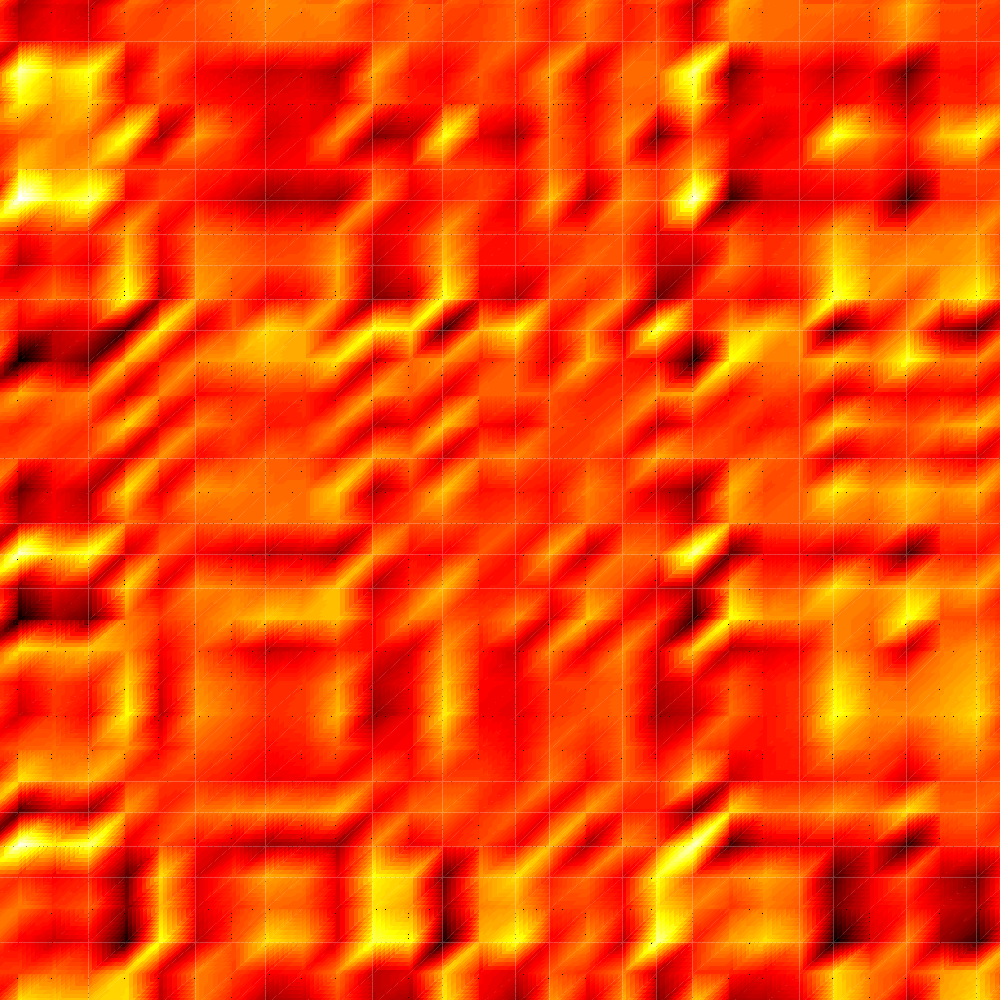}
     \end{subfigure}
     \hfill
     \begin{subfigure}[b]{0.3\textwidth}
         \centering
         \caption*{$\tilde{\bSigma}_{\xi,|t_1-t_2|}(\bs_1,\bs_2)$, $t_1-t_2=2$}
         \includegraphics[width=.95\textwidth]{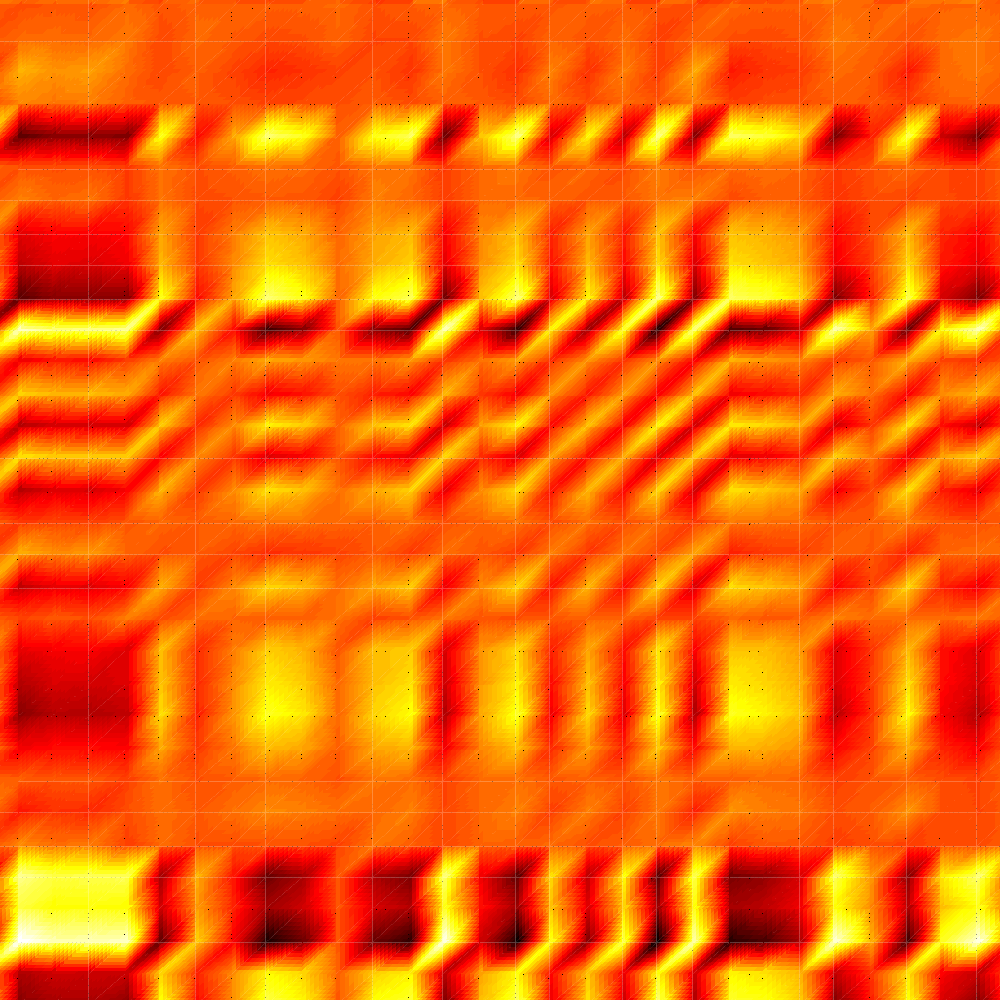}
     \end{subfigure}
     
     \begin{subfigure}[b]{0.3\textwidth}
         \centering
         \caption*{$\hat{\bSigma}_{\xi,|t_1-t_2|}(\bs_1,\bs_2)$, $t_1-t_2=0$}
         \includegraphics[width=.95\textwidth]{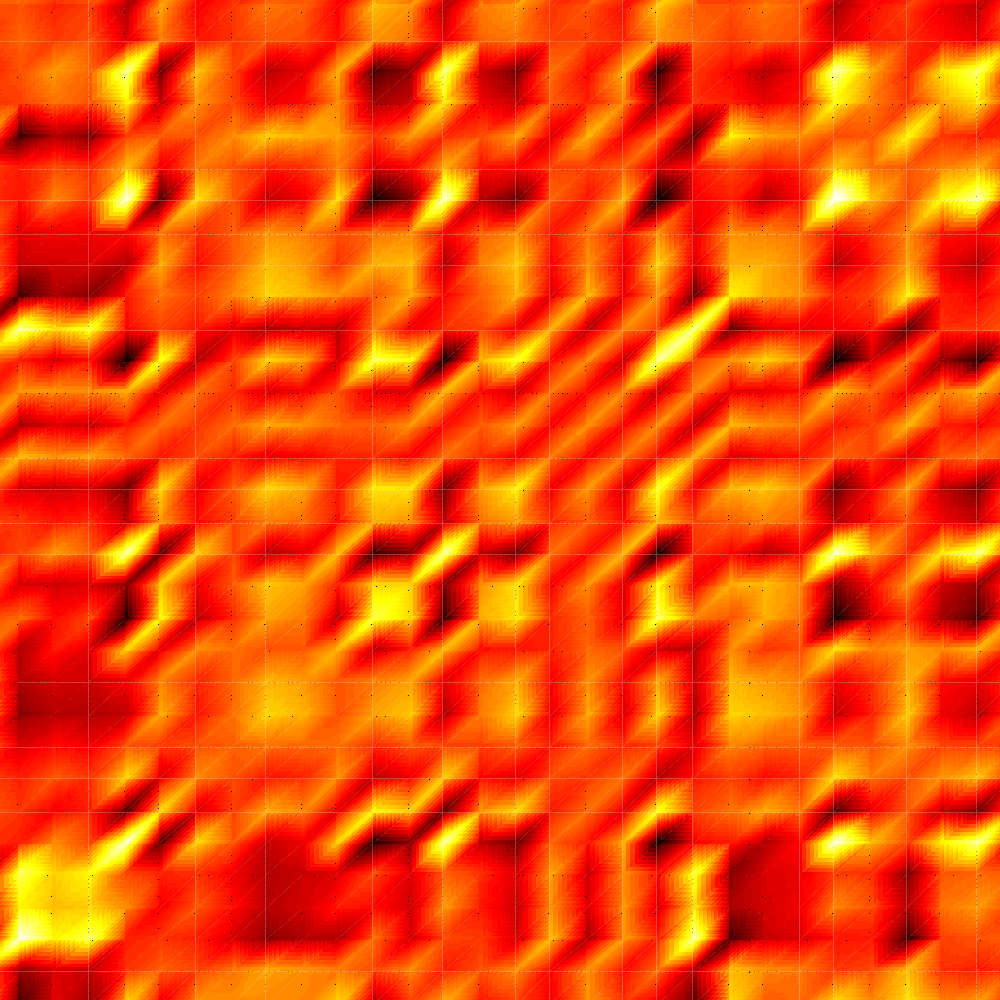}
     \end{subfigure}
     \hfill
     \begin{subfigure}[b]{0.3\textwidth}
         \centering
         \caption*{$\hat{\bSigma}_{\xi,|t_1-t_2|}(\bs_1,\bs_2)$, $t_1-t_2=1$}
         \includegraphics[width=.95\textwidth]{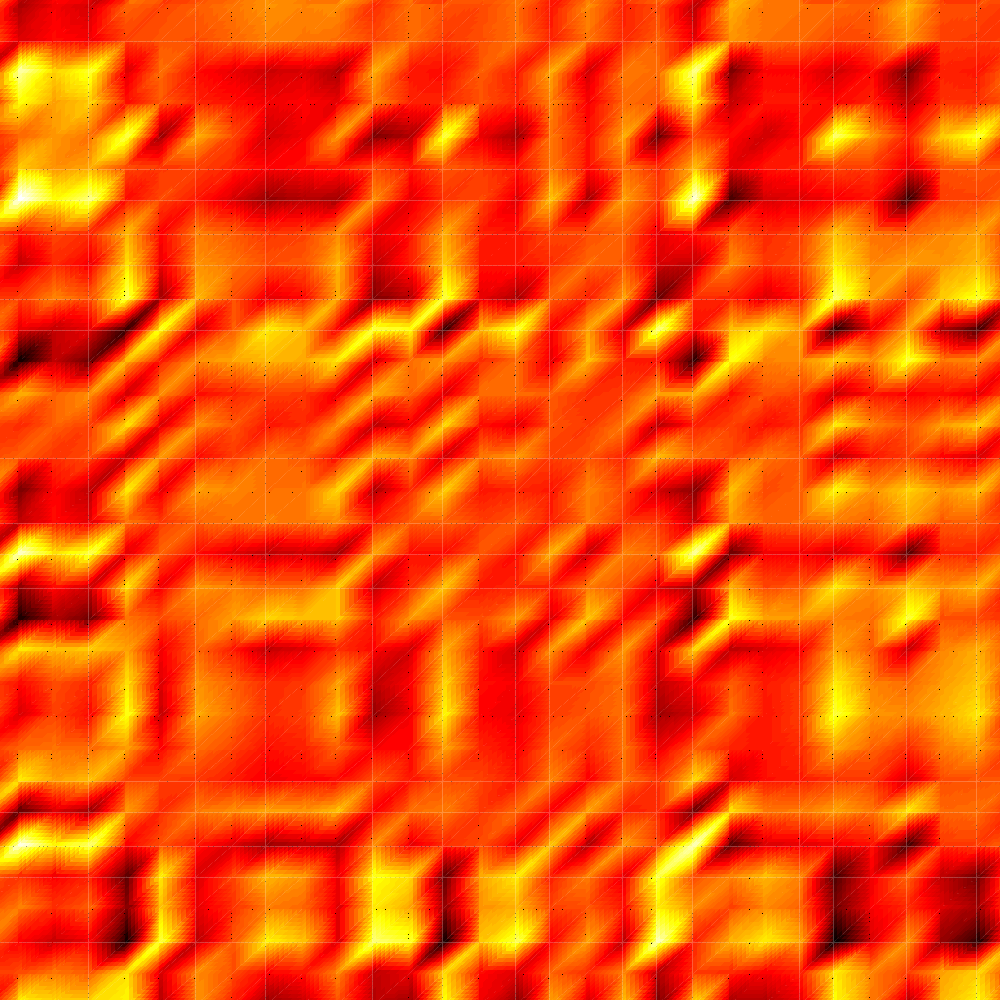}
     \end{subfigure}
     \hfill
     \begin{subfigure}[b]{0.3\textwidth}
         \centering
         \caption*{$\hat{\bSigma}_{\xi,|t_1-t_2|}(\bs_1,\bs_2)$, $t_1-t_2=2$}
         \includegraphics[width=.95\textwidth]{./figs/cov_xi_1_1}
     \end{subfigure}
\caption{The sample temporal covariance $\tilde{\bSigma}_{\xi,|t_1-t_2|}$ and estimated temporal covariance matrix $\hat{\bSigma}_{\xi,|t_1-t_2|}$. Here, the time lag $t_1-t_2=0,1,2$. $n=400$, $p=40$, and $T=240$. $\bs_1$ is randomly generated from $[-1,1]^2$. }
\label{fig:cov_real_xi.0}
\end{figure}

\begin{figure}[ht!]
	\centering
	\includegraphics[width=0.85\textwidth]{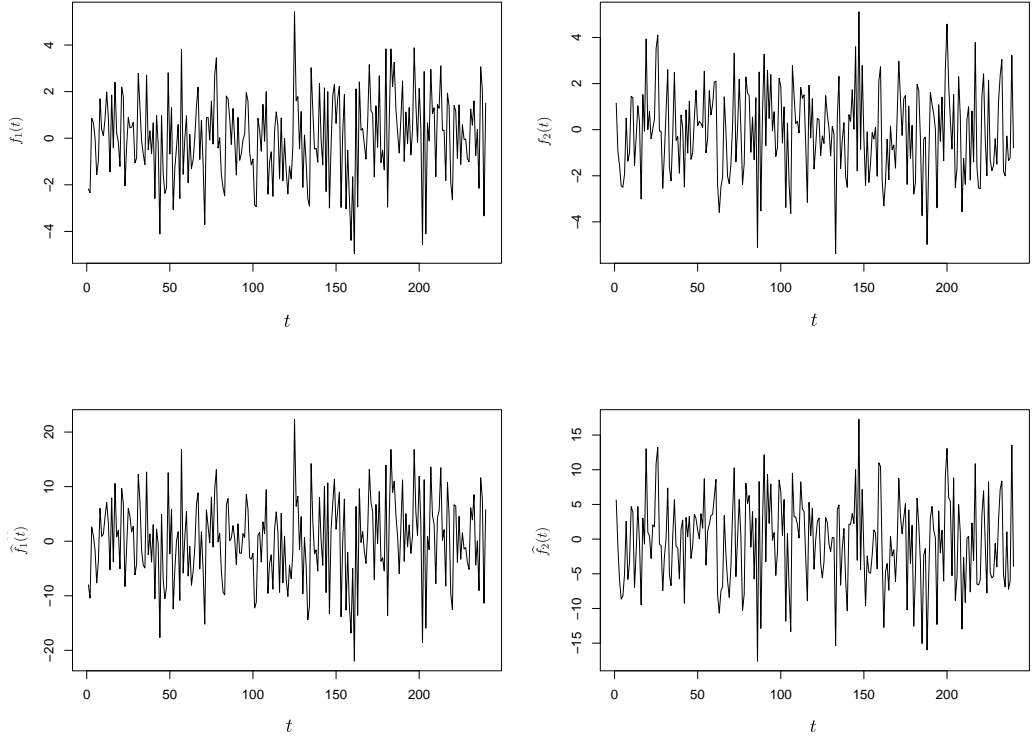}
	\caption{True factor $\bff_t(\bs)$ (on the top) and the estimated factor $\hat{\bff}_t(\bs)$ (on the bottom) by proposed method. $n=400$, $p=40$, and $T=240$. $\bs$ is randomly generated from $[-1,1]^2$}
	\label{fig:spdistf.0}
\end{figure}

\begin{figure}[ht!]
	\centering
	\includegraphics[width=0.85\textwidth]{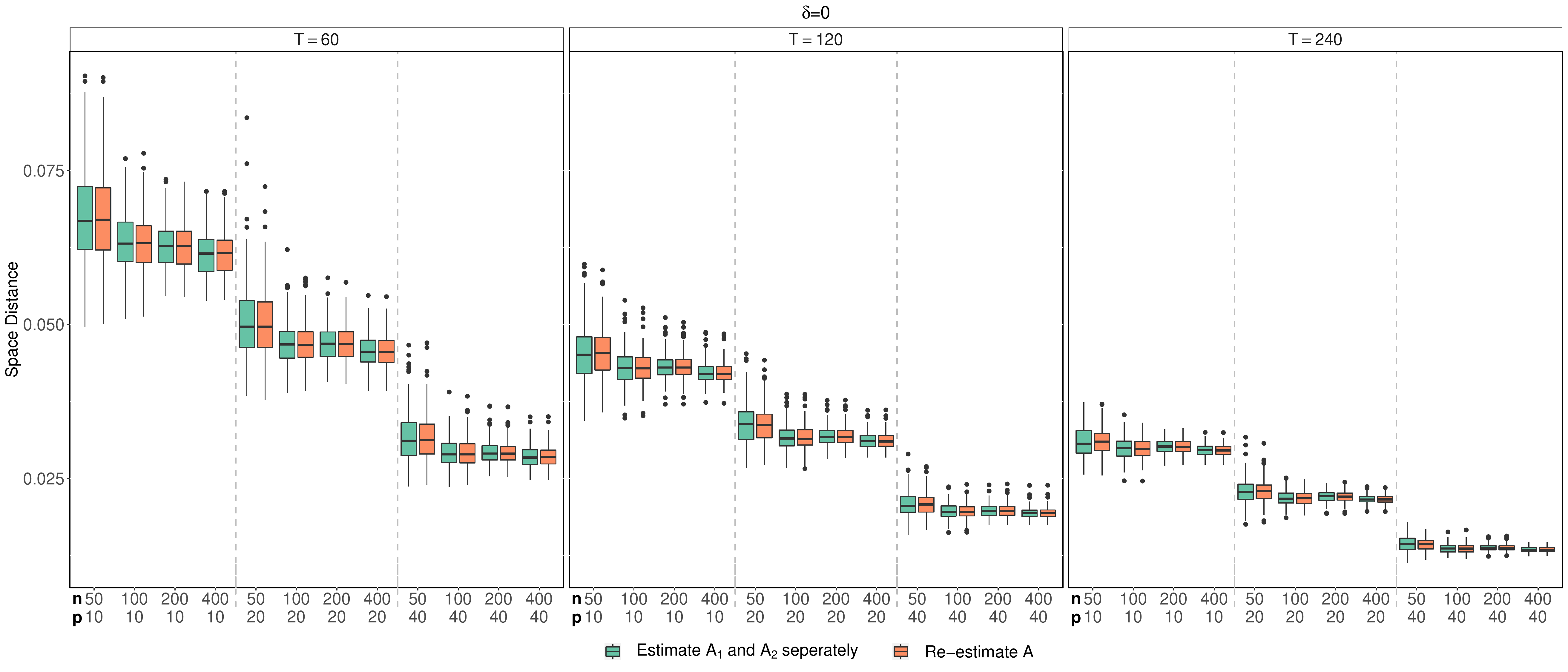}
	\caption{Box-plots of the estimation accuracy measured by $\mathcal{D}(\widehat{\bA}(\bs), \bA(\bs))$ for the case of orthogonal constraints. Gray boxes represent the average of $\mathcal{D}(\widehat{\bA}_1(\bs), \bA_1(\bs))$ and $\mathcal{D}(\widehat{\bA}_2(\bs), \bA_2(\bs))$. The results are based on $200$ iterations. See Table \ref{table:spdist_msd_table} in Appendix \ref{appendix:tableplots} for mean and standard deviations of the spatial distance. }
	\label{fig:spdistA.0}
\end{figure}

\begin{figure}[ht!]
	\centering
	\includegraphics[width=0.85\textwidth]{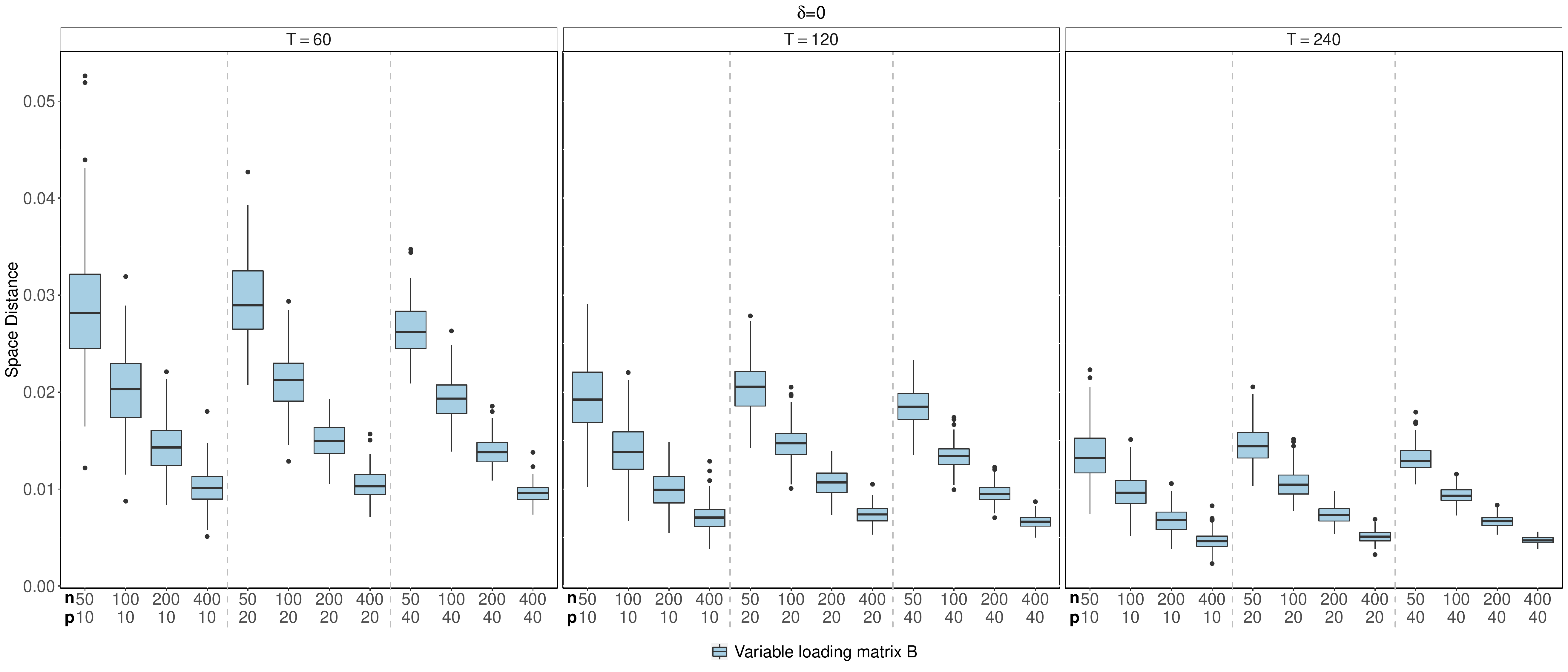}
	\caption{Box-plots of the estimation accuracy of variable loading matrix measured by $\mathcal{D}(\widehat{\bB}, \bB)$. The results are based on $200$ iterations. See Table \ref{table:spdist_msd_table} in Appendix \ref{appendix:tableplots} for mean and standard deviations of the spatial distance. }
	\label{fig:spdistB.0}
\end{figure}

\begin{figure}[ht!]
	\centering
	\includegraphics[width=0.85\textwidth]{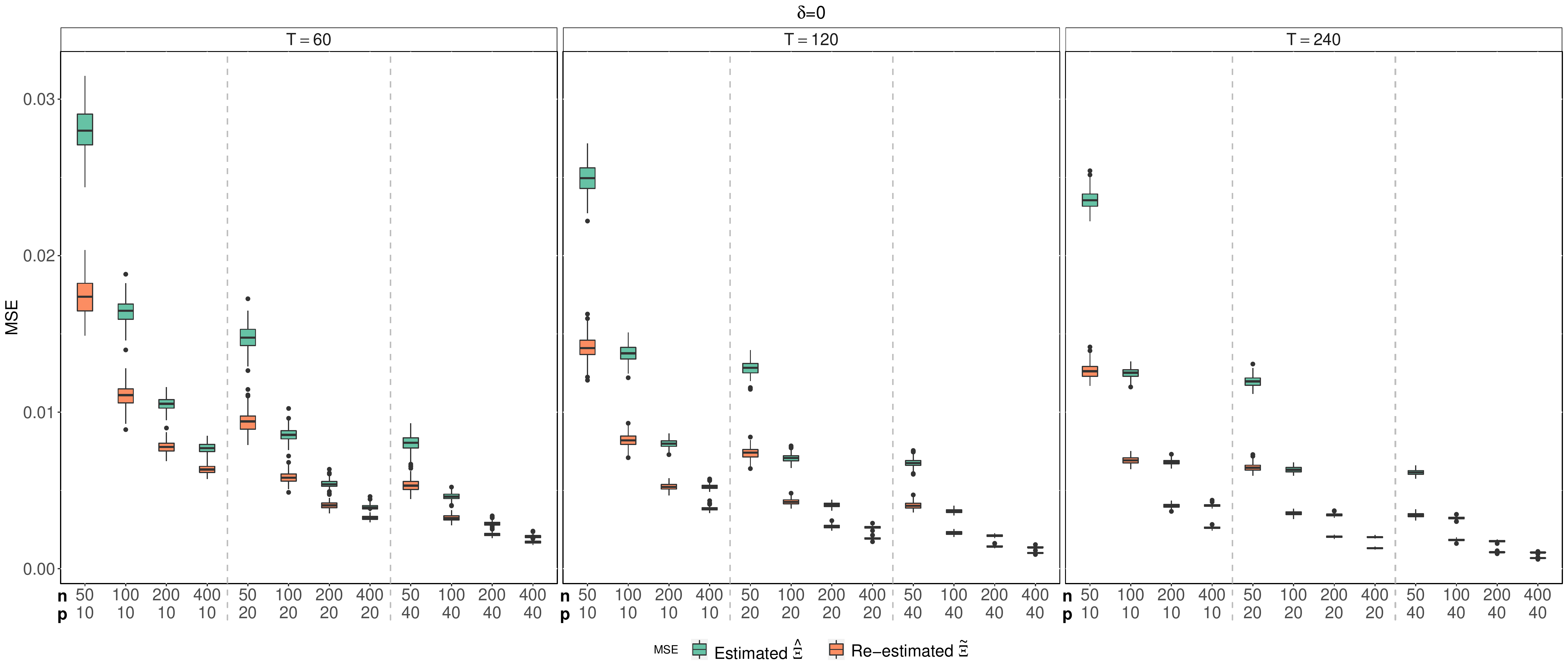}
	\caption{Box-plots of the estimation of signals MSE. Gray boxes represent the our procedure. The results are based on $200$ iterations. See Table \ref{table:spdist_msd_table} in Appendix \ref{appendix:tableplots} for mean and standard deviations of the MSE. }
	\label{fig:signalest.0}
\end{figure}

\begin{figure}[ht!]
	\centering
	\includegraphics[width=0.85\textwidth]{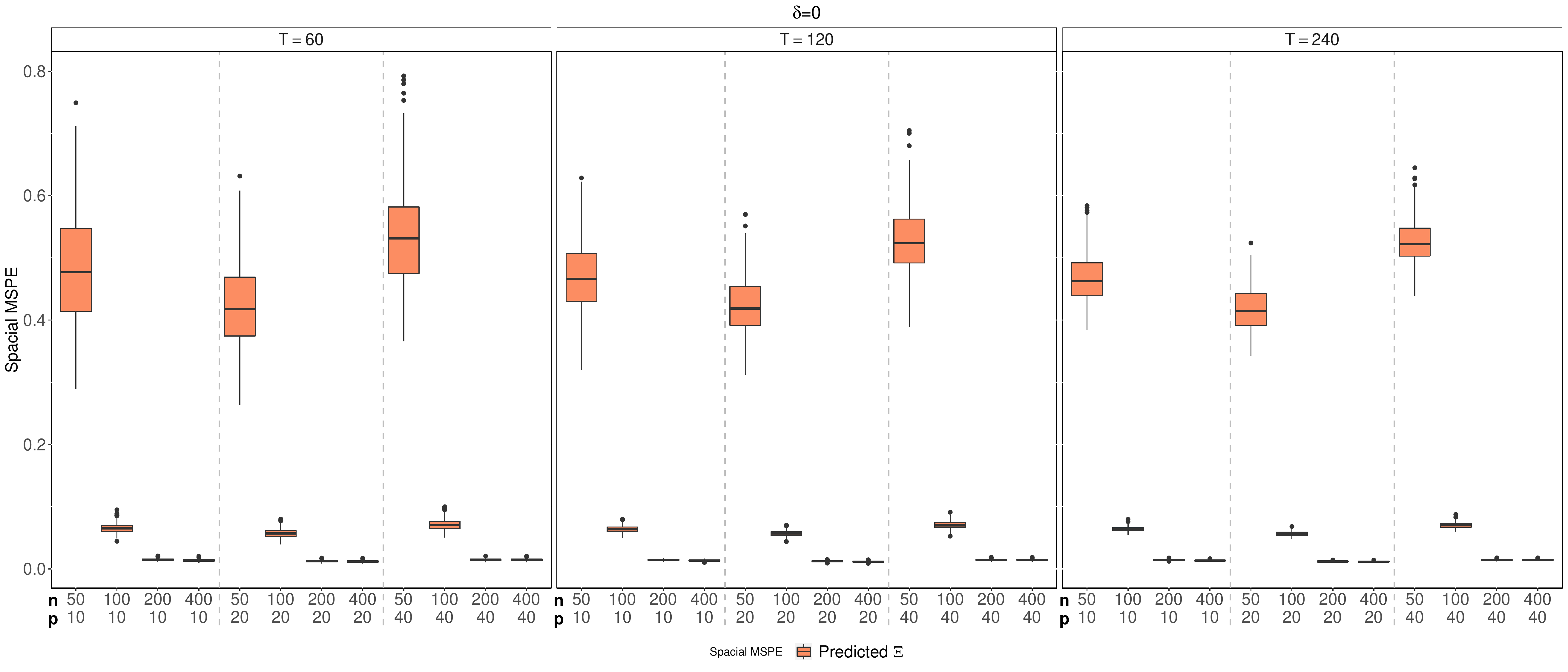}
	\caption{Box-plots of the spatial prediction measured by average MSPE for $50$ new locations. Colored boxes represent the our model. The results are based on $200$ iterations. See Table \ref{table:STprediction} in Appendix \ref{appendix:tableplots} for mean and standard deviations of the MSPE. }
	\label{fig:signalspaceprd.0}
\end{figure}

\begin{figure}[ht!]
	\centering
	\includegraphics[width=0.85\textwidth]{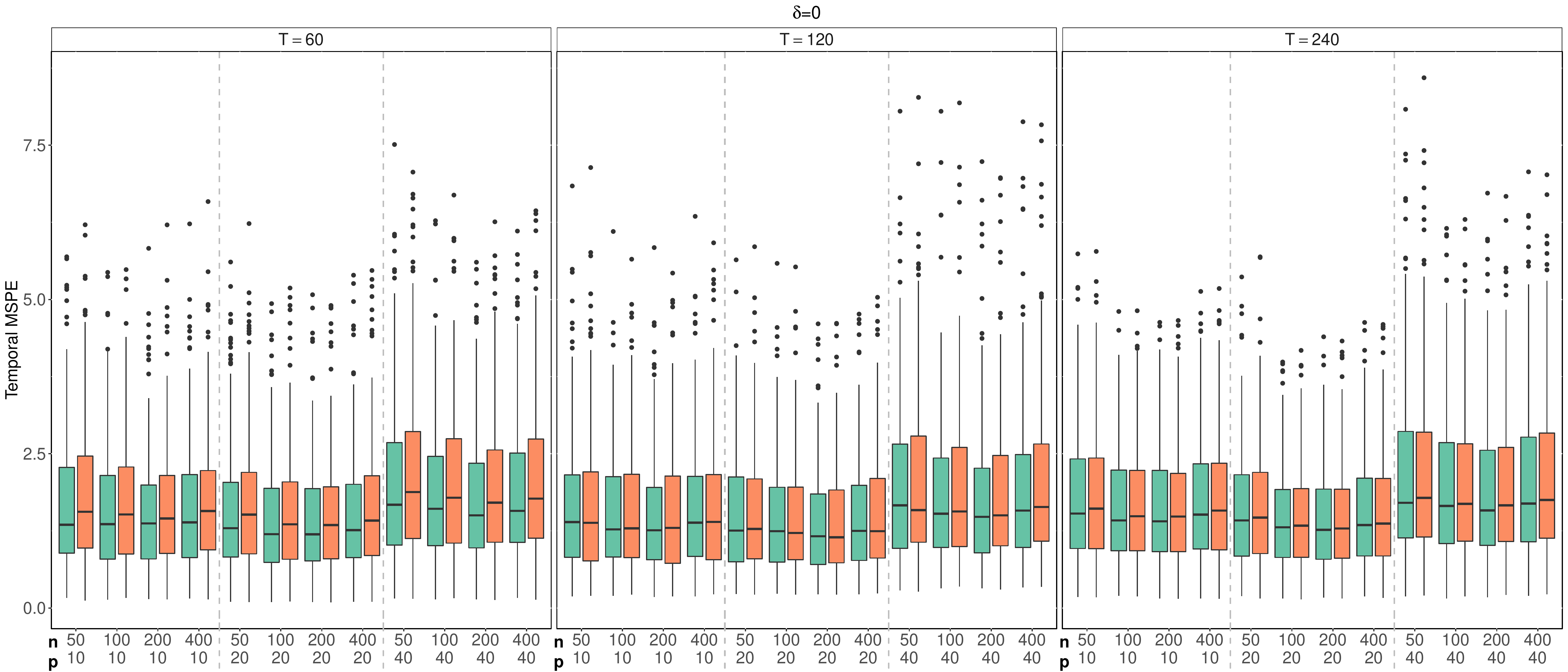}
	\caption{Box-plots of the one step ahead forecasting accuracy measured by MSPE. Gray boxes represent the MAR(1) model. The results are based on $200$ iterations. See Table \ref{table:STprediction} in Appendix \ref{appendix:tableplots} for mean and standard deviations of the MSPE. }
	\label{fig:signaltimeprd.0}
\end{figure}

Define the mean squared error of estimated signals $\hat{\bxi}$ as
\[
MSE(\hat{\bxi}) = \frac{1}{npT} \sum_{t=1}^{T} \sum_{i=1}^{n} \norm{ \hat{\bxi}_t(\bs_i) - \bxi_t(\bs_i) }^2_2.
\]
We compare the mean square error between first estimated $\hat{\bXi}_t$ defined in (\ref{eqn:signal_mat_est_sep}) and re-estimated $\tilde{\bXi}_t$ defined as
\[
\tilde{\bXi} = \begin{bmatrix} \tilde{\bXi}_1, \cdots, \tilde{\bXi}_T \end{bmatrix} = \tilde{\bA} \tilde{\bX} \hat{\bB}'.
\]
The box plots of $MSE(\hat{\bxi})$ and $MSE(\tilde{\bxi})$ are in Figure \ref{fig:signalspaceprd.0}. Re-estimated provides much more accurate estimate for $\bxi_t(\bs_j)$ than $\tilde{\bxi}_t(\bs_j)$ does.

To demonstrate the performance of spatial prediction, we generate data at a set $\calS_0$ of $50$ new locations randomly sampled from $\calU[-1,1]^2$. For each $t = 1, \ldots, T$, we calculate the spatial prediction $\hat{\by}_t(\cdot) = \hat{\bxi}_t(\cdot)$ defined in (\ref{eqn:pred_xi_s0}) for each location in $\calS_0$. The mean squared spatial prediction error is calculated as
\[
MSPE(\hat{\by}) = \frac{1}{50 p T} \sum_{t=1}^{T} \sum_{s_0 \in \calS_0} \norm{ \hat{\by}_t(\bs_0) - \bxi_t(\bs_0) }^2_2.
\]

To demonstrate the performance of temporal forecasting, we generate $\bX_{T+h}$ according to the matrix time series (\ref{eqn:mar1}) for $h=1,2$ and compute both the one-step-ahead and two-step-ahead predictions at time $T$. The mean square temporal prediction error is computed as
\[
MSPE(\hat{\by}_{T+h}) = \frac{1}{n p} \sum_{j=1}^{n} \norm{ \hat{\by}_{T+h}(\bs_j) - \bxi_{T+h}(\cdot) }^2_2.
\]

Figure \ref{fig:signalspaceprd.0} presents box-plots of the spatial prediction measured by average MSPE for $50$ new locations. The results are based on $200$ iterations. Figure \ref{fig:signaltimeprd.0} compares the MSPEs using matrix time series MAR(1) and vectorized time series VAR(1) estimates.

The means and standard errors of the MSPEs from 200 simulations for each model setting are reported in Table \ref{table:STprediction} in Appendix \ref{appendix:tableplots}. It also reports the means and standard errors of the MSPEs using matrix time series MAR(1) and vectorized time series VAR(1) estimates.

\section{Real Data Applications} \label{sec:appl}

In this section, we apply the proposed method to the Comprehensive Climate Data Set (CCDS) -- a collection of climate records of North America. The data set was compiled from five federal agencies sources by \cite{lozano2009spatial}\footnote{\url{http://www-bcf.usc.edu/~liu32/data/NA-1990-2002-Monthly.csv}}. It contains monthly observations of 17 climate variables spanning from 1990 to 2001 on a $2.5 \times 2.5$ degree grid for latitudes in $(30.475, 50.475)$, and longitudes in $(-119.75, -79.75)$. The total number of observation locations is 125 and the whole time series spans from January, 1991 to December, 2002. We use a subset of the original data set because of the data quality. It contains measurements of 16 variables at all the locations range from January, 1992 to December, 2002. Thus, the dimensions our our data set are 125 (locations) $\times$ 16 (variables) $\times$ 132 (time points).  Table \ref{table:CCDS_varlist} lists the variables used in our analysis. Detailed information about data is given in \cite{lozano2009spatial}.

\begin{table}[htpb!]
\centering
\caption{Variables and data sources in the Comprehensive Climate Data Set (CCDS)}
\label{table:CCDS_varlist}
\resizebox{0.7\textwidth}{!}{%
\begin{tabular}{l|c|l|c}
\hline
Variables (Short name) & Variable group & \multicolumn{1}{c|}{Type} & Source \\ \hline
Methane (CH4) & $CH_4$ & \multirow{4}{*}{Greenhouse Gases} & \multirow{4}{*}{NOAA} \\
Carbon-Dioxide (CO2) & $CO_2$ &  &  \\
Hydrogen (H2) & $H_2$ &  &  \\
Carbon-Monoxide (CO) & $CO$ &  &  \\ \hline
Temperature (TMP) & TMP & \multirow{8}{*}{Climate} & \multirow{8}{*}{CRU} \\
Temp Min (TMN) & TMP &  &  \\
Temp Max (TMX) & TMP &  &  \\
Precipitation (PRE) & PRE &  &  \\
Vapor (VAP) & VAP &  &  \\
Cloud Cover (CLD) & CLD &  &  \\
Wet Days (WET) & WET &  &  \\
Frost Days (FRS) & FRS &  &  \\ \hline
Global Horizontal (GLO) & SOL & \multirow{4}{*}{Solar Radiation} & \multirow{4}{*}{NCDC} \\
Direct Normal (DIR) & SOL &  &  \\
Global Extraterrestrial (ETR) & SOL &  &  \\
Direct Extraterrestrial (ETRN) & SOL &  &  \\ \hline
\end{tabular}%
}
\end{table}

We first remove seasonal patterns in this data set by taking difference between the same month in consequent years.
We then centralize and standardize each series to have zero mean and unit variance before further investigation.

To estimate the latent dimensions, we combine the method of the scree plots and the eigen-ratio method.
Figure \ref{fig:latent-dimension} shows the scree plots and the eigen-ratio plots of the latent spatial and variable dimensions.
Scree plots show that, in order to achieve $90\%$ variance, we need to have latent spatial dimension $\hat d = 6$ and latent variable dimension $\hat r = 6$.
Eigen-ratio \eqref{eqn:eigen-ratio} estimates latent spatial dimension $\hat d = 12$ and latent variable dimension $\hat r = 4$.
Due to the dominance of the largest factors and weak signal in real data, the estimate by \eqref{eqn:eigen-ratio} tends to be less useful than the one given by the scree plot.
In the following, we choose $(\hat d, \hat r) = (6,6)$ as the latent dimensions.

\begin{figure}[htpb!]
     \centering
     \begin{subfigure}[b]{0.7\textwidth}
         \centering
         \caption{Scree plots}
         \includegraphics[width=\textwidth]{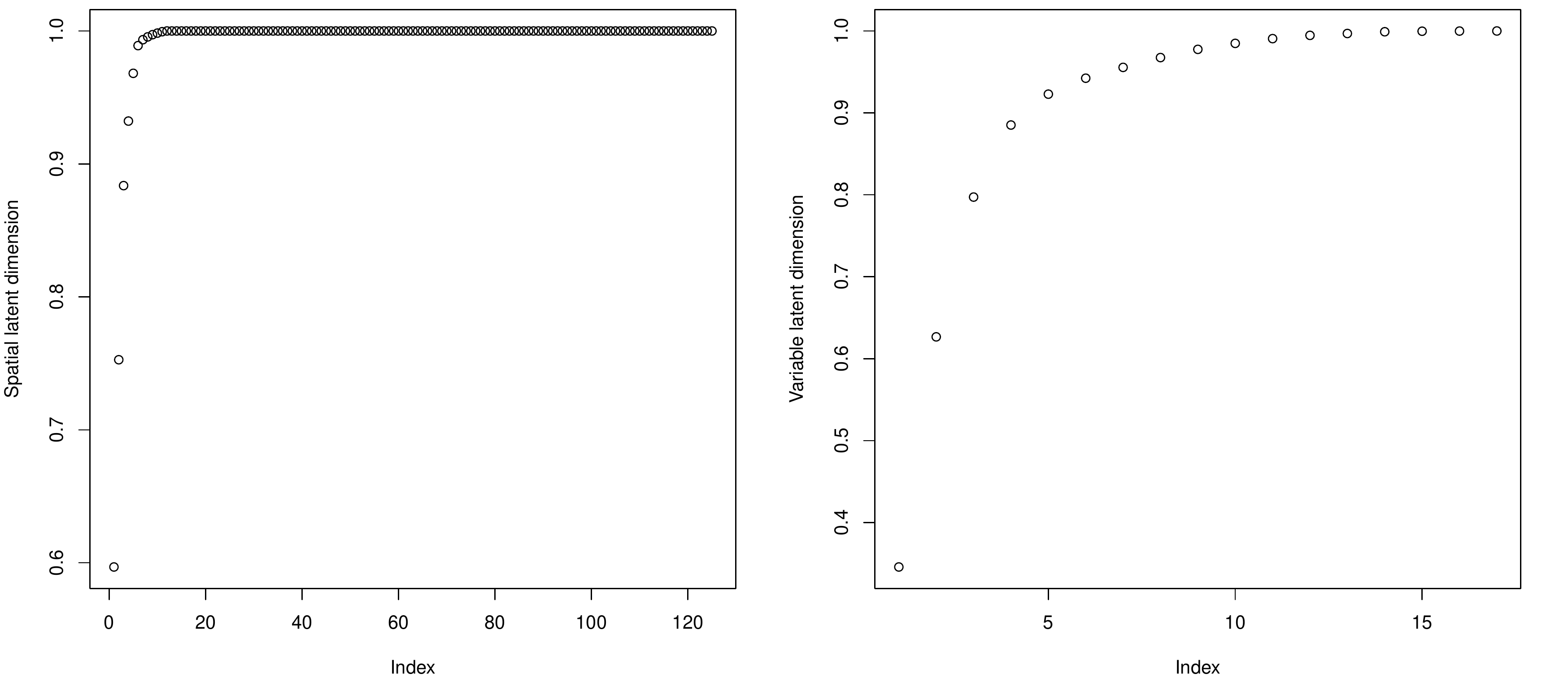}
         \label{fig:scree-plots}
     \end{subfigure}

      \begin{subfigure}[b]{0.7\textwidth}
         \centering
         \caption{Eigen-ratio plots}
         \includegraphics[width=\textwidth]{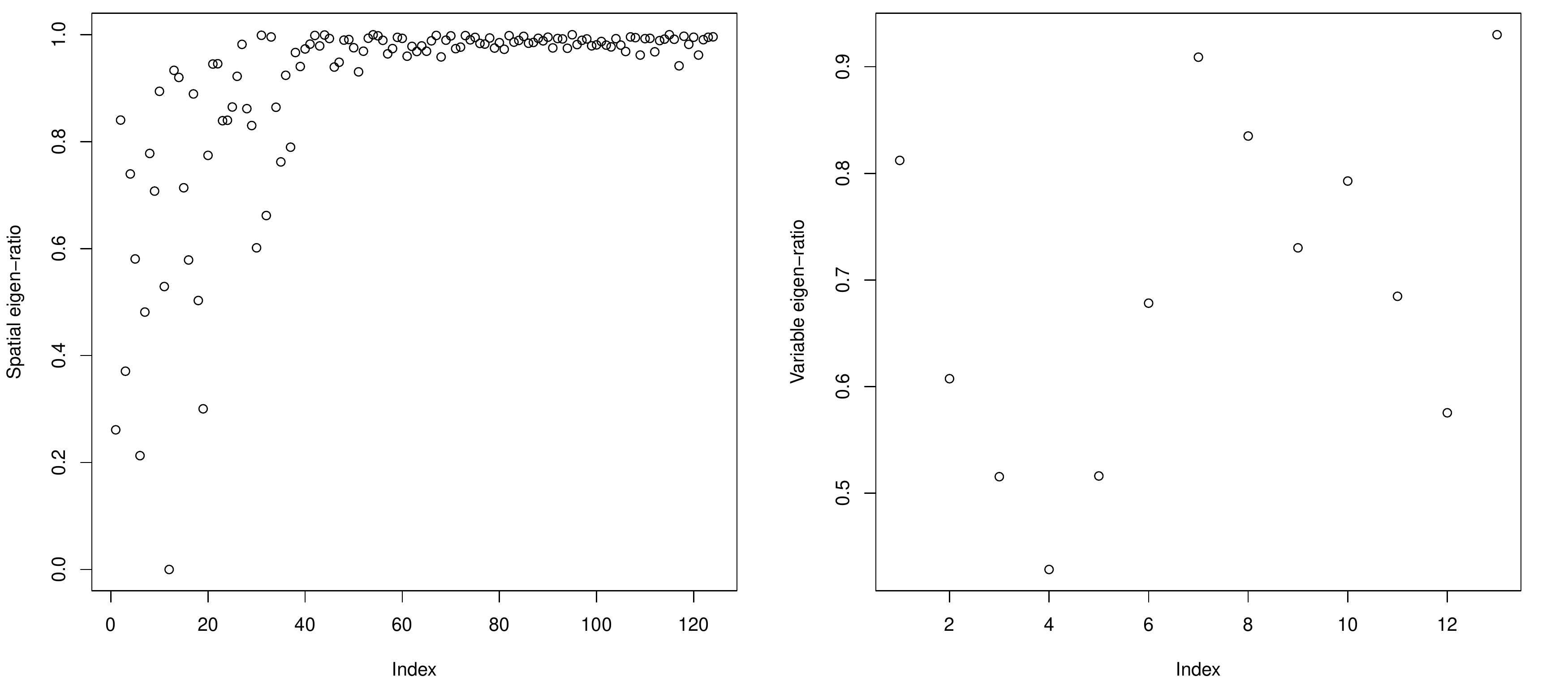}
         \label{fig:eigen-ratio-plots}
      \end{subfigure}
      \vspace{-3ex}
      \caption{Latent dimensions}
      \label{fig:latent-dimension}
\end{figure}

For kriging in space, we compare the performance kriging with kernel smoothing and prediction with functional $\bA(\bs)$.
We randomly pick a portion of locations (10\%, 25\% and 33\% of all locations) and eliminate the measurements of all variables over the whole time span.
Then, we produce the estimates for all variables of each timestamp. We repeat the procedure for 100 times.
Table \ref{table:compare_number_training_stations} report the average prediction RMSEs for all timestamps and 10 random sets of missing locations.
It shows that the prediction by the proposed prediction with functional estimation $\bA(\bs)$ performs much better than kriging with kernel smoothing. 

We also compare the sample spatial-temporal covariance of the real data $\by_t$, and estimated spatial-temporal covariance of $\hat\xi_t$ with the reduced rank structure in the proposed model with time lag $t_1-t_2=0,1,2$ and two randomly selected location $\bs_1$ and $\bs_2$ in Figure \ref{fig:cov_xi_space.0}. In this data, we only observe the $\by_t$, which include the noise $\bepsilon_t$.
It shows that the co-variance structure of the real data is largely preserved with the reduced rank approximation even when the dimension reduction is significant. 

\begin{table}[htpb!]
\centering
\caption{Means of standard errors of MSPE by the proposed method for CCDS dataset. Results are based 100 simulations.}
\label{table:compare_number_training_stations}
\begin{tabular}{c|ccc}
\hline
\% Testing Sites & 33\% & 25\% & 10\% \\
\# Training / Testing Sites & 84 / 41 & 94 / 31 & 113 / 12 \\ \hline
Kriging with kernel smoothing & $ 0.580 \; (0.028) $  & $ 0.578 \; (0.027) $ & $ 0.572 \; (0.035) $ \\
Prediction with functional $\bA(\bs)$ & $ 0.314 \; (0.011) $  & $ 0.312 \; (0.009) $ & $ 0.309 \; (0.013) $ \\ \hline
\end{tabular}
\end{table}

\begin{figure}[ht!]
     \centering
     \begin{subfigure}[b]{0.3\textwidth}
         \centering
         \caption*{$\tilde{\bSigma}_{y,|t_1-t_2|}(\bs_1,\bs_2)$, $t_1-t_2=0$}
         \includegraphics[width=.95\textwidth]{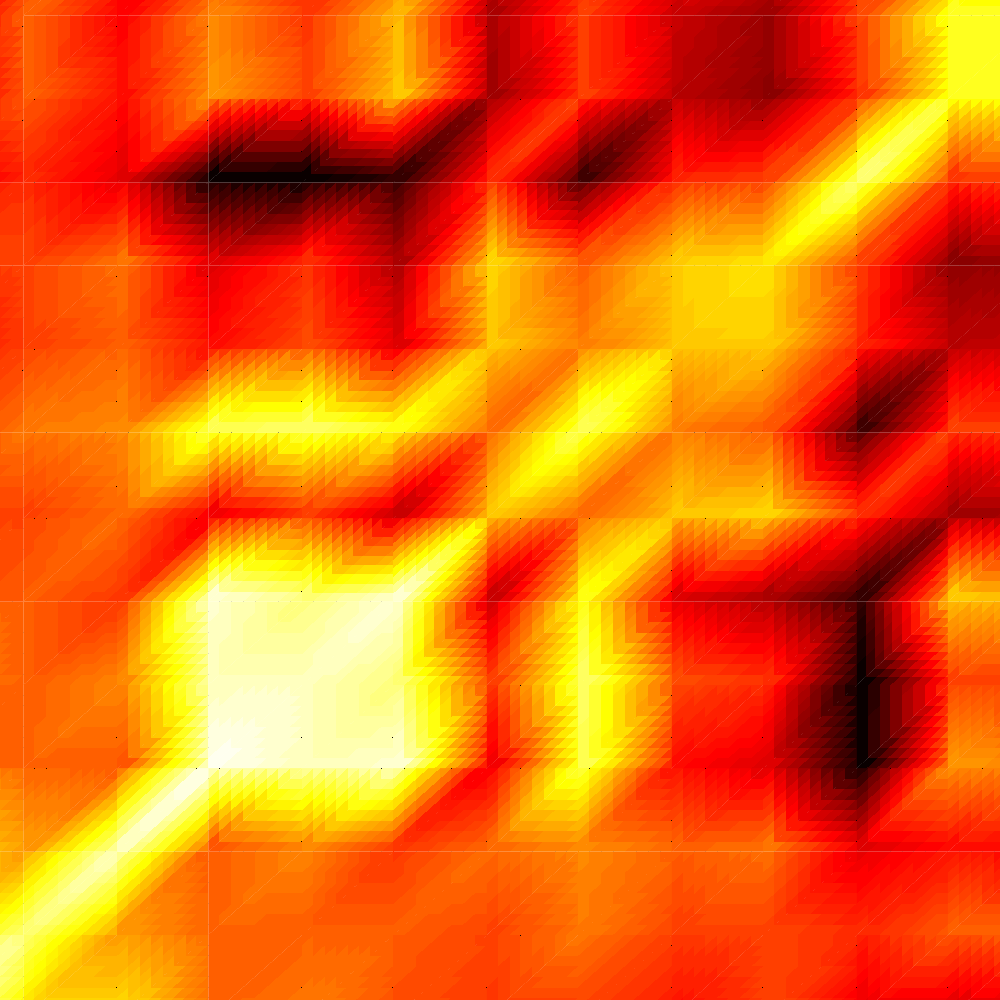}
     \end{subfigure}
     \hfill
     \begin{subfigure}[b]{0.3\textwidth}
         \centering
         \caption*{$\tilde{\bSigma}_{y,|t_1-t_2|}(\bs_1,\bs_2)$, $t_1-t_2=1$}
         \includegraphics[width=.95\textwidth]{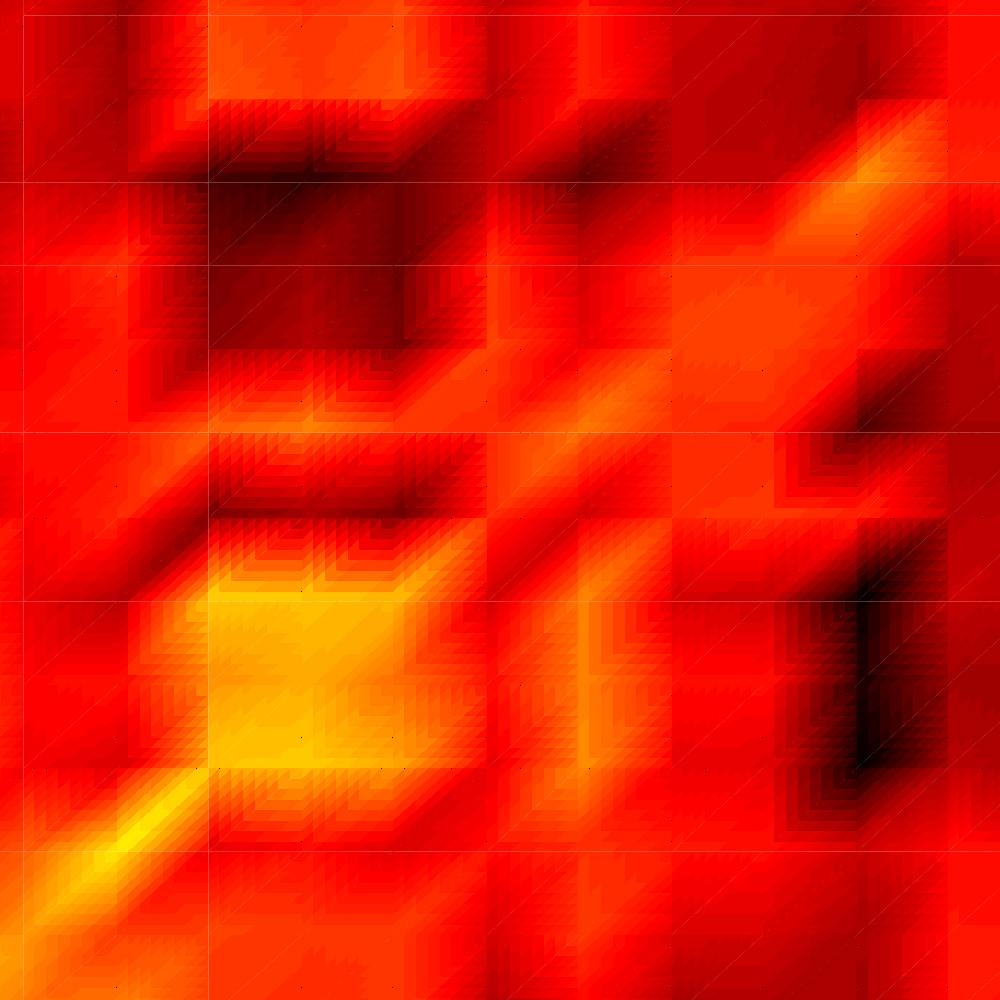}
     \end{subfigure}
     \hfill
     \begin{subfigure}[b]{0.3\textwidth}
         \centering
         \caption*{$\tilde{\bSigma}_{y,|t_1-t_2|}(\bs_1,\bs_2)$, $t_1-t_2=2$}
         \includegraphics[width=.95\textwidth]{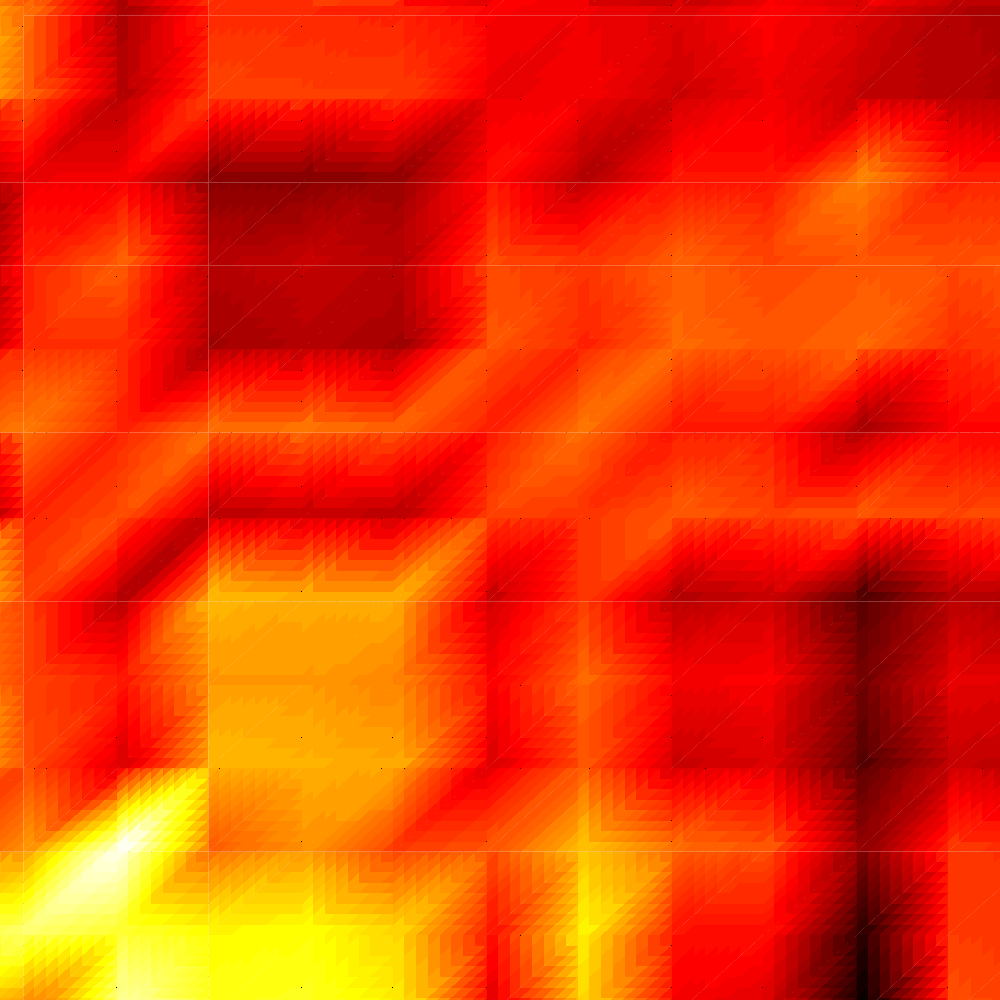}
     \end{subfigure}
     
     \begin{subfigure}[b]{0.3\textwidth}
         \centering
         \caption*{$\hat{\bSigma}_{\xi,|t_1-t_2|}(\bs_1,\bs_2)$, $t_1-t_2=0$}
         \includegraphics[width=.95\textwidth]{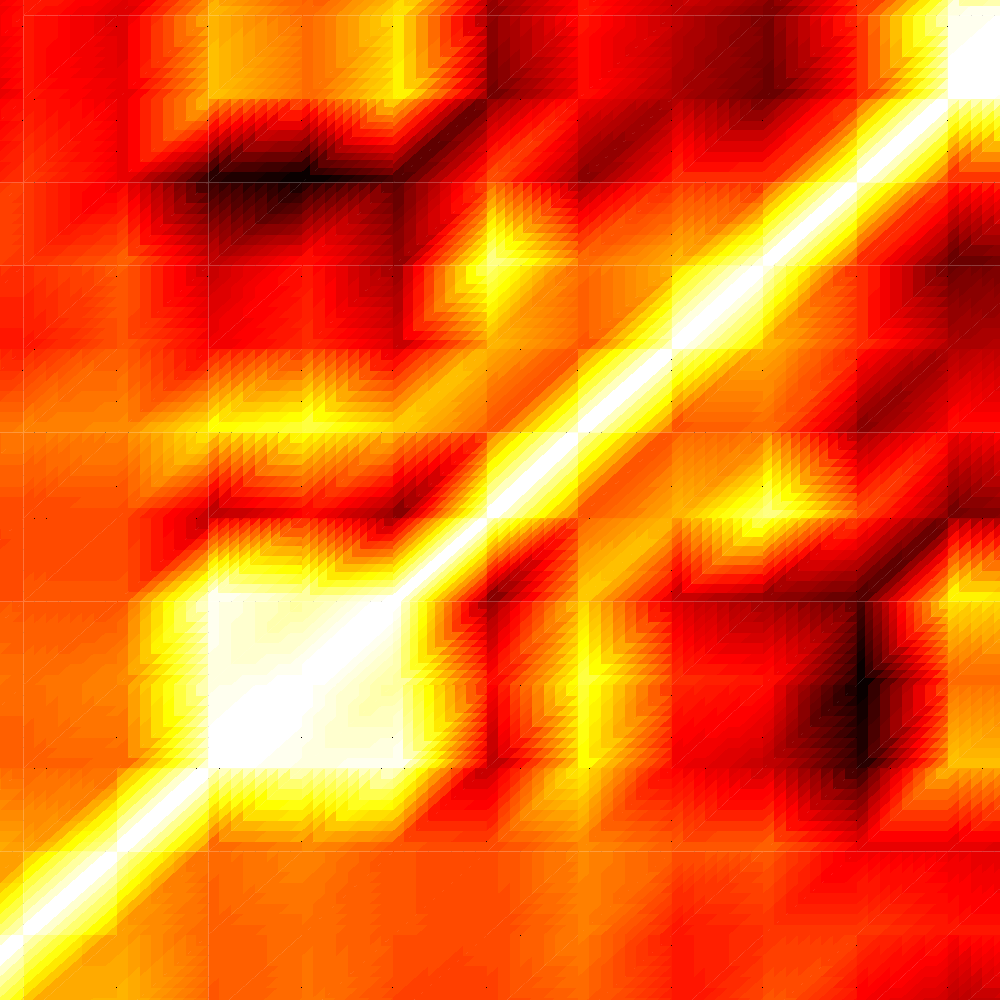}
     \end{subfigure}
     \hfill
     \begin{subfigure}[b]{0.3\textwidth}
         \centering
         \caption*{$\hat{\bSigma}_{\xi,|t_1-t_2|}(\bs_1,\bs_2)$, $t_1-t_2=1$}
         \includegraphics[width=.95\textwidth]{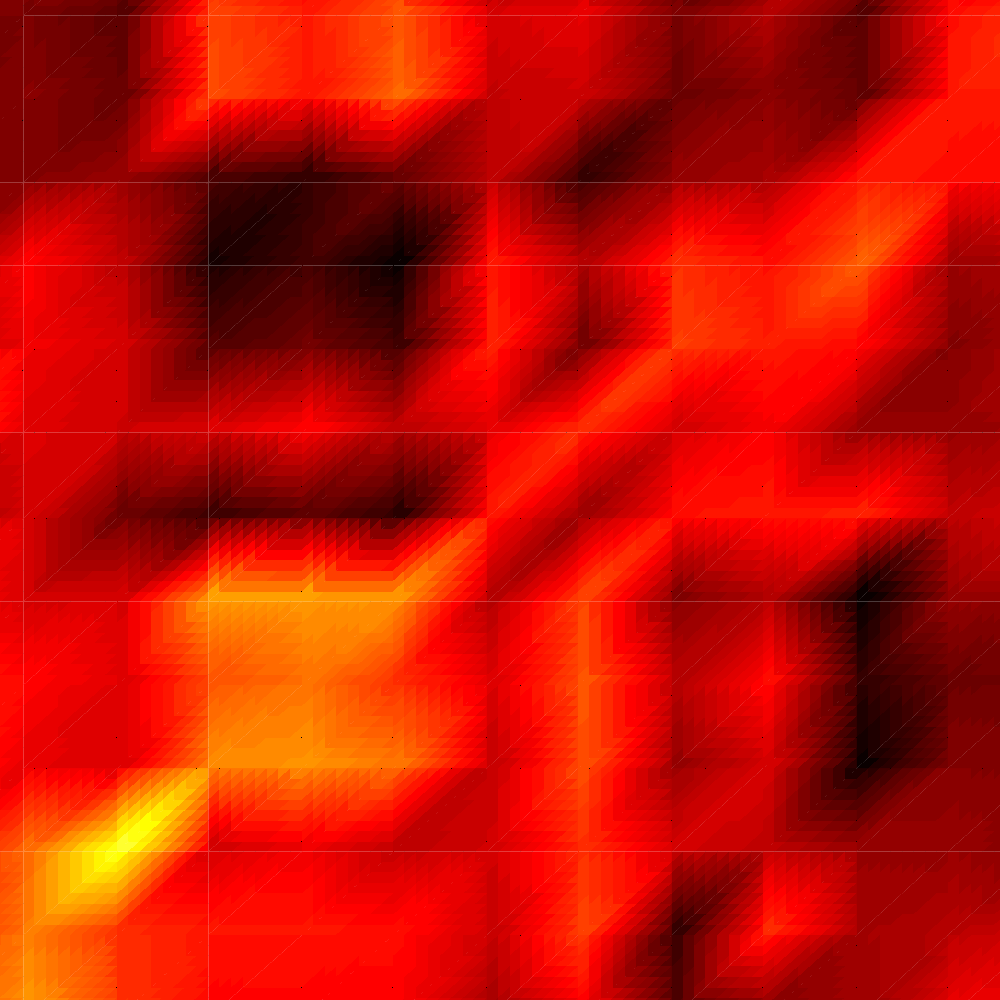}
     \end{subfigure}
     \hfill
     \begin{subfigure}[b]{0.3\textwidth}
         \centering
         \caption*{$\hat{\bSigma}_{\xi,|t_1-t_2|}(\bs_1,\bs_2)$, $t_1-t_2=2$}
         \includegraphics[width=.95\textwidth]{./figs/cov_xi_space_1_1}
     \end{subfigure}
\caption{Heat map of the spacial-temporal covariances of $\by$ and $\hat\bxi$ on the testing set.
Top three heat maps corresponds to the sample spacial-temporal covariance of $\by_t(\bs)$.
Bottom three heat maps corresponds to the estimated spacial-temporal covariance of $\hat\bxi_t(\bs)$ with the reduced rank structure. 
The time lags are chosen $t_1-t_2=0,1,2$. 
The testing sites is 10\%. $\bs_1$ and $\bs_2$ are randomly chosen in 12 test points.
The covariances based on the low rank model are very close to the sample covariances of the full data.}
\label{fig:cov_xi_space.0}
\end{figure}

For temporal forecasting, we are interested in forecasting values in year 2001 and 2002. We experiment with two different length of training data -- 5 and 9 years -- respectively. For each setting, we estimate the loading matrices and factor matrices using the training data and make 1-step and 2-step prediction. We move forward with one month for both training and testing data and repeat the process until we reach 2002-12. For example with 5 training years, we start with estimation with 5 years training data from 1996-01 to 2000-12 and make 1-step prediction on 2001-01 and 2-step prediction on 2001-02. Then we move forward with one month -- estimation with training data from 1996-02 to 2001-01 and prediction on the month 2001-02 and 2001-03. We repeat this process until the last estimation with 1998-11 to 2002-10 data and prediction on 2002-11 and 2002-12. So in total we have 23 predictions for 1-step and 2-step forecasting each for a given length of training set. With latent matrix time series, we predict each individual time series using \textit{auto.arima} and \textit{forecast} functions in the R \textit{forecast} package. This is feasible because the latent factor matrix is low dimensional. With original matrix time series of $125 \times 16$ dimension, the computational cost is much higher. Table \ref{table:forecasting} reports the mean and standard deviation of the mean squared prediction errors.
As shown by the results, temporal prediction is much harder than spatial prediction. 
\begin{table}[htpb!]
\centering
\caption{Means (standard deviation) of MSPE by the proposed method for CCDS dataset.}
\label{table:forecasting}
\begin{tabular}{c|c|c}
\hline
Training Years & 5 & 9 \\ \hline
1-step MSPE & $0.633 \, (0.181)$ & $ 0.574 \, (0.141)$  \\
2-step MSPE & $0.682 \, (0.225)$ & $0.623 \, (0.190)$ \\
Time (min) & $ 0.56 \, (0.04)$ & $1.55 \, (0.20)$ \\ \hline
\end{tabular}%
\end{table}

\newpage
\section{Summary} \label{sec:summ}

In this paper, we study the problem of large-scale multivariate spatial-temporal data analysis with a focus on dimension reduction and spatial/temporal forecasting. 
We propose a new class of multivariate spatial-temporal models that model spatial, temporal and multivariate dependencies simultaneously. 
This is made possible by an innovative combination of the multivariate factor analysis with the method of empirical orthogonal functions.
For estimation, we assembled the observations from discrete spatial locations as a time series of matrices whose rows and columns correspond to sampling sites and variables, respectively.
The matrix structure of observations is well preserved through the matrix factor model reformulation, while further incorporating the functional structure of the spatial process and dynamics of the latent matrix factor. 
We proposed methods of prediction over space and time based on the estimated latent structure. 
We established theoretical properties of the estimators and predictors. 
We validate the correctness and efficiency of our proposed method on both the synthetic and real application data sets.

For future work, we are interested in incorporating time-variant loading matrices to deal with possible structural changes. 
To improve the performance of spatial prediction, it is of great interest to investigate different ways to include spatial variograms. 
Since we use a two-step method to estimate the loading functions, possibly ways to estimate loading functions directly in one-step would also be an interesting direction for future research.

%
%

\clearpage
\bibliographystyle{\mybibsty}
\bibliography{\mybib}

%
%
\clearpage
\begin{appendices}

\section{Proofs}
\subsection{Factor loadings}

We start by defining some population covariance as

\begin{equation*}
\bOmega^A_{s_1 s_2, ij}=\frac{1}{T}\sum_{t=1}^T\Cov{ (\bA_1 \bX_t \bb_{i\cdot}, \bA_2\bX_t \bb_{j\cdot}) }, \qquad \bOmega^A_{s_1 e_2, ij}=\frac{1}{T}\sum_{t=1}^T\Cov{ (\be_{1t,\cdot i}, \be_{2t,\cdot j}) },
\end{equation*}

\begin{equation*}
\bOmega^A_{s_1 e_2, ij}=\frac{1}{T}\sum_{t=1}^T\Cov{( \bA_1 \bX_t \bb_{i\cdot}, \be_{2t,\cdot j}) }, \qquad \bOmega^A_{e_1 s_2, ij}=\frac{1}{T}\sum_{t=1}^T\Cov{( \be_{1t,\cdot i}, \bA_2\bX_t \bb_{j\cdot}) },
\end{equation*}

\begin{equation*}
\bOmega^B_{s_1 s_2, ij}=\frac{1}{T}\sum_{t=1}^T\Cov{( \ba_{2, j\cdot}\bX_t\bB', \ba_{1, i\cdot}\bX_t\bB') }, \qquad \bOmega^B_{e_1 e_2, ij}=\frac{1}{T}\sum_{t=1}^T\Cov{ (\be_{1t,i\cdot}, \be_{2t,j\cdot} )},
\end{equation*}

\begin{equation*}
\bOmega^B_{s_1 e_2, ij}=\frac{1}{T}\sum_{t=1}^T\Cov{( \ba_{1, i\cdot}\bX_t\bB',  \be_{2t,j\cdot}) }, \qquad \bOmega^B_{e_1 s_2, ij}=\frac{1}{T}\sum_{t=1}^T\Cov{ (\be_{1t,i\cdot}, \ba_{2, j\cdot}\bX_t^*\bB' )},
\end{equation*}

and their sample versions

\begin{equation*}
\widehat{\bOmega}^A_{s_1 s_2, ij}=\frac{1}{T}\sum_{t=1}^T\bA_1 \bX_t \bb_{i\cdot}\left(\bA_2\bX_t\bb_{j\cdot}\right)', \qquad \widehat{\bOmega}^A_{e_1 e_2, ij}=\frac{1}{T}\sum_{t=1}^T\be_{t,\cdot i}\be'_{t,\cdot j},
\end{equation*}

\begin{equation*}
\widehat{\bOmega}^A_{s_1 e_2, ij}=\frac{1}{T}\sum_{t=1}^T\bA_1 \bX_t \bb_{i\cdot}\be'_{t,\cdot j}, \qquad \widehat{\bOmega}^A_{e_1 s_2, ij}=\frac{1}{T}\sum_{t=1}^T\be_{t,\cdot i}\left(\bA_2\bX_t\bb_{j\cdot}\right)',
\end{equation*}

\begin{equation*}
\widehat{\bOmega}^B_{s_1 s_2, ij}=\frac{1}{T}\sum_{t=1}^T\left(\ba_{1, i\cdot}\bX_t\bB'\right)'\ba_{2, j\cdot}\bX_t\bB', \qquad \widehat{\bOmega}^B_{e_1 e_2, ij}=\frac{1}{T}\sum_{t=1}^T\be'_{1t,i\cdot}\be_{2t,j\cdot},
\end{equation*}

\begin{equation*}
\widehat{\bOmega}^B_{s_1 e_2, ij}=\frac{1}{T}\sum_{t=1}^T\left(\ba_{1, i\cdot}\bX_t\bB'\right)'\be_{2t,j\cdot}, \qquad \widehat{\bOmega}^B_{e_1 e_2, ij}=\frac{1}{T}\sum_{t=1}^T\be'_{1t,i\cdot}\ba_{2, j\cdot}\bX_t\bB'.
\end{equation*}

\begin{lemma}  \label{lemma:entrywise_conv_rate_cov_vecXt}
Let $X_{t, ij}$ denote the $ij$-th entry of $\bX_t$. Under Condition \ref{cond:vecXt_alpha_mixing} and \ref{cond:Xt_cov_fullrank_bounded}, for any $i, k = 1, \ldots, d$ and $j, l = 1, \cdots, r$, we have
\begin{equation}
\left|\frac{1}{T} \sum_{t=1}^{T} \left( X_{t, ij} X_{t,kl} - Cov(X_{t, ij} X_{t,kl}) \right) \right| = O_p(T^{-1/2}).
\end{equation}
\end{lemma}
\begin{proof}
This lemma can be derived directly by following \cite{wang2019factor}. 
\end{proof}
\begin{lemma}  \label{lemma:4_conv_rates}
Under Conditions 1-6, it holds that
\begin{eqnarray}
\sum_{i=1}^{p} \sum_{j=1}^{p} \norm{ \widehat{\bOmega}^A_{s_1 s_2, ij} - \bOmega^A_{s_1 s_2, ij} }^2_2 & = & O_p(n_1n_2 p^{2-2\gamma} T^{-1}),   \label{Lemma2-sig-converg},\\
\sum_{i=1}^{p} \sum_{j=1}^{p} \norm{ \widehat{\bOmega}^A_{s_1 e_2, ij} - \bOmega^A_{s_1 e_2, ij} }^2_2 & = & O_p(n_1^{2} p^{2-\gamma} T^{-1}),  \label{Lemma2-sig-err-converg}, \\
\sum_{i=1}^{p} \sum_{j=1}^{p} \norm{ \widehat{\bOmega}^A_{e_1s_2, ij} - \bOmega^A_{e_1s_2, ij} }^2_2 & = & O_p(n_2^{2} p^{2-\gamma} T^{-1}),  \label{Lemma2-err-sig-converg}, \\
\sum_{i=1}^{p} \sum_{j=1}^{p} \norm{ \widehat{\bOmega}^A_{e_1 e_2, ij} - \bOmega^A_{e_1 e_2, ij} }^2_2 & = & O_p(n_1 n_2 p^2 T^{-1}).  \label{Lemma2-err-converg}\\
\sum_{i=1}^{m} \sum_{j=1}^{m} \norm{ \widehat{\bOmega}^B_{s_1 s_2, ij} - \bOmega^B_{s_1 s_2, ij} }^2_2 & = & O_p(m^{2} p^{2-2\gamma} T^{-1}),   \label{Lemma2-sig-converg1},\\
\sum_{i=1}^{m} \sum_{j=1}^{m} \norm{ \widehat{\bOmega}^B_{s_1 e_2, ij} - \bOmega^B_{s_1 e_2, ij} }^2_2 & = & O_p(m^{2} p^{2-\gamma} T^{-1}),  \label{Lemma2-sig-err-converg1}, \\
\sum_{i=1}^{m} \sum_{j=1}^{m} \norm{ \widehat{\bOmega}^B_{e_1s_2, ij} - \bOmega^B_{e_1s_2, ij} }^2_2 & = & O_p(m^{2} p^{2-\gamma} T^{-1}),  \label{Lemma2-err-sig-converg1}, \\
\sum_{i=1}^{m} \sum_{j=1}^{m} \norm{ \widehat{\bOmega}^B_{e_1 e_2, ij} - \bOmega^B_{e_1 e_2, ij} }^2_2 & = & O_p(m^2 p^2 T^{-1}).  \label{Lemma2-err-converg1}
\end{eqnarray}
\end{lemma}
\begin{proof}
This lemma can be derived directly by following \cite{wang2019factor}. 
\end{proof}

\begin{lemma}  \label{lemma:conv_rate_covYt}
Under Conditions 1-6, it holds that
\begin{equation}
\sum_{i=1}^{p} \sum_{j=1}^{p} \norm{ \widehat{\bOmega}_{A_l,ij} - \bOmega_{A,ij} }^2_2 = = O_p \left( n_1^{2} p^{2-\gamma} T^{-1} + n_2^{2} p^{2-\gamma} T^{-1} + n_1 n_2 p^2 T^{-1} \right) .
\end{equation}
\end{lemma}

\begin{proof}

\begin{align*}
\hat{\bOmega}_{A,ij} & = \frac{1}{T} \sum_{t=1}^{T} \bY_{1 t, \cdot i} \bY'_{2 t, \cdot j} \\
& = \frac{1}{T} \sum_{t=1}^{T} \left( \bA_1 \bX_t \bb_{i \cdot} + \be_{t, \cdot i} \right) \left( \bA_2 \bX_t \bb_{j \cdot} + \be_{t, \cdot j} \right)' \\
& = \hat{\bOmega}^A_{s_1 s_2, ij} + \hat{\bOmega}^A_{s_1 e_2, ij} + \hat{\bOmega}^A_{e_1 s_2, ij} + \hat{\bOmega}^A_{e_1 e_2, ij}.
\end{align*}

\begin{align*}
& \sum_{i=1}^{p} \sum_{j=1}^{p} \norm{\hat{\bOmega}_{A,ij} - \bOmega_{A,ij}}^2_2 \\
& \le 4 \sum_{i=1}^{p} \sum_{j=1}^{p} \left( \norm{ \hat{\bOmega}^A_{s_1s_2, ij} - \bOmega^A_{s_1s_2, ij} }^2_2 + \norm{\hat{\bOmega}^A_{s_1e_2, ij} - \bOmega^A_{s_1e_2, ij}}^2_2 + \norm{\hat{\bOmega}^A_{e_1s_2, ij} - \bOmega^A_{e_1s_2, ij}}^2_2 + \norm{\hat{\bOmega}^A_{e_1e_2, ij} - \bOmega^A_{e_1e_2, ij}}^2_2 \right) \\
& = O_p(n_1^{2} p^{2-\gamma} T^{-1} + n_2^{2} p^{2-\gamma} T^{-1} + n_1 n_2 p^2 T^{-1})
\end{align*}
\end{proof}

\begin{lemma}  \label{lemma:conv_rate_Mhat}
Under Conditions 1-6 and  $n_1 n_2 p^{-2+2\gamma}  T^{-1/2} = o(1)$, it holds that
\begin{equation}
\norm{ \widehat{\bM}_{A_l} - \bM_{A_l} }_2 = O_p \left( n^{2} p^{2-\gamma} T^{-1/2} \right),
\end{equation}
where $l=1,2$.
\end{lemma}

\begin{proof}
Here we only consider $l=1$. If follows the same procedure for $l=2$.
\begin{align*}
\sum_{i=1}^{p} \sum_{j=1}^{p} \norm{\bOmega_{ij}}^2_2 & = \sum_{i=1}^{p} \sum_{j=1}^{p} \norm{ \bA_1 \frac{1}{T} \sum_{t=1}^{T} \Cov{\left( \bX_t \bb_{i \cdot}, \bX_t \bb_{j \cdot} \right)} \bA'_2}^2_2 \\
& \le \sum_{i=1}^{p} \sum_{j=1}^{p} \norm{\bA_1}^2_2 \norm{\bA_2}^2_2 \norm{\frac{1}{T} \sum_{t=1}^{T} \E{ \bX_t \bb_{i \cdot} \bb'_{j \cdot} \bX'_t  }}^2_2 \\
& \le \norm{\bA_1}^2_2 \norm{\bA_2}^2_2 \sum_{i=1}^{p} \sum_{j=1}^{p} \norm{\frac{1}{T} \sum_{t=1}^{T} \E{ \bX_t \otimes \bX_t} \vect{\bb_{i \cdot} \bb'_{j \cdot}} }^2_2 \\
& \le \norm{\bA_1}^2_2 \norm{\bA_2}^2_2 \sum_{i=1}^{p} \sum_{j=1}^{p} \norm{\frac{1}{T} \sum_{t=1}^{T} \E{ \bX_t \otimes \bX_t} }^2_2 \norm{ \vect{\bb_{i \cdot} \bb'_{j \cdot}} }^2_2 \\
& = \norm{\bA_1}^2_2 \norm{\bA_2}^2_2 \sum_{i=1}^{p} \sum_{j=1}^{p} \norm{\frac{1}{T} \sum_{t=1}^{T} \E{ \bX_t \otimes \bX_t} }^2_2 \norm{ \bb_{i \cdot} \bb'_{j \cdot} }^2_F \\
& \le \norm{\bA_1}^2_2 \norm{\bA_2}^2_2 \norm{\frac{1}{T} \sum_{t=1}^{T} \E{ \bX_t \otimes \bX_t} }^2_2 \sum_{i=1}^{p} \sum_{j=1}^{p} \norm{ \bb_{i \cdot}}^2_2 \norm{\bb'_{j \cdot}}^2_2 \\
& = \norm{\bA_1}^2_2 \norm{\bA_2}^2_2 \norm{\frac{1}{T} \sum_{t=1}^{T} \E{ \bX_t \otimes \bX_t} }^2_2 \norm{\bB}^4_F \\
& \le \norm{\bA_1}^2_2 \norm{\bA_2}^2_2 \norm{\frac{1}{T} \sum_{t=1}^{T} \E{ \bX_t \otimes \bX_t} }^2_2 \cdot r^2 \cdot \norm{\bB}^4_2  \\
& = O_p \left(n_1n_2p^{2-2\gamma} \right)
\end{align*}

Then,
\begin{eqnarray*}
\norm{\hat{\bM}_{A_1} - \bM_{A_1} }_2 & = & \norm{\sum_{i=1}^{p} \sum_{j=1}^{p} \left( \hat{\bOmega}_{A,ij}\hat{\bOmega}'_{A,ij} - \bOmega_{A,ij} \bOmega'_{A,ij} \right)}_2 \\
& \le & \sum_{i=1}^{p} \sum_{j=1}^{p} \norm{\hat{\bOmega}_{A,ij} - \bOmega_{A,ij}}^2_2 + 2 \sum_{i=1}^{p} \sum_{j=1}^{p} \norm{\bOmega_{A,ij}}_2 \norm{\hat{\bOmega}_{A,ij} - \bOmega_{A,ij}}_2 \\
& \le & \sum_{i=1}^{p} \sum_{j=1}^{p} \norm{\hat{\bOmega}_{A,ij} - \bOmega_{A,ij}}^2_2 + 2 \left( \sum_{i=1}^{p} \sum_{j=1}^{p} \norm{\bOmega_{A,ij}}^2_2 \cdot \sum_{i=1}^{p} \sum_{j=1}^{p} \norm{\hat{\bOmega}_{A,ij} - \bOmega_{A,ij}}^2_2 \right)^{1/2} \\
& = & O_p((n_1^{2-\delta} p^{2-\gamma} + n_2^{2-\delta} p^{2-\gamma} + n_1 n_2 p^2) T^{-1}) \\
&   & + O_p \left( ( (n_1^{3} n_2 p^{4-3\gamma} + n_1 n_2^{3} p^{4-3\gamma} + n_1^{2} n_2^{2} p^{4-2\gamma} ) T^{-1} )^{1/2} \right).
\end{eqnarray*}

\end{proof}

\begin{lemma}  \label{conv_rate_sigval_M_1}
Under Condition \ref{cond:eigenval_cov_Et_bounded} and \ref{cond:Xt_cov_fullrank_bounded}, we have
\begin{equation*}
\lambda_i(\bM_{A_l}) \asymp n_1 n_2 p^{2-2\gamma}, \qquad i = 1, 2, \ldots, d,	
\end{equation*}
and
\begin{equation*}
\lambda_i(\bM_B) \asymp m^{2} p^{2-2\gamma}, \qquad i = 1, 2, \ldots, r,	
\end{equation*}
where $\lambda_i(\bM_{A_l})$ and $\lambda_i(\bM_B)$ denotes the $i$-th largest singular value of $\bM_{A_l}$ and $\bM_B$.
\end{lemma}

\begin{proof}
Note that
\begin{equation*}
    \begin{aligned}
        \bOmega_{A,ij}&= \bA_1(\bs) \Cov{ \left(\bX_t \bb'_{i \cdot}, \bX^*_t \bb'_{j \cdot} \right)} \bA_2(\bs)\\
        & = \bA_1(\bs) \mathbb{E}\left[(\bb_{i \cdot} \otimes I_d) \vect{\bX_t} \vect{\bX_t}'(\bb'_{j \cdot} \otimes I_d)\right] \bA_2(\bs)\\
        & = \bA_1(\bs) (\bb_{i \cdot} \otimes I_d)\bSigma{(\bX_t)} (\bb'_{j \cdot} \otimes I_d)\bA_2(\bs)
    \end{aligned}
\end{equation*}
where $\bSigma{(\bX_t)}=\Cov{\left[ \vect{\bX_t} \vect{\bX_t}'\right]}$.
By assumptions, we have
\begin{equation*}
    \begin{aligned}
    \lambda_{d}(\bM_{A_1}) & = \lambda_d\left(\sum_{i=1}^{p} \sum_{j=1}^{p} \bOmega_{A,ij} \bOmega'_{A,ij}\right)\\
    & \geq \norm{\bA_1(\bs)}_{min}^{2}\norm{\bA_2(\bs)}_{min}^{2} \lambda\left(\sum_{i=1}^{p} \sum_{j=1}^{p} (\bb_{i \cdot} \otimes I_d)\bSigma{(\bX_t)} (\bb'_{j \cdot} \otimes I_d)      (\bb_{j \cdot} \otimes I_d)\bSigma{(\bX_t)'} (\bb'_{i \cdot} \otimes I_d)\right)\\
    & = \norm{\bA_1(\bs)}_{min}^{2}\norm{\bA_2(\bs)}_{min}^{2} \lambda_d\left(\sum_{i=1}^{p} \sum_{j=1}^{p} (\bb_{i \cdot} \otimes I_d)\bSigma{(\bX_t)} (\bb'_{j \cdot}\bb_{j \cdot} \otimes I_d)\bSigma{(\bX_t)'} (\bb'_{i \cdot} \otimes I_d)\right)\\
    & = \norm{\bA_1(\bs)}_{min}^{2}\norm{\bA_2(\bs)}_{min}^{2} \lambda_d\left(\sum_{i=1}^{p}  (\bb_{i \cdot} \otimes I_d)\bSigma{(\bX_t)} (\bB'\bB \otimes I_d)\bSigma{(\bX_t)'} (\bb'_{i \cdot} \otimes I_d)\right)\\
    & = \norm{\bA_1(\bs)}_{min}^{2}\norm{\bA_2(\bs)}_{min}^{2} \lambda_d\left(\sum_{i=1}^{p}  (\bb_{i \cdot} \otimes I_d)\bSigma{(\bX_t)} (\bB' \otimes I_d)(\bB \otimes I_d)\bSigma{(\bX_t)'} (\bb'_{i \cdot} \otimes I_d)\right)\\
    & = \norm{\bA_1(\bs)}_{min}^{2}\norm{\bA_2(\bs)}_{min}^{2} \lambda_d\left(\sum_{i=1}^{p}  (\bB \otimes I_d)\bSigma{(\bX_t)'} (\bb'_{i \cdot} \otimes I_d) (\bb_{i \cdot} \otimes I_d)\bSigma{(\bX_t)} (\bB' \otimes I_d)\right)\\
    & = \norm{\bA_1(\bs)}_{min}^{2}\norm{\bA_2(\bs)}_{min}^{2} \lambda_d\left( (\bB \otimes I_d)\bSigma{(\bX_t)'} (\bB'\bB \otimes I_d) \bSigma{(\bX_t)} (\bB' \otimes I_d)\right)\\
    \\
    \end{aligned}
\end{equation*}
Since $\bB'\bB$ is a $r\times r$ symmetric positive definite matrix, we can find a $r\times r$ positive definite matrix $\bU$ such that $\bB'\bB=\bU'\bU$ and $\norm{\bU}_2^2\asymp O(p^{1-\gamma})\asymp \norm{\bU}_{min}^2$. By the property of Kronecker product, we can show that $\sigma_1(\bU\otimes I_d) \asymp O(p^{1/2-1/\gamma}) \asymp \sigma_{dr}(\bU\otimes I_d)$. Based on the results in \cite{merikoski2004inequalities}, we have $\sigma_d(\bSigma(\bX_t)'(\bU\otimes I_d)) \asymp O(p^{1/2-1/\gamma})$. Then,

\begin{equation*}
\begin{aligned}
    \lambda_{d}(\bM_{A_1})& \geq \norm{\bA_1(\bs)}_{min}^{2}\norm{\bA_2(\bs)}_{min}^{2}\lambda_d\left( (\bU' \otimes I_d)\bSigma{(\bX_t)} (\bU \otimes I_d)(\bU' \otimes I_d) \bSigma{(\bX_t)'} (\bU \otimes I_d)\right)\\
    & \geq \norm{\bA_1(\bs)}_{min}^{2}\norm{\bA_2(\bs)}_{min}^{2} \sigma_d\left[ (\bU' \otimes I_d)\bSigma{(\bX_t)} (\bU \otimes I_d )\right]^2 =  O\left(n_1 n_2 p^{2-2\gamma}\right).
\end{aligned}
\end{equation*}
Other results follow the same procedure.
\end{proof}

\noindent\textbf{Proof of Theorem \ref{thm:A_sep_err_bnd}.} 
\begin{proof}
By Perturbation Theorem,
\begin{eqnarray}
\norm{\hat{\bQ}_{A,1} - \bQ_{A,1}}_2 & \le & \frac{8}{\lambda_{min}(\bM_1)} \norm{\hat{\bM}_1 - \bM_1}_2 \nonumber \\
& = & O_p( ( n_1 n_2^{-1} p^{\gamma} + n_1^{-1} n_2 p^{\gamma} +  p^{2\gamma} ) T^{-1} ) \nonumber \\
&  &  + O_p( ( n_1 n_2^{-1} p^{\gamma} + n_1^{-1} n_2 p^{\gamma} + p^{2\gamma} ) T^{-1} )^{1/2} \nonumber \\
& = & O_p( ( n_1 n_2^{-1} p^{\gamma} + n_1^{-1} n_2 p^{\gamma} +  p^{2\gamma} ) T^{-1} )^{1/2}. \nonumber
\end{eqnarray}

If $n_1 \asymp n_2 \asymp n/2$, we have $\norm{\hat{\bQ}_{A,1} - \bQ_{A,1}}_2 = O_p( p^{\gamma} T^{-1/2} )$.

If set $n_2=c$ fixed and $n_1=n-c$, we have $\norm{\hat{\bQ}_{A,1} - \bQ_{A,1}}_2 = O_p(  n^{1/2} p^{-\gamma/2}  p^{\gamma} T^{-1/2} )$.

We have the same result for $\norm{\hat{\bQ}_{A,2} - \bQ_{A,2}}_2$.
\end{proof}

\begin{lemma}  \label{lemma:conv_rate_covYt1}
Under Conditions 1-6, it holds that
\begin{equation}
\sum_{i=1}^{m} \sum_{j=1}^{m}  \norm{ \widehat{\bOmega}_{B,ij} - \bOmega_{B,ij} }^2_2 = O_p \left( m^2 p^2 T^{-1}\right) .
\end{equation}
\end{lemma}

\begin{proof}

\begin{align*}
\hat{\bOmega}_{B,ij} & = \frac{1}{T} \sum_{t=1}^{T} \bY'_{1 t,i\cdot} \bY_{2 t, j\cdot} \\
& = \frac{1}{T} \sum_{t=1}^{T} \left( \ba_{1, i\cdot}\bX_t\bB^\top+\be_{1t,i\cdot} \right) \left( \ba_{2, j\cdot}\bX_t^*\bB^\top+\be_{1t,j\cdot} \right)' \\
& = \hat{\bOmega}^B_{s_1 s_2, ij} + \hat{\bOmega}^B_{s_1 e_2, ij} + \hat{\bOmega}^B_{e_1 s_2, ij} + \hat{\bOmega}^B_{e_1 e_2, ij}.
\end{align*}

\begin{align*}
& \sum_{i=1}^{m}\sum_{j=1}^{m}  \norm{\hat{\bOmega}_{B,ij} - \bOmega_{B,ij}}^2_2 \\
& \le 4 \sum_{i=1}^{m} \sum_{j=1}^{m} \left( \norm{ \hat{\bOmega}^B_{s_1s_2, ij} - \bOmega^B_{s_1s_2, ij} }^2_2 + \norm{\hat{\bOmega}^B_{s_1e_2, ij} - \bOmega^B_{s_1e_2, ij}}^2_2 + \norm{\hat{\bOmega}^B_{e_1s_2, ij} - \bOmega^B_{e_1s_2, ij}}^2_2 + \norm{\hat{\bOmega}^B_{e_1e_2, ij} - \bOmega^B_{e_1e_2, ij}}^2_2 \right) \\
& = O_p(m^2 p^2 T^{-1})
\end{align*}
\end{proof}

\begin{lemma} \label{conver_rate_M_B_hat}
Under Condition 1-6, and other conditions, it holds that
\begin{equation}
\norm{ \widehat{\bM}_B - \bM_B }_2 = O_p \left(   m^2 p^2 T^{-1/2} \right).
\end{equation}
\end{lemma}
\begin{proof}
\begin{align*}
\sum_{i=1}^{m}\sum_{j=1}^{m} \norm{\bOmega_{B,ij}}^2_2 & = \sum_{i=1}^{m}  \norm{ \bB \frac{1}{T} \sum_{t=1}^{T} \Cov{ \left(\ba_{1, i\cdot}\bX_t, \ba_{2, j\cdot}\bX_t^* \right)} \bB'}^2_2 \\
& \le \sum_{i=1}^{m}\sum_{j=1}^{m} \norm{\bB}^4_2  \norm{\frac{1}{T} \sum_{t=1}^{T} \E{ \bX'_t \ba'_{1, i\cdot}\ba_{2, i\cdot}\bX_t^* }}^2_2 \\
& \le \norm{\bB}^4_2 \sum_{i=1}^{m}\sum_{j=1}^{m} \norm{\frac{1}{T} \sum_{t=1}^{T} \E{ \bX'_t \otimes \bX'_t} \vect{\ba'_{1,i\cdot} \ba_{2, i \cdot}} }^2_2 \\
& \le \norm{\bA_1}^2_2 \norm{\bA_2}^2_2 \sum_{i=1}^{m}\sum_{j=1}^{m} \norm{\frac{1}{T} \sum_{t=1}^{T} \E{ \bX'_t \otimes \bX'_t} }^2_2 \norm{ \vect{\ba'_{1,i\cdot} \ba_{2, i \cdot}} }^2_2 \\
& = \norm{\bB}^4_2 \sum_{i=1}^{m}\sum_{j=1}^{m}  \norm{\frac{1}{T} \sum_{t=1}^{T} \E{ \bX'_t \otimes \bX'_t} }^2_2 \norm{ \ba'_{1,i\cdot} \ba_{2, i \cdot} }^2_F \\
& \le \norm{\bB}^4_2 \norm{\frac{1}{T} \sum_{t=1}^{T} \E{ \bX'_t \otimes \bX'_t} }^2_2 \sum_{i=1}^{m} \sum_{j=1}^{m} \norm{ \ba'_{1,i\cdot}}^2_2 \norm{\ba_{2, i \cdot}}^2_2 \\
& = \norm{\bB}^4_2 \norm{\frac{1}{T} \sum_{t=1}^{T} \E{ \bX'_t \otimes \bX'_t} }^2_2 \norm{\bA_1}^2_F \norm{\bA_2}^2_F \\
& \le \norm{\bA_1}^2_2 \norm{\bA_2}^2_2 \norm{\frac{1}{T} \sum_{t=1}^{T} \E{ \bX_t \otimes \bX_t} }^2_2 \cdot d^2 \cdot \norm{\bB}^4_2  \\
& = O_p \left( m^{2} p^{2-2\gamma} \right)
\end{align*}

Then,
\begin{eqnarray*}
\norm{\hat{\bM}_B - \bM_B}_2 & = & \norm{\sum_{i=1}^{m}\sum_{j=1}^{m} \left( \hat{\bOmega}_{B,ij}\hat{\bOmega}'_{B,ij} - \bOmega_{B,ij} \bOmega'_{B,ij} \right)}_2 \\
& \le & \sum_{i=1}^{m}\sum_{j=1}^{m}  \norm{\hat{\bOmega}_{B,ij} - \bOmega_{B,ij}}^2_2 + 2 \sum_{i=1}^{m}\sum_{j=1}^{m}  \norm{\bOmega_{B,ij}}_2 \norm{\hat{\bOmega}_{B,ij} - \bOmega_{B,ij}}_2 \\
& \le & \sum_{i=1}^{m} \sum_{j=1}^{m} \norm{\hat{\bOmega}_{B,ij} - \bOmega_{B,ij}}^2_2 + 2 \left( \sum_{i=1}^{m} \sum_{j=1}^{m} \norm{\bOmega_{B,ij}}^2_2 \cdot \sum_{i=1}^{m}\sum_{j=1}^{m} \norm{\hat{\bOmega}_{B,ij} - \bOmega_{B,ij}}^2_2 \right)^{1/2} \\
& = &  O_p \left( m^2 p^2 T^{-1} \right) + O_p \left( m^{2} p^{2-\gamma} T^{-1/2} \right).
\end{eqnarray*}
\end{proof}

\noindent\textbf{Proof of Theorem \ref{thm:B_sep_err_bnd}.}
\begin{proof}
By perturbation theorem,
\begin{eqnarray}
\norm{\hat \bQ_B - \bQ_B}_2 & \le & \frac{8}{\lambda_{min}(\bM_B)} \norm{\hat{\bM}_B - \bM_B}_2 \nonumber \\
& = & O_p(  p^{\gamma} T^{-1/2} ). \nonumber
\end{eqnarray}
\end{proof}

\noindent\textbf{Proof of Theorem \ref{thm:signal_err_bnd}.}
\begin{proof}
\begin{eqnarray}
\norm{\hat{\bXi}_{1t} - \bXi_{1t}}_2  & = &  \norm{ \hat \bQ_{A, 1} \hat \bQ_{A, 1}' \left( \bQ_{A, 1} \bZ_{1t} \bQ_B' + \bE_{1t} \right) \hat \bQ_B \hat \bQ_B' - \bQ_{A, 1} \bZ_{1t} \bQ_B'}_2 \nonumber \\
& \le &  \norm{ \hat \bQ_{A, 1} \hat \bQ_{A, 1}' \bQ_{A, 1} \bZ_{1t} \bQ_B' \left( \hat \bQ_B \hat \bQ_B' - \bQ_B \bQ_B' \right) }_2 \nonumber \\
& & + \; \norm{ \left( \hat \bQ_{A, 1} \hat \bQ_{A, 1}' - \bQ_{A, 1} \bQ_{A, 1}' \right) \bQ_{A, 1} \bZ_{1t} \bQ_B'}_2 \nonumber \\
& & +  \norm{ \hat \bQ_{A, 1} \hat \bQ_{A, 1}' \bE_{1t} \hat \bQ_B \hat \bQ_B' }_2 \nonumber \\
& = & \bI_1 + \bI_2 + \bI_3. \nonumber
\end{eqnarray}
Note that $\norm{\bZ_{1t}}_2 \asymp m^{1/2-\delta/2}p^{1/2-\gamma/2}$. 
Thus, we have
\begin{eqnarray}
\bI_1 & \le &\norm{ \hat \bQ_{A, 1} \hat \bQ_{A, 1}' \bQ_{A, 1} \bZ_{1t} \left( \bQ_B'\hat \bQ_B\hat \bQ_B'-\bQ_B' \right) }_2 \nonumber \\
& = & \norm{ \hat \bQ_{A, 1} \hat \bQ_{A, 1}' \bQ_{A, 1} \bZ_{1t} \left[ \bQ_B'\left(\hat \bQ_B-\bQ_B+\bQ_B\right)\hat \bQ_B'-\bQ_B' \right] }_2 \nonumber \\
& \le & \norm{ \hat \bQ_{A, 1} \hat \bQ_{A, 1}' \bQ_{A, 1} \bZ_{1t}  \bQ_B'\left(\hat \bQ_B-\bQ_B\right)\hat \bQ_B' }_2 + \norm{ \hat \bQ_{A, 1} \hat \bQ_{A, 1}' \bQ_{A, 1} \bZ_{1t}  \left(\hat \bQ_B'-\bQ_B'\right) }_2 \nonumber \\
& \le & \norm{ \hat \bQ_{A, 1}}_2\norm{\hat \bQ_{A, 1}'}_2\norm{\bQ_{A, 1}}_2\norm{\bZ_{1t}}_2\norm{\bQ_B'}_2\norm{\hat \bQ_B-\bQ_B}_2\norm{\hat \bQ_B'}_2  \nonumber \\
&  & + \norm{ \hat \bQ_{A, 1}}_2 \norm{\hat \bQ_{A, 1}'}_2 \norm{\bQ_{A, 1}}_2 \norm{\bZ_{1t}}_2  \norm{\hat \bQ_B'-\bQ_B' }_2  \nonumber \\
& = &  \Op{m^{1/2} p^{1/2+\gamma/2} T^{-1/2}}.  \nonumber
\end{eqnarray}

\begin{eqnarray*}
&\bI_2 & \le \norm{ \left(\hat \bQ_{A, 1} \hat \bQ_{A, 1}' \bQ_{A, 1} - \bQ_{A, 1}\right) \bZ_{1t} \bQ_B' }_2 = \norm{ \left[\hat \bQ_{A, 1} \left(\hat \bQ_{A, 1}' -\bA'_1 + \bA'_1\right) \bQ_{A, 1} - \bQ_{A, 1}\right] \bZ_{1t} \bQ_B' }_2 \\
& & \le \norm{ \hat \bQ_{A, 1} \left(\hat \bQ_{A, 1}' -\bA'_1 \right) \bQ_{A, 1}  \bZ_{1t} \bQ_B' }_2 + \norm{  \left(\hat \bQ_{A, 1} -\bQ_{A, 1} \right)\bZ_{1t} \bQ_B' }_2\\
& & \le \norm{ \hat \bQ_{A, 1}}_2 \norm{ \hat \bQ_{A, 1}' -\bA'_1 }_2 \norm{ \bQ_{A, 1}}_2  \norm{ \bZ_{1t}}_2 \norm{ \bQ_B' }_2 + \norm{ \hat \bQ_{A, 1} -\bQ_{A, 1} }_2 \norm{ \bZ_{1t}}_2 \norm{ \bQ_B' }_2\\
& & = O_p\left( m^{1/2} p^{1/2+\gamma/2} T^{-1/2}\right).
\end{eqnarray*}

\begin{eqnarray*}
&\bI_3 & \le \norm{ \hat \bQ_{A, 1}' \bE_t \hat \bQ_B }_2 = \norm{(\hat \bQ_B' \otimes \hat \bQ_{A, 1}' )\vect{\bE_t}}_2 \le rd \norm{\bSigma_e}_2 = O_p(1).
\end{eqnarray*}

Thus,
\begin{equation*}
\norm{\hat{\bXi}_{1t} - \bXi_{1t}}_2  = O_p\left( m^{1/2} p^{1/2+\gamma/2} T^{-1/2}\right) + O_p(1).
\end{equation*}

\begin{equation*}
m^{-1/2} p^{-1/2} \norm{\hat{\bXi}_{1t} - \bXi_{1t}}_2  = O_p\left(  p^{\gamma/2} T^{-1/2}\right) + O_p(n_1^{-1/2} p^{-1/2}).
\end{equation*}

Similarly for ${\bXi}_{2t}$, we have
\begin{equation*}
\norm{\hat{\bXi}_{2t} - \bXi_{2t}}_2 = O_p\left( m^{1/2} p^{1/2+\gamma/2} T^{-1/2}\right) + O_p(1).
\end{equation*}

\begin{equation*}
n_2^{-1/2} p^{-1/2} \norm{\hat{\bXi}_{2t} - \bXi_{2t}}_2  = O_p\left(  p^{\gamma/2} T^{-1/2}\right) + O_p(n_2^{-1/2} p^{-1/2}).
\end{equation*}

Now we find the $L_2$-norm bounds for
\begin{equation*}
\norm{\hat{\bXi}_{t} - \bXi_{t}}^2_2 = \norm{ \begin{pmatrix}
\hat{\bXi}_{1t} - \bXi_{1t} \\
\hat{\bXi}_{2t} - \bXi_{2t}
\end{pmatrix}  }^2_2 .
\end{equation*}

Let $\bM = \hat{\bXi}_{t} - \bXi_{t} = \begin{pmatrix} \bM_1 \\ \bM_2 \end{pmatrix}$, the above problem is equivelent to finding $\lambda_{max}(\bM'\bM)$ from $\lambda_{max}(\bM_1'\bM_1)$ and $\lambda_{max}(\bM_2'\bM_2)$. \\

Since
\begin{equation*}
\lambda_{max}(\bM'\bM) = \lambda_{max}(\bM_1'\bM_1 + \bM_2'\bM_2) \le \lambda_{max}(\bM_1'\bM_1) + \lambda_{max}(\bM_2'\bM_2),
\end{equation*}

We have
\begin{align*}
\norm{\hat{\bXi}_{t} - \bXi_{t}}^2_2  & \le \norm{\hat{\bXi}_{1t} - \bXi_{1t}}^2_2  + \norm{\hat{\bXi}_{2t} - \bXi_{2t}}^2_2  \\
& =  O_p\left( m p^{1+\gamma} T^{-1}\right) + O_p\left( m^{1/2} p^{1/2+\gamma/2} T^{-1/2}\right) + O_p(1).
\end{align*}

$$n^{-1} p^{-1} \norm{\hat{\bXi}_{t} - \bXi_{t}}^2_2  = O_p{( p^{\gamma} T^{-1} + m^{-1/2} p^{-1/2+\gamma/2} T^{-1/2} + n^{-1} p^{-1})}.$$
\end{proof}

\subsection{Spacial loading matrix re-estimation}

\begin{lemma}
If $n_1 \asymp n2 \asymp n$, then
\begin{equation}
n^{-1/2} p^{-1/2} \norm{\hat{\bPsi}_{lt} - \bPsi_{lt}}_2  = O_p( p^{\gamma/2} T^{-1/2}) + O_p(n^{-1/2} p^{-1/2}),
\end{equation}
where $l = 1, 2$, and
\begin{equation}
n^{-1} p^{-1} \norm{\hat{\bPsi}_{t} - \bPsi_{t}}^2_2  = \Op{ p^{\gamma} T^{-1} + n^{-1} p^{-1}}
\end{equation}
\end{lemma}

\begin{proof}

\begin{eqnarray*}
\norm{\bPsi_{it} - \bPsi_{it}}_2 & = & \norm{ \hat{\bQ}_{A_i} \hat{\bZ}_t - \bQ_{A_i} \bZ_t}_2 = \norm{ \hat{\bQ}_{A_i} \hat{\bQ}'_{A_i} (\bQ_{A_i} \bZ_t \bQ'_B + \bE_t) \hat{\bQ}_B - \bQ_{A_i} \bZ_t}_2 \\
& = & \norm{ \hat{\bQ}_{A_i} \hat{\bQ}'_{A_i} \bQ_{A_i} \bZ_t \bQ'_B ( \hat{\bQ}_B - \bQ_B) + (\hat{\bQ}_{A_i} \hat{\bQ}'_{A_i} - \bQ_{A_i} \bQ'_{A_i})\bQ_{A_i} \bZ_t + \hat{\bQ}_{A_i} \hat{\bQ}'_{A_i} \bE_t \hat{\bQ}_B}_2 \\
& \le & \norm{ \hat{\bQ}_{A_i} \hat{\bQ}'_{A_i} \bQ_{A_i} \bZ_t \bQ'_B ( \hat{\bQ}_B - \bQ_B) }_2 + \norm{ (\hat{\bQ}_{A_i} \hat{\bQ}'_{A_i} - \bQ_{A_i} \bQ'_{A_i})\bQ_{A_i} \bZ_t }_2 + \norm{ \hat{\bQ}_{A_i} \hat{\bQ}'_{A_i} \bE_t \hat{\bQ}_B}_2
\end{eqnarray*}

Then, similar to the proof of Theorem \ref{thm:signal_err_bnd}, we have the desired results.
\end{proof}

Let $\Delta_{npT} =  p^{\gamma} T^{-1} + n^{-1} p^{-1}$. 
Then $\Delta_{npT}$ is the convergence rate of $n^{-1} p^{-1} \norm{\hat \bPsi_t - \bPsi_t}^2_2$. 
Since $\norm{\hat \bPsi_t - \bPsi_t}^2_2 \le \norm{\hat \bPsi_t - \bPsi_t}^2_F \le r \norm{\hat \bPsi_t - \bPsi_t}^2_2$ where $r$ is fixed, we have $n^{-1} p^{-1} \norm{\hat \bPsi_t - \bPsi_t}^2_F = O_p(\Delta_{npT})$.

By definition, we have $\frac{1}{T} \sum_{t=1}^T \bPsi_t \bPsi_t' = \frac{1}{T} \sum_{t=1}^T \bQ_A \bZ_t \bZ_t' \bQ_A'$. Thus $\hat \bQ_A$ and $\hat \bZ_t$ can be estimated from $\frac{1}{T} \sum_{t=1}^T \hat\bPsi_t \hat\bPsi_t'  = \frac{1}{T}(\bPsi_t+\bU_t)(\bPsi_t+\bU_t)'$, where $\bU_t = \hat \bPsi_t - \bPsi_t$ is the approximation error from the previous steps.

Let $\bV_{npT}$ be the $d \times d$ diagonal matrix of the first $d$ largest eigenvalues of $\frac{1}{T} \sum_{t=1}^T \hat\bPsi_t \hat\bPsi_t'$ in decreasing order. By definition of eigenvectors and eigenvalues, we have $\frac{1}{T} \sum_{t=1}^T \hat\bPsi_t \hat\bPsi_t' \hat\bQ_A = \hat\bQ_A \bV_{npT}$ or $\frac{1}{T} \sum_{t=1}^T \hat\bPsi_t \hat\bPsi_t' \hat\bQ_A \bV_{npT}^{-1} = \hat\bQ_A$.

Define $\bH = \frac{1}{T} \sum_{t=1}^T \bZ_t \bZ_t' \bQ_A' \hat\bQ_A \bV^{-1}_{npT}$, then
\begin{eqnarray}
\hat \bQ_A - \bQ_A\bH & = & \frac{1}{T} \sum_{t=1}^T \paran{\bQ_A \bZ_t + \bU_t} \paran{\bQ_A \bZ_t + \bU_t}' \hat\bQ_A \bV_{npT}^{-1} -  \frac{1}{T} \sum_{t=1}^T \bQ_A \bZ_t \bZ_t' \bQ_A' \hat\bQ_A \bV^{-1}_{npT} \nonumber \\
& = & \left( \frac{1}{T} \sum_{t=1}^T \bQ_A \bZ_t \bU_t' \hat\bQ_A  + \frac{1}{T} \sum_{t=1}^T \bU_t \bZ_t' \bQ_A' \hat\bQ_A + \frac{1}{T} \sum_{t=1}^T \bU_t\bU_t' \hat\bQ_A \right) \bV^{-1}_{npT} \nonumber \\
& = & \left( \bN_1 + \bN_2 + \bN_3 \right) \bV^{-1}_{npT}. \nonumber
\end{eqnarray}

\begin{lemma}
$\norm{\bN_1}^2_2 = \norm{\bN_2}^2_2 = O_p(n^{2} p^{2-\gamma} \Delta_{npT})$ and $\norm{\bN_3}^2_2 = O_p(n^2 p^2 \Delta_{npT}^2)$.
\end{lemma}

\begin{proof}

Note that $\norm{\bU_t}^2_2 = \norm{\hat \bPsi_t - \bPsi_t}^2_2 = \Op{np \Delta_{npT} }$ and $\norm{\bQ_A}^2_F = \norm{\hat\bQ_A}^2_2 = \Op{1}$. 
In addition, we have $\norm{\bZ_t}^2_2 \asymp \Op{np^{1-\gamma}}$.

Thus,
\begin{eqnarray}
\norm{\bN_1}^2_2 & \le & \frac{1}{T} \sum_{t=1}^T \norm{\bQ_A \bZ_t \bU_t' \hat\bQ_A}_2^2 \le \norm{\bQ_A}^2_2 \norm{\bZ_t}^2_2 \norm{\bU'}^2_2 \norm{\hat\bQ_A}^2_2 = O_p(n^{2} p^{2-\gamma} \Delta_{npT}) \nonumber \\
\norm{\bN_2}^2_2 & \le & \frac{1}{T} \sum_{t=1}^T \norm{\bU_t \bZ_t' \bQ_A' \hat\bQ_A}^2_2 \le \norm{\bU_t}^2_F \norm{\bZ_t'}^2_2 \norm{\bQ_A'}^2_2 \norm{\hat\bQ_A}^2_2= \Op{n^{2} p^{2-\gamma} \Delta_{npT}} \nonumber \\
\norm{\bN_3}^2_2 & \le & \frac{1}{T} \sum_{t=1}^T \norm{\bU_t\bU_t' \hat\bQ_A}^2_2 \le \norm{\bU_t\bU_t'}^2_2 \norm{\hat\bQ_A}^2_2 = O_p(n^2 p^2 \Delta^2_{npT} ) \nonumber
\end{eqnarray} 
\end{proof}

\begin{lemma}
(i) $\norm{\bV_{npT}}_2 = O_p(n p^{1-\gamma})$, $\norm{\bV_{npT}^{-1}}_2 = O_p(n^{-1} p^{\gamma-1})$. \\
(ii) $\norm{\bH}_2 = O_p(1)$. 
\end{lemma}

\begin{proof}
The $d$ eigenvalues of $\frac{1}{np} \bV_{npT}$ are the same as those of 
\begin{eqnarray}
\frac{1}{npT} \sum_{t=1}^T \hat\bPsi_t' \hat\bPsi_t  & = & \frac{1}{npT} \sum_{t=1}^T \paran{\bQ_A \bZ_t + \bU_t}' \paran{\bQ_A \bZ_t + \bU_t} \nonumber \\
& = & \frac{1}{npT} \sum_{t=1}^T \paran{\bZ_t' \bZ_t + \bZ_t' \bQ_A' \bU_t + \bU_t' \bQ_A \bZ_t + \bU_t' \bU_t}, \nonumber
\end{eqnarray}
which follows from $\hat \bPsi=\bQ_A \bZ_t + \bU$ and $\bQ_A \bQ_A = I_d$. Thus
\begin{eqnarray}
\norm{\frac{1}{npT} \sum_{t=1}^T \hat\bPsi_t' \hat\bPsi_t -  \frac{1}{npT} \sum_{t=1}^T \bZ_t' \bZ_t}_2^2 & \le &  \norm{\frac{1}{np} \bZ_t' \bQ_A' \bU_t}_2^2 + \norm{\frac{1}{np} \bU_t' \bQ_A \bZ_t}_2^2 + \norm{\frac{1}{np} \bU_t' \bU_t}_2^2  \nonumber \\
& = & \Op{ p^{-\gamma} \Delta_{npT} + \Delta_{npT}^2 } \nonumber.
\end{eqnarray}

Using the inequality for the $k$th eigenvalue of any matrix $\bM$, that is $\abs{\lambda_k(\bM) - \lambda_k(\bM)} \le \norm{\bM - \bM}$, we have $\abs{ \lambda_k \paran{\frac{1}{npT} \sum_{t=1}^T \hat\bPsi_t' \hat\bPsi_t} - \lambda_k \paran{ \frac{1}{npT} \sum_{t=1}^T \bZ_t' \bZ_t } } = o_p(1)$. $\lambda_k\paran{\frac{1}{npT} \sum_{t=1}^T \bZ_t' \bZ_t} \asymp  p^{-\gamma}$, $k=1, \ldots, d$. Thus, $\norm{\frac{1}{np}\bV_{npT}}_{min} \asymp  p^{-\gamma} \asymp \norm{\bV_{npT}}_2$, $\norm{ np \bV_{npT}^{-1}}_{min} \asymp  p^{\gamma} \asymp \norm{np \bV_{npT}^{-1}}_2$, and $\norm{\bH}_2 = O_p(1)$.

\end{proof}

\begin{lemma}
\begin{equation*}
\norm{\hat\bQ_A - \bQ_A\bH}^2_2 = \Op{ p^{\gamma} \Delta_{npT} +  p^{2\gamma} \Delta_{npT}^2}
\end{equation*}
\end{lemma}

\begin{proof}
Follow from Lemma 6, 7 and 8.
\end{proof}

\begin{lemma}
\[
\norm{\bH - \bI_d}_2 =  O_p \left(\Delta_{npT} +  p^{\gamma}\Delta_{npT}^2 \right)  + O_p \left(\Delta_{npT} T^{-1} +  p^{\gamma}\Delta_{npT}^2 T^{-1} \right)^{1/2}.
\]
\end{lemma}

\begin{proof}
By definition, we have $\bH = \frac{1}{T} \sum_{t=1}^T \bZ_t \bZ_t' \bQ_A' \hat\bQ_A \bV^{-1}_{npT}$. Thus,
\begin{eqnarray}
\norm{\bI_d - \hat \bQ_A' \bQ_A\bH}_2 & = & \norm{ \hat \bQ_A'( \hat \bQ_A - \bQ_A\bH)}_F \nonumber \\
& \le &\norm{\hat\bQ_A - \bQ_A\bH}^2_F + \norm{\bQ_A(\hat\bQ_A' - \bQ_A \bH)}_F \nonumber \\
& = &  \Op{ p^{\gamma} \Delta_{npT} +  p^{2\gamma} \Delta_{npT}^2}  +  \Op{ p^{\gamma} \Delta_{npT} +  p^{2\gamma} \Delta_{npT}^2}^{1/2} \nonumber
\end{eqnarray}

\begin{eqnarray}
\norm{\hat \bQ_A' \bQ_A \bH - \bH'\bH }_2 & = & \norm{ \paran{\hat \bQ_A - \bQ_A \bH}' \bQ_A\bH }_2 = \Op{ p^{\gamma} \Delta_{npT} +  p^{2\gamma} \Delta_{npT}^2}^{1/2} \nonumber
\end{eqnarray}
Thus,
\[
\norm{\bH'\bH - \bI_d}_2 = \Op{ p^{\gamma} \Delta_{npT} +  p^{2\gamma} \Delta_{npT}^2}  +  \Op{ p^{\gamma} \Delta_{npT} + p^{2\gamma} \Delta_{npT}^2}^{1/2}.
\]
In addition, by the definition of $\bH = \frac{1}{T} \sum_{t=1}^T \bZ_t \bZ_t' \bQ_A' \hat\bQ_A \bV^{-1}_{npT}$, we have
\begin{eqnarray}
\norm{\bH \frac{1}{np}\bV_{npT} - \frac{1}{npT} \sum_{t=1}^T \bZ_t \bZ_t' \bQ_A' \bQ_A \bH }_2 & = & \norm{ \frac{1}{npT} \sum_{t=1}^T \bZ_t \bZ_t' \bQ_A' \paran{  \hat\bQ_A - \bQ_A \bH } }  \nonumber \\
& \le & \frac{1}{npT} \sum_{t=1}^T \norm{ \bZ_t \bZ_t' \bQ_A' \paran{  \hat\bQ_A - \bQ_A \bH } }  \nonumber \\
& = & \Op{  p^{-\gamma} \paran{ p^{\gamma} \Delta_{npT} +  p^{2\gamma} \Delta_{npT}^2}^{1/2} }. \nonumber
\end{eqnarray}

With the same argument of Proposition C.3 in \cite{fan2016projected}, we have
\[
\norm{\bH - \bI_d}_2^2 = \Op{ p^{\gamma} \Delta_{npT} +  p^{2\gamma} \Delta_{npT}^2}.
\]

\end{proof}

\noindent\textbf{Proof of Theorem \ref{thm:reest_Q_A}.}
\begin{proof}
\begin{eqnarray}
\norm{\hat \bQ_A - \bQ_A}^2_2 & \le & \norm{\hat \bQ_A - \bQ_A\bH}^2_2 + 
\norm{\bQ_A}^2_2 \norm{\bH - \bI_d}^2_2 \nonumber \\
& = & \Op{ p^{\gamma} \Delta_{npT} + p^{2\gamma} \Delta_{npT}^2} \nonumber \\
& = & \Op{ p^{\gamma} \Delta_{npT}} \nonumber \\
& = & \Op{ p^{2\gamma} T^{-1} + n^{-1} p^{\gamma-1}}. \nonumber
\end{eqnarray}
\end{proof}

\noindent\textbf{Proof of Theorem \ref{thm:space-krig-bound}.}
\begin{proof}
\begin{eqnarray}
\norm{\hat \bZ_t - \bZ_t}^2_2 & = & \norm{\hat \bQ_A' \hat \bPsi_t  - \bQ_A' \bPsi_t} \nonumber \\
& \le & \norm{\hat \bQ_A' \paran{ \hat \bPsi_t -  \bPsi_t} }_2^2 + \norm{\paran{\hat \bQ_A - \bQ_A}' \bQ_A \bZ_t }_2^2 \nonumber \\
& \le & \norm{\hat \bQ_A'}_2^2  \norm{ \hat \bPsi_t -  \bPsi_t }_2^2 + \norm{\paran{\hat \bQ_A - \bQ_A}'}_2^2 \norm{ \bQ_A}_2^2 \norm{ \bZ_t }_2^2  \nonumber \\
& = & \Op{n p \Delta_{npT} + n p^{1+\gamma} \Delta_{npT}^2}. \nonumber
\end{eqnarray}
\end{proof}

\subsection{Sieve approximation of spacial loading function}

Spacial loading function $\bQ_A(\bs) = (q_{a,1}(\bs), \cdots, q_{a,d}(\bs))$, now we want to approximate $q_{a,j}(\bs)$ with linear combination of basis functions, the approximating functions are $\hat{g}_j(\bs)$.
We estimate $\hat{g}_j(\bs)$ based on estimated value $\hat \bQ_{A, \cdot j}$'s, where $\hat q_{a,j}(\bs)= q_{a,j}(\bs) + e_{a,j}(\bs)$.
From Theorem \ref{thm:reest_Q_A} we have $\norm{\hat \bQ_{A, \cdot j} - \bQ_{A, \cdot j} }^2_2 = \Op{ p^{\gamma} \Delta_{npT}}$, $j \in [d]$, where $\Delta_{npT} =  p^{\gamma} T^{-1} + n^{-1} p^{-1}$.

\begin{lemma}
If $q_{a,j}(\bs)$ and  $\hat q_{a,j}(\bs)$ belong to H\"{o}lder class, then $| q_{a,j}(\bs) |^2_{\infty} \asymp n^{-1}$, $\abs{ e_{a,j}(\bs) }^2_{\infty} = n^{-1} p^{\gamma} \Delta_{npT}$.
\end{lemma}

\begin{proof}
\begin{eqnarray*}
\lambda_{max}(\bQ_A \bQ_A') =  \lambda_{max}(\sum_{j=1}^{d} \bQ_{A, \cdot j} \bQ_{A, \cdot j}') \ge \lambda_{min}(\sum_{j=1}^{d} \bQ_{A, \cdot j} \bQ_{A, \cdot j}') & \ge & \sum_{j=1}^{d} \lambda_{min}(\bQ_{A, \cdot j}'\bQ_{A, \cdot j}) = \sum_{j=1}^{d} \sum_{i=1}^{n} A_{i j}^2 \\
\lambda_{min}(\bQ_A \bQ_A') = \lambda_{min}(\sum_{j=1}^{d} \bQ_{A, \cdot j} \bQ_{A, \cdot j}') \le \lambda_{max}(\sum_{j=1}^{d} \bQ_{A, \cdot j} \bQ_{A, \cdot j}') & \le & \sum_{j=1}^{d} \lambda_{max}(\bQ_{A, \cdot j}'\bQ_{A, \cdot j}) = \sum_{j=1}^{d} \sum_{i=1}^{n} A_{i j}^2
\end{eqnarray*}

Since $\norm{\bQ_A}^2_{min} \asymp \norm{\bQ_A}^2_{max} \asymp 1$, then $\| \bQ_{A, \cdot j} \|^2 \asymp 1$.
If $q_{a,j}(\bs)$ belongs to H\"{o}lder class, then $\abs{ q_{a,j}(\bs) }^2_{\infty} \asymp n^{-1}$ by multivariate Taylor expansion and Sandwich Theorem.

Similarly, $e_{a,j}(\bs) = \hat q_{a,j}(\bs) - q_{a,j}(\bs)$ belongs to H\"{o}lder class, from that fact that $\norm{\hat \bQ_{A, \cdot j} - \bQ_{A, \cdot j} }^2_2 = \Op{ p^{\gamma} \Delta_{npT}}$, $j \in [d]$, we have that $\abs{ e_{a,j}(\bs) }^2_{\infty} = n^{-1} p^{\gamma} \Delta_{npT}$.

\end{proof}

\begin{lemma}
$\norm{\hat{g}_j(\bs) - q_{a,j}(\bs)}_{\infty} = \Op{J_n^{-\kappa}  p^{1/2-\gamma/2}} + \Op{\sqrt{n^{-1} p^{\gamma} \Delta_{npT}}}$.
\end{lemma}

\begin{proof}
Following Theorem 12.6, 12.7 and 12.8 in \cite{schumaker2007spline}, we have
\begin{eqnarray}
\norm{\hat{g}_j(\bs) - q_{a,j}(\bs)}_{\infty} & = & \norm{\bP \hat q_{a,j}(\bs) - q_{a,j}(\bs)}_{\infty} \nonumber \\
& \le & \norm{\bP  q_{a,j}(\bs) - q_{a,j}(\bs)}_{\infty} + \norm{ \bP e_{a,j}(\bs) - e_{a,j}(\bs) }_{\infty} + \norm{ e_{a,j}(\bs) }_{\infty} \nonumber \\
& = & \Op{J_n^{-\kappa} n^{-1/2}} + \Op{J_n^{-\kappa} \sqrt{n^{-1} p^{\gamma} \Delta_{npT}}} + \Op{\sqrt{n^{-1} p^{\gamma} \Delta_{npT}}}  \nonumber \\
& = & \Op{J_n^{-\kappa} n^{-1/2}} + \Op{\sqrt{n^{-1} p^{\gamma} \Delta_{npT}}}  \nonumber
\end{eqnarray}

where $\bP \hat q_{a,j}(\bs)$ denotes the project of function $q_{a,j}(\bs)$ on the vector space spanned by the $J_n$ basis functions.
\end{proof}

\noindent\textbf{Proof of Theorem \ref{thm:Zt-bound}.}
\begin{proof}
We have $ \bxi_t(\bs_0) = \bQ_B \bZ'_t \bq_a(\bs_0) $ and $\hat \bxi_t(\bs_0) = \hat \bQ_B \hat \bZ'_t \hat \bg(\bs_0) $, thus
\begin{eqnarray}
\frac{1}{p} \norm{\hat \bxi_t(\bs_0) - \bxi_t(\bs_0)}_2^2 & = & \frac{1}{p} \norm{\hat \bQ_B \hat \bZ'_t \hat \bg(\bs_0) - \bQ_B \bZ'_t \bq_a(\bs_0)}_2^2  \nonumber \\
 & \le & \frac{1}{p} \norm{ \paran{ \hat \bQ_B - \bQ_B } \paran{\hat \bZ'_t -\bZ'_t } \paran{ \hat \bg(\bs_0) - \bq_a(\bs_0)} }_2^2  \nonumber \\
  &  & + \; \frac{1}{p} \norm{ \bQ_B \paran{\hat \bZ'_t -\bZ'_t } \paran{ \hat \bg(\bs_0) - \bq_a(\bs_0)} }_2^2  \nonumber \\
 &  & + \frac{1}{p} \norm{ \paran{ \hat \bQ_B - \bQ_B } \bZ'_t  \paran{ \hat \bg(\bs_0) - \bq_a(\bs_0)} }_2^2  \nonumber \\
 &  & + \frac{1}{p} \norm{ \paran{ \hat \bQ_B - \bQ_B } \paran{\hat \bZ'_t -\bZ'_t } \bq_a(\bs_0)}_2^2  \nonumber \\
 &  & + \frac{1}{p} \norm{ \bQ_B \bZ'_t \paran{ \hat \bg(\bs_0) - \bq_a(\bs_0)} }_2^2  \nonumber \\
 &  & + \frac{1}{p} \norm{ \bQ_B \paran{\hat \bZ'_t -\bZ'_t } \bq_a(\bs_0) }_2^2  \nonumber \\
 &  & + \frac{1}{p} \norm{ \paran{ \hat \bQ_B - \bQ_B } \bZ'_t \bq_a(\bs_0) }_2^2.  \nonumber
\end{eqnarray}

Obviously, last three terms are the dominiating terms.

\begin{eqnarray}
\frac{1}{p} \norm{ \bQ_B \bZ'_t \paran{ \hat \bg(\bs_0) - \bq_a(\bs_0)} }_2^2  & \le & \frac{1}{p} \norm{ \bQ_B }_2^2 \norm{\bZ'_t}_2^2 \norm{ \hat \bg(\bs_0) - \bq_a(\bs_0) }_2^2  \nonumber \\
& = & p^{-1} \cdot \Op{n p^{1-\gamma}} \cdot \Op{J_n^{-2\kappa} n^{-1} + n^{-1} p^{\gamma} \Delta_{npT}}  \nonumber \\
& = & \Op{ J_n^{-2\kappa}  p^{-\gamma} + \Delta_{npT}}. \nonumber
\end{eqnarray}

\begin{eqnarray}
\frac{1}{p} \norm{ \bQ_B \paran{\hat \bZ'_t -\bZ'_t } \bq_a(\bs_0) }_2^2 & \le & \frac{1}{p} \norm{ \bQ_B }_2^2 \norm{ \hat \bZ'_t -\bZ'_t }_2^2 \norm{ \bq_a(\bs_0) }_2^2 \nonumber \\
& = & \frac{1}{p} \cdot \Op{n p \Delta_{npT} + n p^{1+\gamma} \Delta_{npT}^2} \cdot \Op{n^{-1}} \nonumber \\
& = & \Op{ \Delta_{npT} +  p^{\gamma} \Delta_{npT}^2}. \nonumber
\end{eqnarray}

\begin{eqnarray}
\frac{1}{p} \norm{ \paran{ \hat \bQ_B - \bQ_B } \bZ'_t \bq_a(\bs_0) }_2^2  & = & \frac{1}{p}  \norm{ \paran{ \hat \bQ_B - \bQ_B } }_2^2 \norm{ \bZ'_t }_2^2 \norm{ \bq_a(\bs_0) }_2^2  \nonumber  \\
& = & \frac{1}{p} O_p(  p^{2\gamma} T^{-1} )  \cdot \Op{n p^{1-\gamma}} \cdot \Op{n^{-1}}  \nonumber \\
& = & \Op{  p^{\gamma} T^{-1} }
\end{eqnarray}

Thus, we have

\begin{eqnarray}
\frac{1}{p} \norm{\hat \bxi_t(\bs_0) - \bxi_t(\bs_0)}_2^2 & = &  \Op{ J_n^{-2\kappa}  p^{-\gamma} + \Delta_{npT} +  p^{\gamma} \Delta_{npT}^2 +  p^{\gamma} T^{-1} }. \nonumber
\end{eqnarray}

\end{proof}

\section{Tables and Plots}  \label{appendix:tableplots}

\begin{table}[htpb!]
\centering
\caption{Mean and standard deviations (in parentheses) of the estimated accuracy measured by $\calD(\hat{\cdot}, \cdot)$ for spatial and variable loading matrices. All numbers in the table are 10 times the true numbers for clear representation. The results are based on 200 simulations.}
\label{table:spdist_msd_table}
\resizebox{\textwidth}{!}{%
\begin{tabular}{ccc|ccccc|ccccl}
\hline
\multicolumn{3}{c|}{} & \multicolumn{5}{c|}{$\gamma = 0$} & \multicolumn{5}{c}{$\gamma=0.5$} \\ \hline
T & p & n & $\calD(\hat{\bA}_1, \bA_1)$ & $\calD(\hat{\bA}_2, \bA_2)$ & Average & $\calD(\hat{\bA}, \bA)$ & $\calD(\hat{\bB}, \bB)$ & $\calD(\hat{\bA}_1, \bA_1)$ & $\calD(\hat{\bA}_2, \bA_2)$ & Average & $\calD(\hat{\bA}, \bA)$ & $\calD(\hat{\bB}, \bB)$ \\ \hline
60 & 10 & 50 & 0.68(0.1) & 0.67(0.1) & 0.68(0.08) & 0.67(0.07) & 0.53(0.11) & 1.27(0.19) & 1.25(0.21) & 1.26(0.16) & 1.25(0.15) & 0.69(0.14) \\
120 & 10 & 50 & 0.45(0.06) & 0.46(0.06) & 0.45(0.05) & 0.45(0.04) & 0.5(0.12) & 0.83(0.12) & 0.84(0.12) & 0.84(0.09) & 0.84(0.08) & 0.63(0.13) \\
240 & 10 & 50 & 0.31(0.04) & 0.31(0.04) & 0.31(0.03) & 0.31(0.02) & 0.49(0.11) & 0.57(0.07) & 0.57(0.08) & 0.57(0.05) & 0.57(0.04) & 0.6(0.13) \\ \hdashline
60 & 20 & 50 & 0.5(0.07) & 0.5(0.09) & 0.5(0.06) & 0.5(0.06) & 0.52(0.08) & 1.18(0.21) & 1.18(0.24) & 1.18(0.17) & 1.17(0.15) & 0.69(0.1) \\
120 & 20 & 50 & 0.34(0.05) & 0.34(0.05) & 0.34(0.03) & 0.34(0.03) & 0.5(0.07) & 0.79(0.12) & 0.79(0.12) & 0.79(0.09) & 0.78(0.08) & 0.6(0.08) \\
240 & 20 & 50 & 0.23(0.03) & 0.23(0.03) & 0.23(0.02) & 0.23(0.02) & 0.47(0.06) & 0.52(0.07) & 0.52(0.07) & 0.52(0.05) & 0.52(0.05) & 0.54(0.06) \\ \hdashline
60 & 40 & 50 & 0.32(0.06) & 0.32(0.05) & 0.32(0.04) & 0.32(0.04) & 0.49(0.07) & 0.98(0.21) & 0.95(0.19) & 0.96(0.15) & 0.95(0.13) & 0.67(0.07) \\
120 & 40 & 50 & 0.21(0.03) & 0.21(0.03) & 0.21(0.02) & 0.21(0.02) & 0.48(0.05) & 0.63(0.1) & 0.62(0.1) & 0.63(0.08) & 0.62(0.07) & 0.58(0.06) \\
240 & 40 & 50 & 0.15(0.02) & 0.14(0.02) & 0.14(0.01) & 0.14(0.01) & 0.46(0.05) & 0.42(0.06) & 0.41(0.06) & 0.41(0.04) & 0.41(0.03) & 0.53(0.06) \\ \hline
60 & 10 & 100 & 0.63(0.06) & 0.63(0.07) & 0.63(0.05) & 0.63(0.05) & 0.36(0.07) & 1.13(0.12) & 1.13(0.13) & 1.13(0.1) & 1.13(0.09) & 0.48(0.09) \\
120 & 10 & 100 & 0.43(0.04) & 0.43(0.04) & 0.43(0.03) & 0.43(0.03) & 0.35(0.07) & 0.77(0.08) & 0.77(0.07) & 0.77(0.05) & 0.77(0.05) & 0.44(0.08) \\
240 & 10 & 100 & 0.3(0.03) & 0.3(0.03) & 0.3(0.02) & 0.3(0.02) & 0.34(0.07) & 0.54(0.05) & 0.53(0.05) & 0.54(0.03) & 0.54(0.03) & 0.41(0.08) \\ \hdashline
60 & 20 & 100 & 0.47(0.05) & 0.47(0.05) & 0.47(0.04) & 0.47(0.04) & 0.35(0.05) & 1.01(0.11) & 1.02(0.11) & 1.01(0.08) & 1.01(0.08) & 0.47(0.06) \\
120 & 20 & 100 & 0.32(0.03) & 0.32(0.03) & 0.32(0.02) & 0.32(0.02) & 0.34(0.05) & 0.68(0.07) & 0.68(0.07) & 0.68(0.05) & 0.68(0.05) & 0.41(0.05) \\
240 & 20 & 100 & 0.22(0.02) & 0.22(0.02) & 0.22(0.01) & 0.22(0.01) & 0.32(0.05) & 0.47(0.04) & 0.47(0.04) & 0.47(0.03) & 0.47(0.03) & 0.37(0.05) \\ \hdashline
60 & 40 & 100 & 0.29(0.03) & 0.29(0.03) & 0.29(0.02) & 0.29(0.02) & 0.34(0.04) & 0.77(0.1) & 0.77(0.1) & 0.77(0.07) & 0.77(0.07) & 0.47(0.04) \\
120 & 40 & 100 & 0.2(0.02) & 0.2(0.02) & 0.2(0.01) & 0.2(0.01) & 0.32(0.04) & 0.52(0.05) & 0.51(0.05) & 0.52(0.04) & 0.52(0.04) & 0.4(0.04) \\
240 & 40 & 100 & 0.14(0.01) & 0.14(0.01) & 0.14(0.01) & 0.14(0.01) & 0.32(0.03) & 0.35(0.03) & 0.36(0.03) & 0.35(0.02) & 0.35(0.02) & 0.35(0.04) \\ \hline
60 & 10 & 200 & 0.63(0.05) & 0.62(0.05) & 0.63(0.04) & 0.63(0.04) & 0.26(0.06) & 1.11(0.08) & 1.1(0.08) & 1.1(0.07) & 1.1(0.07) & 0.33(0.07) \\
120 & 10 & 200 & 0.43(0.03) & 0.43(0.03) & 0.43(0.02) & 0.43(0.02) & 0.25(0.05) & 0.77(0.05) & 0.76(0.05) & 0.77(0.04) & 0.77(0.04) & 0.31(0.06) \\
240 & 10 & 200 & 0.3(0.02) & 0.3(0.02) & 0.3(0.01) & 0.3(0.01) & 0.24(0.05) & 0.54(0.03) & 0.54(0.03) & 0.54(0.02) & 0.54(0.02) & 0.29(0.06) \\ \hdashline
60 & 20 & 200 & 0.47(0.04) & 0.47(0.04) & 0.47(0.03) & 0.47(0.03) & 0.25(0.03) & 0.99(0.07) & 0.98(0.07) & 0.98(0.06) & 0.98(0.06) & 0.34(0.05) \\
120 & 20 & 200 & 0.32(0.02) & 0.32(0.02) & 0.32(0.02) & 0.32(0.02) & 0.24(0.04) & 0.68(0.05) & 0.67(0.04) & 0.67(0.04) & 0.67(0.03) & 0.29(0.04) \\
240 & 20 & 200 & 0.22(0.01) & 0.22(0.01) & 0.22(0.01) & 0.22(0.01) & 0.23(0.03) & 0.47(0.03) & 0.47(0.03) & 0.47(0.02) & 0.47(0.02) & 0.26(0.04) \\ \hdashline
60 & 40 & 200 & 0.29(0.03) & 0.29(0.02) & 0.29(0.02) & 0.29(0.02) & 0.24(0.03) & 0.73(0.06) & 0.73(0.05) & 0.73(0.05) & 0.73(0.05) & 0.33(0.04) \\
120 & 40 & 200 & 0.2(0.01) & 0.2(0.01) & 0.2(0.01) & 0.2(0.01) & 0.23(0.02) & 0.5(0.03) & 0.5(0.03) & 0.5(0.03) & 0.5(0.03) & 0.28(0.03) \\
240 & 40 & 200 & 0.14(0.01) & 0.14(0.01) & 0.14(0.01) & 0.14(0.01) & 0.22(0.02) & 0.35(0.02) & 0.35(0.02) & 0.35(0.01) & 0.35(0.01) & 0.25(0.03) \\ \hline
60 & 10 & 400 & 0.61(0.04) & 0.61(0.04) & 0.61(0.04) & 0.61(0.04) & 0.18(0.04) & 1.08(0.07) & 1.08(0.07) & 1.08(0.06) & 1.08(0.06) & 0.24(0.05) \\
120 & 10 & 400 & 0.42(0.02) & 0.42(0.02) & 0.42(0.02) & 0.42(0.02) & 0.17(0.04) & 0.75(0.04) & 0.75(0.04) & 0.75(0.03) & 0.75(0.03) & 0.22(0.05) \\
240 & 10 & 400 & 0.3(0.01) & 0.3(0.01) & 0.3(0.01) & 0.3(0.01) & 0.17(0.04) & 0.52(0.02) & 0.53(0.02) & 0.53(0.02) & 0.53(0.02) & 0.2(0.04) \\ \hdashline
60 & 20 & 400 & 0.46(0.03) & 0.46(0.03) & 0.46(0.03) & 0.46(0.03) & 0.18(0.03) & 0.95(0.05) & 0.95(0.06) & 0.95(0.05) & 0.95(0.05) & 0.24(0.04) \\
120 & 20 & 400 & 0.31(0.02) & 0.31(0.02) & 0.31(0.01) & 0.31(0.01) & 0.17(0.02) & 0.65(0.04) & 0.65(0.03) & 0.65(0.03) & 0.65(0.03) & 0.2(0.03) \\
240 & 20 & 400 & 0.22(0.01) & 0.22(0.01) & 0.22(0.01) & 0.22(0.01) & 0.16(0.02) & 0.46(0.02) & 0.46(0.02) & 0.46(0.01) & 0.46(0.01) & 0.18(0.03) \\ \hdashline
60 & 40 & 400 & 0.29(0.02) & 0.29(0.02) & 0.29(0.02) & 0.29(0.02) & 0.17(0.02) & 0.7(0.04) & 0.7(0.05) & 0.7(0.04) & 0.7(0.04) & 0.24(0.02) \\
120 & 40 & 400 & 0.19(0.01) & 0.19(0.01) & 0.19(0.01) & 0.19(0.01) & 0.16(0.02) & 0.49(0.02) & 0.48(0.02) & 0.48(0.02) & 0.48(0.02) & 0.2(0.02) \\
240 & 40 & 400 & 0.13(0.01) & 0.13(0.01) & 0.13(0) & 0.13(0) & 0.16(0.02) & 0.34(0.02) & 0.34(0.01) & 0.34(0.01) & 0.34(0.01) & 0.18(0.02) \\ \hline
\end{tabular}%
}
\end{table}

\begin{table}[htpb!]
\centering
\caption{Mean and standard deviations (in parentheses) of the mean squared prediction errors (MSPE).}
\label{table:STprediction}
\resizebox{\textwidth}{!}{%
\begin{tabular}{ccc|c|cc|cc}
\hline
 &  &  & Spatial & \multicolumn{2}{c|}{Temporal MAR(1)} & \multicolumn{2}{c}{Temporal VAR(1)} \\ \hline
T & p & n & $MSPE(\hat{\by}_t(\bs_0)))$ & $MSPE(\hat{\by}_{t+1}(\bs)))$ & $MSPE(\hat{\by}_{t+2}(\bs)))$ & $MSPE(\hat{\by}_{t+1}(\bs)))$ & $MSPE(\hat{\by}_{t+2}(\bs)))$ \\ \hline
60 & 10 & 50 & 0.486(0.089) & 1.716(1.064) & 1.823(1.201) & 1.825(1.075) & 2.019(1.257) \\
120 & 10 & 50 & 0.471(0.06) & 1.658(1.121) & 1.634(1.116) & 1.705(1.133) & 1.732(1.144) \\
240 & 10 & 50 & 0.47(0.041) & 1.78(1.079) & 1.588(1.244) & 1.802(1.076) & 1.624(1.229) \\ \hdashline
60 & 20 & 50 & 0.424(0.069) & 1.592(1.004) & 1.657(1.033) & 1.69(1.032) & 1.819(1.061) \\
120 & 20 & 50 & 0.424(0.048) & 1.535(0.972) & 1.547(1.111) & 1.575(0.983) & 1.634(1.128) \\
240 & 20 & 50 & 0.419(0.036) & 1.619(0.985) & 1.426(1.05) & 1.64(0.988) & 1.463(1.047) \\ \hdashline
60 & 40 & 50 & 0.537(0.085) & 2.001(1.237) & 2.101(1.353) & 2.13(1.276) & 2.308(1.39) \\
120 & 40 & 50 & 0.534(0.055) & 2.006(1.345) & 1.94(1.286) & 2.065(1.36) & 2.051(1.296) \\
240 & 40 & 50 & 0.53(0.037) & 2.141(1.434) & 1.834(1.237) & 2.162(1.432) & 1.877(1.23) \\ \hline
60 & 10 & 100 & 0.067(0.009) & 1.597(0.966) & 1.647(1.006) & 1.685(0.969) & 1.82(1.03) \\
120 & 10 & 100 & 0.066(0.006) & 1.564(0.984) & 1.502(0.95) & 1.608(0.997) & 1.593(0.973) \\
240 & 10 & 100 & 0.065(0.004) & 1.631(0.92) & 1.476(1.02) & 1.65(0.915) & 1.514(1.015) \\ \hdashline
60 & 20 & 100 & 0.058(0.008) & 1.466(0.876) & 1.508(0.901) & 1.557(0.891) & 1.663(0.926) \\
120 & 20 & 100 & 0.058(0.005) & 1.45(0.883) & 1.403(0.915) & 1.489(0.891) & 1.478(0.922) \\
240 & 20 & 100 & 0.058(0.004) & 1.491(0.856) & 1.317(0.864) & 1.51(0.854) & 1.353(0.859) \\ \hdashline
60 & 40 & 100 & 0.072(0.01) & 1.845(1.075) & 1.893(1.105) & 1.975(1.113) & 2.085(1.126) \\
120 & 40 & 100 & 0.072(0.006) & 1.889(1.229) & 1.765(1.076) & 1.939(1.247) & 1.859(1.077) \\
240 & 40 & 100 & 0.072(0.005) & 1.961(1.223) & 1.707(1.074) & 1.984(1.22) & 1.754(1.068) \\ \hline
60 & 10 & 200 & 0.015(0.002) & 1.542(0.922) & 1.597(0.972) & 1.629(0.921) & 1.766(1) \\
120 & 10 & 200 & 0.015(0.001) & 1.515(0.976) & 1.454(0.913) & 1.557(0.982) & 1.538(0.934) \\
240 & 10 & 200 & 0.015(0.001) & 1.599(0.915) & 1.42(0.988) & 1.619(0.912) & 1.458(0.988) \\ \hdashline
60 & 20 & 200 & 0.013(0.002) & 1.419(0.86) & 1.461(0.88) & 1.51(0.88) & 1.61(0.897) \\
120 & 20 & 200 & 0.013(0.001) & 1.401(0.853) & 1.358(0.88) & 1.44(0.861) & 1.429(0.883) \\
240 & 20 & 200 & 0.013(0.001) & 1.464(0.859) & 1.276(0.84) & 1.481(0.86) & 1.308(0.838) \\ \hdashline
60 & 40 & 200 & 0.015(0.002) & 1.786(1.04) & 1.836(1.099) & 1.906(1.066) & 2.02(1.122) \\
120 & 40 & 200 & 0.015(0.001) & 1.828(1.211) & 1.714(1.042) & 1.875(1.22) & 1.808(1.049) \\
240 & 40 & 200 & 0.015(0.001) & 1.92(1.214) & 1.652(1.031) & 1.941(1.213) & 1.698(1.027) \\ \hline
60 & 10 & 400 & 0.014(0.002) & 1.63(0.965) & 1.714(1.033) & 1.727(0.965) & 1.893(1.059) \\
120 & 10 & 400 & 0.014(0.001) & 1.63(1.058) & 1.556(0.975) & 1.676(1.069) & 1.647(1.009) \\
240 & 10 & 400 & 0.014(0.001) & 1.711(0.985) & 1.527(1.077) & 1.728(0.983) & 1.568(1.075) \\ \hdashline
60 & 20 & 400 & 0.012(0.002) & 1.511(0.914) & 1.561(0.926) & 1.611(0.936) & 1.719(0.949) \\
120 & 20 & 400 & 0.012(0.001) & 1.502(0.923) & 1.452(0.934) & 1.543(0.931) & 1.534(0.945) \\
240 & 20 & 400 & 0.012(0.001) & 1.569(0.929) & 1.373(0.915) & 1.589(0.931) & 1.407(0.912) \\ \hdashline
60 & 40 & 400 & 0.015(0.002) & 1.907(1.108) & 1.964(1.166) & 2.033(1.14) & 2.159(1.181) \\
120 & 40 & 400 & 0.015(0.001) & 1.967(1.319) & 1.831(1.107) & 2.021(1.334) & 1.937(1.117) \\
240 & 40 & 400 & 0.015(0.001) & 2.062(1.314) & 1.775(1.118) & 2.086(1.31) & 1.823(1.111) \\ \hline
\end{tabular}%
}
\end{table}
\end{appendices}

\end{document}